\documentclass[11pt,reqno]{amsart}

\usepackage{amsfonts,amssymb,amsmath,gensymb,epic,eepic,float,fullpage}
\usepackage[usenames,dvipsnames]{color}
\usepackage{caption}
\usepackage{siunitx}
\usepackage{empheq}
\usepackage[mathscr]{euscript}
\usepackage[T1]{fontenc}
\usepackage{mathtools}
\usepackage{comment}
\usepackage{bm}
\usepackage{caption}
\usepackage{hyperref}
\usepackage{enumerate}
\usepackage[shortlabels]{enumitem}
\usepackage{tikz}
\usetikzlibrary{graphs,patterns,decorations.markings,arrows,matrix}
\usetikzlibrary{calc,decorations.pathmorphing,decorations.pathreplacing,shapes}
\allowdisplaybreaks

\newtheorem{thm}{Theorem}[section]

\newtheorem{prop}[thm]{Proposition}

\newtheorem{cor}[thm]{Corollary}

\theoremstyle{definition}
\newtheorem{defn}[thm]{Definition}

\theoremstyle{remark}
\newtheorem{rmk}[thm]{Remark}

\theoremstyle{remark}
\newtheorem{ex}[thm]{Example}

\numberwithin{equation}{section}

\def\wh{\widehat}

\def\leq{\leqslant}
\def\geq{\geqslant}

\def\s{\mathfrak{s}}

\def\ie{{\it i.e.}\/,}
\def\cf{{\it cf.}\/}

\colorlet{lgray}{white!85!black}
\colorlet{lred}{white!85!red}
\colorlet{lgreen}{white!80!green}
\colorlet{dgreen}{black!30!green}
\definecolor{green}{rgb}{0.1,0.8,0.1}
\definecolor{yellow}{rgb}{1.0,0.85,0.25}

\newcommand{\bra}[1]{\left\langle #1\right|}
\newcommand{\ket}[1]{\left|#1\right\rangle}

\newcommand\fs{\footnotesize}

\renewcommand{\tikz}[2]{
\begin{tikzpicture}[scale=#1,baseline=(current bounding box.center),>=stealth]
#2
\end{tikzpicture}}

\def\C{\mathcal{C}}

\def\I{\bm{I}}
\def\J{\bm{J}}
\def\K{\bm{K}}

\def\AA{\bm{A}}

\def\M{\bm{M}}

\newcommand{\Is}[2]{\I_{[#1,#2]}}

\makeatletter
\newcommand{\subalign}[1]{%
  \vcenter{%
    \Let@ \restore@math@cr \default@tag
    \baselineskip\fontdimen10 \scriptfont\tw@
    \advance\baselineskip\fontdimen12 \scriptfont\tw@
    \lineskip\thr@@\fontdimen8 \scriptfont\thr@@
    \lineskiplimit\lineskip
    \ialign{\hfil$\m@th\scriptstyle##$&$\m@th\scriptstyle{}##$\crcr
      #1\crcr
    }%
  }
}
\makeatother

\newlength{\cellsize}
\newcommand{\tableaucell}[1]{{%
\def \arg{#1}\def \void{}%
\ifx \void \arg
\vbox to \cellsize{\vfil \hrule width \cellsize height 0pt}%
\else
\unitlength=\cellsize
\begin{picture}(1,1)
\put(0,0){\makebox(1,1){$#1$}}
\put(0,0){\line(1,0){1}}
\put(0,1){\line(1,0){1}}
\put(0,0){\line(0,1){1}}
\put(1,0){\line(0,1){1}}
\end{picture}%
\fi}}

\newcommand\tableau[1]{
\vcenter{
\let\\=\cr
\baselineskip=-16000pt
\lineskiplimit=16000pt
\lineskip=0pt
\halign{&\tableaucell{##}\cr#1\crcr}}}

\begin{document}

\title{Nonsymmetric Macdonald polynomials \\ via integrable vertex models}

\author{Alexei Borodin}

\address[Alexei Borodin]{ Department of Mathematics, MIT, Cambridge, USA, and
Institute for Information Transmission Problems, Moscow, Russia. E-mail: borodin@math.mit.edu }

\author{Michael Wheeler}

\address[Michael Wheeler]{ School of Mathematics and Statistics, The University of Melbourne, Parkville,
Victoria, Australia. E-mail: wheelerm@unimelb.edu.au}

\begin{abstract}
Starting from an integrable rank-$n$ vertex model, we construct an explicit family of partition functions indexed by compositions 
$\mu = (\mu_1,\dots,\mu_n)$. Using the Yang--Baxter algebra of the model and a certain rotation operation that acts on our partition functions, we show that they are eigenfunctions of the Cherednik--Dunkl operators $Y_i$ for all $1 \leq i \leq n$, and are thus equal to nonsymmetric Macdonald polynomials $E_{\mu}$. Our partition functions have the combinatorial interpretation of ensembles of coloured lattice paths which traverse a cylinder. Applying a simple bijection to such path ensembles, we show how to recover the well-known combinatorial formula for $E_{\mu}$ due to Haglund--Haiman--Loehr.
\end{abstract}

\maketitle

\setcounter{tocdepth}{1}
\makeatletter
\def\l@subsection{\@tocline{2}{0pt}{2.5pc}{5pc}{}}
\makeatother
\tableofcontents

\section{Introduction}

\subsection{Nonsymmetric Macdonald polynomials}

The {\it nonsymmetric Macdonald polynomials} were introduced in a series of papers by Cherednik \cite{Cherednik,Cherednik2}, Macdonald \cite{Macdonald-Bourbaki} and Opdam \cite{Opdam}. They depend on an alphabet $(x_1,\dots,x_n)$, two supplementary parameters $(q,t)$, and are indexed by nonnegative integer compositions $\mu = (\mu_1,\dots,\mu_n)$; they are traditionally denoted $E_{\mu}(x_1,\dots,x_n;q,t)$. In this work we employ the notation
\begin{align}
f_{(\mu_1,\dots,\mu_n)}(x_1,\dots,x_n;q,t) 
:= 
E_{(\mu_n,\dots,\mu_1)}(x_n,\dots,x_1;q,t)
\end{align}
and perform all calculations in terms of the $f_{\mu}$ family; we do this for notational consistency with our previous paper \cite{BorodinW2}. Up to overall normalization, the nonsymmetric Macdonald polynomials are characterized as the unique polynomial eigenfunctions of the Cherednik--Dunkl operators:
\begin{align}
\label{eigen-intro}
Y_i \cdot f_{\mu}(x_1,\dots,x_n;q,t)
=
y_i(\mu;q,t)
f_{\mu}(x_1,\dots,x_n;q,t),
\qquad
1 \leq i \leq n,
\end{align}
where $Y_i = T_{i-1} \cdots T_{1} \cdot \omega \cdot T^{-1}_{n-1} \cdots T^{-1}_{i}$ is a product of generators of the affine Hecke algebra of type $A_{n-1}$, and the eigenvalue $y_i(\mu;q,t)$ is an explicit integral power of $q$ times an explicit integral power of $t$. For precise definitions of these objects, see equations \eqref{Hecke}, \eqref{omega} and \eqref{rho} of the text.

Aside from the fact that their symmetrization yields symmetric Macdonald polynomials \cite{Macdonald-Bourbaki,Macdonald,Macdonald2}, nonsymmetric Macdonald polynomials have a host of remarkable properties including orthogonality and Cauchy-type summation identities \cite{MimachiN}, and a recursive construction via polynomial representations of the Hecke algebra \cite{Sahi}. For a general survey of their theory, we refer the reader to \cite{Marshall}. 

The goal of the present work is to study the defining relation of the nonsymmetric Macdonald polynomials, \eqref{eigen-intro}, from the perspective of integrable vertex models. Our motivations for doing so are threefold: 

{\bf 1.} An explicit vertex model solution of \eqref{eigen-intro}, in the case when $\mu$ is an {\it anti-dominant} composition\footnote{A weakly increasing composition.}, has been known since the work of \cite{CantiniGW}. Quite recently, the combinatorics underlying the matrix product solution of \cite{CantiniGW} was explored in \cite{CorteelMW}; 

{\bf 2.} A vertex model representation of the nonsymmetric Hall--Littlewood polynomials (which are the $q=0$ degeneration of $f_{\mu}$) appeared in our recent work \cite{BorodinW2}. The present paper thus provides a unification of this special case and that of {\bf 1}; 

{\bf 3.} In the recent integrable probability literature there has been a very rich interplay between (non)symmetric multivariate functions on the one hand, and stochastic (higher-rank) vertex models on the other; see, for example, \cite{Borodin,CorwinP,BorodinP1,Borodin2,BorodinBW,BarraquandBCW,Borodin3,Aggarwal,BorodinW2}. Against such a backdrop, it is natural to hope that our vertex model construction of the nonsymmetric Macdonald polynomials might ultimately lead to interesting probabilistic applications.

Let us now outline some of the particulars of our construction.

\subsection{Higher-rank vertex models and row operators}
\label{intro-models}

The model that we study is based on vertices of the form
\begin{align}
\label{generic-L-intro}
L_x(\I,j; \K,\ell)
\index{L1@$L_x(\I,j; \K,\ell)$; vertex weights}
=
\tikz{0.6}{
\draw[lgray,line width=1.5pt] (-1,-1) -- (1,-1) -- (1,1) -- (-1,1) -- (-1,-1);
\node[left] at (-1,0) {\tiny $j$};\node[right] at (1,0) {\tiny $\ell$};
\node[below] at (0,-1) {\tiny $\I$};\node[above] at (0,1) {\tiny $\K$};
\node[text centered] at (0,0) {$x$};
},
\end{align}
where the labels of the vertical edges are nonnegative integers 
$j,\ell \in \{0,1,\dots,n\}$, while the labels of the horizontal edges are compositions $\I,\K \in (\mathbb{Z}_{\geq 0})^n$. For any definite choice of $\I,j,\K,\ell$, the picture on the right-hand side of \eqref{generic-L-intro} is assigned a corresponding {\it weight} $L_x(\I,j; \K,\ell)$, which is a polynomial in the variable $x$ and the Macdonald parameter $t$; for the precise form of these weights, see equation \eqref{bos-weights} in the text. 

For now it suffices to note that the model that we are using is a rather special solution of the Yang--Baxter equation associated to the quantized affine algebra $U_t(\widehat{\mathfrak{sl}_{n+1}})$; this model (and its $n=1$ reduction) has already appeared in a variety of settings related to Macdonald polynomials and their descendants \cite{Tsilevich,Korff,WheelerZ,CantiniGW,GarbaliGW,BorodinW-spin}.

The object \eqref{generic-L-intro} can be considered as (the matrix elements of) a linear operator acting on the formal vector space $V = {\rm Span}_{\mathbb{C}} (\mathbb{Z}_{\geq 0})^n$; in the language of integrable systems, it is referred to as a {\it Lax matrix}, and the underlying vector space $V$ as the {\it physical} or {\it quantum space}.
\begin{figure}
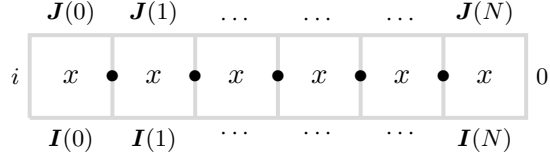

\begin{align*}
\tikz{1.1}{
\draw[lgray,line width=1.5pt] (0.5,-0.5) -- (6.5,-0.5) -- (6.5,0.5) -- (0.5,0.5) -- (0.5,-0.5);
\foreach\x in {1,...,5}{
\draw[lgray,line width=1.5pt] (0.5+\x,-0.5) -- (0.5+\x,0.5);
}
\node[text centered] at (1,0) {$x$};
\node[text centered] at (2,0) {$x$};
\node[text centered] at (3,0) {$x$};
\node[text centered] at (4,0) {$x$};
\node[text centered] at (5,0) {$x$};
\node[text centered] at (6,0) {$x$};
\node[left] at (0.5,0) {\fs $i$};\node[right] at (6.5,0) {\fs $0$};
\foreach\x in {1.5,...,5.5}{
\node at (\x,0) {$\bullet$};
}
\node[below] at (6,-0.5) {\fs $\I(N)$};\node[above] at (6,0.5) {\fs $\J(N)$};
\node[below] at (5,-0.5) {\fs $\cdots$};\node[above] at (5,0.5) {\fs $\cdots$};
\node[below] at (4,-0.5) {\fs $\cdots$};\node[above] at (4,0.5) {\fs $\cdots$};
\node[below] at (3,-0.5) {\fs $\cdots$};\node[above] at (3,0.5) {\fs $\cdots$};
\node[below] at (2,-0.5) {\fs $\I(1)$};\node[above] at (2,0.5) {\fs $\J(1)$};
\node[below] at (1,-0.5) {\fs $\I(0)$};\node[above] at (1,0.5) {\fs $\J(0)$};
}
\end{align*}
\caption{The matrix elements of the row operators $\C_i(x)$. Each of the $N+1$ squares is a vertex of the form \eqref{generic-L-intro}, and all vertices depend on the same parameter $x$. Each $\I(k)$ and $\J(k)$ denotes a fixed composition in $(\mathbb{Z}_{\geq 0})^n$. The value assigned to the leftmost vertical edge is $i \geq 1$, that assigned to the rightmost vertical edge is $0$, and the value assigned to each internal vertical edge (indicated by $\bullet$) is uniquely determined.}
\label{rows-intro}
\end{figure}
By horizontally concatenating several copies of the vertex \eqref{generic-L-intro}, we obtain a family of linear operators $\C_i(x)$, $1 \leq i \leq n$, which we call {\it row operators}; the details of their construction are given in Figure \ref{rows-intro} and in equation \eqref{C-row} of the text. These operators are well-known in the {\it algebraic Bethe Ansatz} approach to integrable systems; they are nothing but elements of the {\it monodromy matrix} built by taking $(N+1)$-fold tensor products of the Lax matrix \eqref{generic-L-intro}, where the value of $N \geq 0$ is arbitrary. 

By virtue of their construction and the Yang--Baxter equation of the model \eqref{generic-L-intro}, the row operators satisfy the following exchange relations:
\begin{align}
\label{CC>-intro}
\C_i(x) \C_j(y) = 
\frac{x-ty}{x-y}\, \C_j(y) \C_i(x) 
- 
\frac{(1-t)y}{x-y}\, \C_j(x) \C_i(y),
\end{align}
valid for all indices $n \geq i>j \geq 1$ and arbitrary complex parameters $x$, $y$. In fact, the operators $\C_i(x)$ are part of a richer structure known as the {\it Yang--Baxter algebra}, of which \eqref{CC>-intro} is but one commutation relation; for our purposes, we will only need to make use of this identity and none of the other relations within the algebra.

\subsection{Cylindrical construction}

Since the matrix elements of each $\C_i(x)$ operator have polynomial dependence on the parameter $x$, by multiplying $n$ such operators and taking matrix elements of the resulting product, we obtain multivariate polynomials. Many such constructions are possible, some leading to symmetric polynomials \cite{WheelerZ,CantiniGW}, and others to nonsymmetric families \cite{BorodinW2}.

In this work we will study one particular construction. Let $\mu = (\mu_1,\dots,\mu_n)$ be a composition and $\rho = (\rho_1,\dots,\rho_n)$ a permutation of $(1,\dots,n)$. We define the family of polynomials
\begin{align}
\label{f-rho-intro}
f_{\mu}^{\rho}(x_1,\dots,x_n;q,t)
=
\Omega_{\mu}(q,t)
\langle \C_{\rho_1}(x_1) \cdots \C_{\rho_n}(x_n) \rangle_{\mu},
\end{align}
where $\langle \cdot \rangle_{\mu}$ denotes a certain linear form on the vector space $\mathbb{V}(N) = \otimes_{i=0}^{N} V_i$ (here $V_i$ denotes a copy of $V$) with $N$ chosen as the largest part of the composition $\mu$, and $\Omega_{\mu}(q,t)$ is an explicit polynomial in the parameters $(q,t)$. The linear form $\langle \cdot \rangle_{\mu}$ is essentially a trace over $\mathbb{V}(N)$, and accordingly, one can think of \eqref{f-rho-intro} as a partition function on a cylinder of length $N+1$ and circumference $n$, with appropriate cuts. A pictorial representation of the right-hand side is given in Figure \ref{lattice-intro}; for the precise definitions of $\langle \cdot \rangle_{\mu}$ and $\Omega_{\mu}(q,t)$, see equations \eqref{form}, \eqref{Omega} and \eqref{alpha} of the text.
\begin{figure}
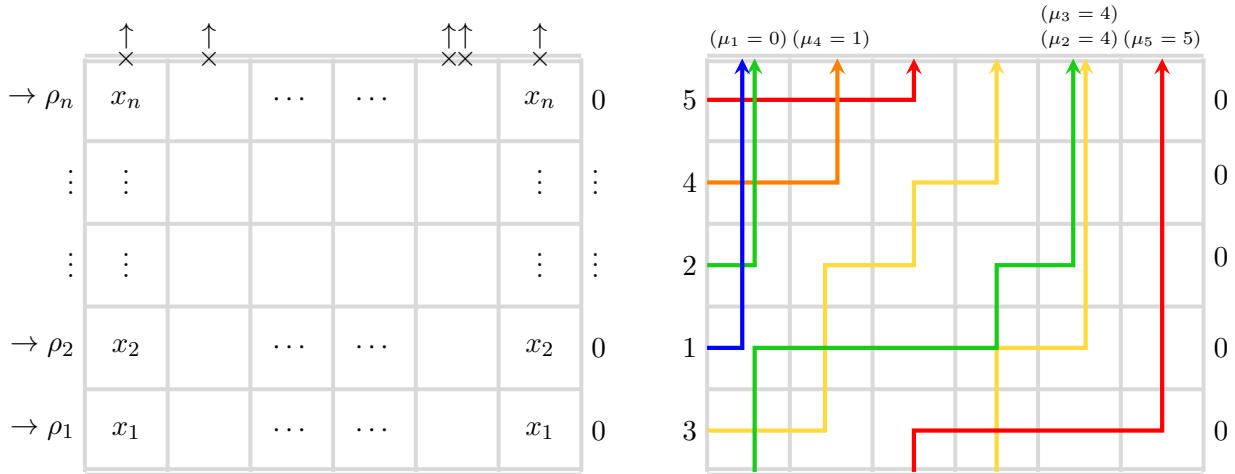

\begin{align*}
\tikz{1.1}{
\foreach\y in {0}{
\draw[lgray,line width=1.5pt,double] (1.5,0.5+\y) -- (7.5,0.5+\y);
}
\foreach\y in {5}{
\draw[lgray,line width=1.5pt,double] (1.5,0.5+\y) -- (7.5,0.5+\y);
}
\foreach\y in {1,...,4}{
\draw[lgray,line width=1.5pt] (1.5,0.5+\y) -- (7.5,0.5+\y);
}
\foreach\x in {0,...,6}{
\draw[lgray,line width=1.5pt] (1.5+\x,0.5) -- (1.5+\x,5.5);
}
\node[text centered] at (2,1) {$x_1$};
\node[text centered] at (4,1) {$\cdots$};
\node[text centered] at (5,1) {$\cdots$};
\node[text centered] at (7,1) {$x_1$};
\node[text centered] at (2,2) {$x_2$};
\node[text centered] at (4,2) {$\cdots$};
\node[text centered] at (5,2) {$\cdots$};
\node[text centered] at (7,2) {$x_2$};
\node[text centered] at (2,5) {$x_n$};
\node[text centered] at (4,5) {$\cdots$};
\node[text centered] at (5,5) {$\cdots$};
\node[text centered] at (7,5) {$x_n$};
\node[text centered] at (2,3.1) {$\vdots$};
\node[text centered] at (2,4.1) {$\vdots$};
\node[text centered] at (7,3.1) {$\vdots$};
\node[text centered] at (7,4.1) {$\vdots$};
\node[above] at (6,5.5) {\phantom{\tiny$(\mu_2=4)$}};
\node[above] at (6,5.8) {\phantom{\tiny$(\mu_3=4)$}};
\node at (2,5.5) {$\times$};
\node at (3,5.5) {$\times$};
\node at (5.9,5.5) {$\times$};
\node at (6.1,5.5) {$\times$};
\node at (7,5.5) {$\times$};
\node[above] at (2,5.5) {$\uparrow$};
\node[above] at (3,5.5) {$\uparrow$};
\node[above] at (5.9,5.5) {$\uparrow$};
\node[above] at (6.1,5.5) {$\uparrow$};
\node[above] at (7,5.5) {$\uparrow$};
\node[right] at (7.5,1) {$0$};
\node[right] at (7.5,2) {$0$};
\node[right] at (7.5,3.1) {$\vdots$};
\node[right] at (7.5,4.1) {$\vdots$};
\node[right] at (7.5,5) {$0$};
\node[left] at (1.5,1) {$\rightarrow \rho_1$};
\node[left] at (1.5,2) {$\rightarrow \rho_2$};
\node[left] at (1.5,3.1) {$\vdots$};
\node[left] at (1.5,4.1) {$\vdots$};
\node[left] at (1.5,5) {$\rightarrow \rho_n$};
}
\qquad
\tikz{1.1}{
\foreach\y in {0}{
\draw[lgray,line width=1.5pt,double] (1.5,0.5+\y) -- (7.5,0.5+\y);
}
\foreach\y in {5}{
\draw[lgray,line width=1.5pt,double] (1.5,0.5+\y) -- (7.5,0.5+\y);
}
\foreach\y in {1,...,4}{
\draw[lgray,line width=1.5pt] (1.5,0.5+\y) -- (7.5,0.5+\y);
}
\foreach\x in {0,...,6}{
\draw[lgray,line width=1.5pt] (1.5+\x,0.5) -- (1.5+\x,5.5);
}
\node[above] at (2,5.5) {\tiny$(\mu_1=0)$};
\node[above] at (3,5.5) {\tiny$(\mu_4=1)$};
\node[above] at (6,5.5) {\tiny$(\mu_2=4)$};
\node[above] at (6,5.8) {\tiny$(\mu_3=4)$};
\node[above] at (7,5.5) {\tiny$(\mu_5=5)$};
\node[right] at (7.5,1) {$0$};
\node[right] at (7.5,2) {$0$};
\node[right] at (7.5,3.1) {$0$};
\node[right] at (7.5,4.1) {$0$};
\node[right] at (7.5,5) {$0$};
\node[left] at (1.5,1) {$3$};
\node[left] at (1.5,2) {$1$};
\node[left] at (1.5,3) {$2$};
\node[left] at (1.5,4) {$4$};
\node[left] at (1.5,5) {$5$};
\draw[ultra thick,yellow,->] (1.5,1) -- (2.925,1) -- (2.925,3) -- (4,3) -- (4,4) -- (5,4) -- (5,5.5);
\draw[ultra thick,yellow,->] (5,0.5) -- (5,2) -- (6.075,2) -- (6.075,5.5);
\draw[ultra thick,red,->] (1.5,5) -- (4,5) -- (4,5.5); 
\draw[ultra thick,red,->] (4,0.5)-- (4,1) -- (7,1) -- (7,5.5);
\draw[ultra thick,orange,->] (1.5,4) -- (3.075,4) -- (3.075,5.5);
\draw[ultra thick,green,->] (1.5,3) -- (2.075,3) -- (2.075,5.5); 
\draw[ultra thick,green,->] (2.075,0.5) -- (2.075,2) -- (5,2) -- (5,3) -- (5.925,3) -- (5.925,5.5);
\draw[ultra thick,blue,->] (1.5,2) -- (1.925,2) -- (1.925,5.5);
}
\end{align*}
\caption{Left panel: a representation $\langle \C_{\rho_1}(x_1) \cdots \C_{\rho_n}(x_n) \rangle_{\mu}$ as a partition function on a cylinder. For all $1 \leq i \leq n$, a path of colour $\rho_i$ enters the lattice via the left edge of row $i$, and traverses the lattice by taking up or right steps. The top and bottom edges of each column are considered to be identified, so that the path may wrap around the cylinder. The path exits the lattice via the top of column $\mu_{\rho_i}$ (columns are labelled from left to right, starting at $0$); such exit points can be viewed as ``cuts'' along the cylinder. 
\\
{\color{white} .}
\hspace{0.1em} Right panel: extracting a sample configuration from the partition function $\langle \C_{3}(x_1) \C_{1}(x_2) \C_{2}(x_3) \C_{4}(x_4) \C_{5}(x_5) \rangle_{\mu}$ in the case $\mu = (0,4,4,1,5)$; the positions where the paths exit the cylinder are indicated along the top edge.}
\label{lattice-intro}
\end{figure}

\subsection{Two properties of the construction}

The partition functions \eqref{f-rho-intro} thus constructed have two key properties, which turn out to uniquely characterize them:
\begin{prop}
\label{prop1-intro}
Fix an integer $1 \leq i \leq n-1$ and let $\rho$ be a permutation such that $\rho_i < \rho_{i+1}$. The polynomials \eqref{f-rho-intro} satisfy the exchange relation
\begin{align}
\label{f-rho-exchange}
T^{-1}_i
\cdot
f_{\mu}^{\rho}(x_1,\dots,x_n;q,t)
=
t^{-1}
f_{\mu}^{\mathfrak{s}_i \cdot \rho}(x_1,\dots,x_n;q,t),
\end{align}
where $T^{-1}_i$ is an inverse Hecke generator, given by equation \eqref{Hecke}, and 
\begin{align*}
\mathfrak{s}_i \cdot \rho = (\rho_1,\dots,\rho_{i+1},\rho_i,\dots,\rho_n)
\end{align*}
is the same as the original permutation $\rho$, up to transposition of the $i$-th and $(i+1)$-th elements.
\end{prop}

\begin{prop}
\label{prop2-intro}
For any composition $\mu$ and permutation $\rho = (\rho_1,\dots,\rho_n)$, there holds
\begin{align}
\label{f-rho-cyclic}
f_{\mu}^{\rho}(x_1,\dots,x_{n-1},qx_n;q,t)
=
t^{n-2\rho_n+1}
y_{\rho_n}(\mu;q,t)
f_{\mu}^{\omega \cdot \rho}(x_n,x_1,\dots,x_{n-1};q,t),
\end{align}
where $y_{\rho_n}(\mu;q,t)$ denotes a Cherednik--Dunkl eigenvalue, given by \eqref{rho}, and
\begin{align*}
\omega \cdot \rho
=
(\rho_n,\rho_1,\dots,\rho_{n-1})
\end{align*}
is a cyclic permutation of the original $\rho$. 
\end{prop}

The proof of Proposition \ref{prop1-intro} is straightforward; it follows as an immediate consequence of the commutation relation \eqref{CC>-intro} when applied to the right-hand side of \eqref{f-rho-intro}. Exchange relations of the type \eqref{f-rho-exchange} are a standard and well understood feature of partition functions constructed in the manner \eqref{f-rho-intro}.

On the other hand, the proof of Proposition \ref{prop2-intro} requires substantially more work. Remarkably, we find that \eqref{f-rho-cyclic} can be proved bijectively: \ie\ we find an explicit one-to-one pairing between lattice configurations of $f_{\mu}^{\rho}(x_1,\dots,x_{n-1},qx_n;q,t)$ and those of $f_{\mu}^{\omega \cdot \rho}(x_n,x_1,\dots,x_{n-1};q,t)$, and show that their weights always match up to the proportionality factor $t^{n-2\rho_n+1}
y_{\rho_n}(\mu;q,t)$. This bijection (Proposition \ref{prop:refined-cyclic-rel} of the text) is one of our key results; in general, when two vertex model partition functions coincide, it is very unusual to have a direct combinatorial proof of their equality. 

By combining Propositions \ref{prop1-intro} and \ref{prop2-intro}, we arrive at the main result of this paper:
\begin{thm}
Let $\mu$ be a generic composition. When $\rho = (1,\dots,n) = {\rm id}$, we have 
\begin{align*}
Y_i
\cdot
f^{\rm id}_{\mu}(x_1,\dots,x_n;q,t)
=
y_{i}(\mu;q,t)
f^{\rm id}_{\mu}(x_1,\dots,x_n;q,t),
\qquad
1 \leq i \leq n,
\end{align*}
and accordingly, 
\begin{align}
\label{f-ansatz-intro}
\Omega_{\mu}(q,t)
\langle \C_{1}(x_1) \cdots \C_{n}(x_n) \rangle_{\mu} 
=
f_{\mu}(x_1,\dots,x_n;q,t),
\end{align}
where the right-hand side of \eqref{f-ansatz-intro} is the (properly normalized) nonsymmetric Macdonald polynomial associated to $\mu$. 
\end{thm}

In general, $f^{\rho}_{\mu}(x_1,\dots,x_n;q,t)$ is equal to the {\it permuted basement nonsymmetric Macdonald polynomial} associated to the composition $\mu$ and permutation $\rho$, as defined in \cite{Alexandersson}.

\subsection{Combinatorial formula of Haglund--Haiman--Loehr}

An important development in the nonsymmetric Macdonald theory came in the paper \cite{HaglundHL2}, where the authors gave an explicit combinatorial construction of the polynomials in terms of a certain class of tableaux with associated $(x_1,\dots,x_n;q,t)$ statistics. This formula was later generalized to all affine root systems in \cite{RamY}. 

Another of our key results concerns matching \eqref{f-ansatz-intro} with a combinatorial formula in the spirit of \cite{HaglundHL2}. In Sections \ref{ssec:fill-config}--\ref{ssec:weight-match}, starting from the algebraic expression \eqref{f-ansatz-intro}, we show that it can be rewritten as
\begin{align}
\label{hhl-formula-intro}
f_{\mu}(x_1,\dots,x_n;q,t)
=
\sum_{\sigma \in \mathfrak{S}(\mu)}
x^{\sigma}
t^{\Delta(\sigma)}
\prod_{s \in \mathcal{D}(\sigma)}
\frac{1-t}{1-q^{l(s)+1} t^{a(s)+1}}
\prod_{s \in \mathcal{A}(\sigma)}
\frac{q^{l(s)+1} t^{a(s)} (1-t)}{1-q^{l(s)+1} t^{a(s)+1}},
\end{align}
where the sum is taken over {\it non-attacking fillings} $\sigma$ of the (Young) diagram associated to $\mu$. The definition of fillings and of all statistics appearing in \eqref{hhl-formula-intro} can be found in Sections \ref{ssec:diag}--\ref{ssec:formula}. We have adhered to the same conventions as in \cite{HaglundHL2}, and one can easily see that (up to inversion of the parameters $q$ and $t$) equation \eqref{hhl-formula-intro} differs from Theorem 3.5.1 of \cite{HaglundHL2} only up to the way $t$ exponents are distributed across the summand.

The correspondence between \eqref{f-ansatz-intro} and \eqref{hhl-formula-intro} is far from obvious, and originally came as a surprise to us. While \eqref{f-ansatz-intro} requires a more complicated technical setup than the tableaux combinatorics underpinning \eqref{hhl-formula-intro}, the algebraic expression \eqref{f-ansatz-intro} does offer some structural insights that were not previously available. 

The connection with Yang--Baxter integrability outlined in the present text is just the tip of a rather large iceberg; see also the recent work \cite{GarbaliW}. Let us also mention that our current work is conceptually analogous to the older paper of \cite{EtingofK}, where the Macdonald polynomials were expressed as certain traces of intertwiners of $U_q(\mathfrak{sl}_{n+1})$ modules; a concrete connection between the two works remains unclear to us at this stage.

\subsection{Acknowledgments}

We would like to thank Amol Aggarwal, Alexandr Garbali, Jan de Gier, Arun Ram and Lauren Williams for discussions on topics related to this work.

The work of A.~B. was partially supported by the NSF grants 
DMS-1607901 and DMS-1664619. M.~W. was partially supported by the ARC grant DE160100958.

\section{Nonsymmetric Macdonald polynomials}

\subsection{Cherednik--Dunkl operators and nonsymmetric Macdonald polynomials}

In what follows, we use the same conventions as Mimachi--Noumi \cite{MimachiN} in writing generators of the Hecke algebra (up to multiplying the generators of \cite{MimachiN} by $t^{1/2}$), the Cherednik--Dunkl operators, and in defining the polynomials $E_{\mu}(x_1,\dots,x_n)$. All operators in this section carry a tilde; this is to distinguish them from a modified form of the operators that we write down in Section \ref{ssec:reverse}.

We begin by recalling the definition of the {\it Hecke algebra} of type $A_{n-1}$. It is the algebra generated by the family of generators $\tilde{T}_1,\dots,\tilde{T}_{n-1}$, which satisfy the relations
\begin{align}
\label{hecke1}
(\tilde{T}_i - t)(\tilde{T}_i + 1) = 0,
\quad 
1 \leq i \leq n-1,
\qquad
\tilde{T}_i \tilde{T}_{i+1} \tilde{T}_i = \tilde{T}_{i+1} \tilde{T}_i 
\tilde{T}_{i+1},
\quad
1 \leq i \leq n-2,
\end{align}
as well as the commutativity property
\begin{align}
\label{hecke2}
[\tilde{T}_i,\tilde{T}_j] = 0,
\quad \forall\ i,j\ \text{such that}\ |i-j| > 1.
\end{align}
A well-known realization of this algebra is its {\it polynomial representation}. In this representation, one identifies the abstract generator $\tilde{T}_i$ and its inverse with explicit operators on $\mathbb{C}[x_1,\dots,x_n]$, the ring of polynomials in $n$ variables:
\begin{align}
\label{rev-Hecke}
\tilde{T}_i \mapsto t - \frac{t x_i-x_{i+1}}{x_i-x_{i+1}} (1-\s_i),
\quad
\tilde{T}_i^{-1} \mapsto t^{-1}\left(1 - \frac{t x_i-x_{i+1}}{x_i-x_{i+1}} (1-\s_i) \right),
\quad
1 \leq i \leq n-1,
\end{align}
where $\mathfrak{s}_i$ denotes the transposition operator on neighbouring variables, namely
\begin{align*}
\mathfrak{s}_i \cdot h(x_1,\dots,x_n)
=
h(x_1,\dots,x_{i+1},x_i,\dots,x_n),
\quad
1 \leq i \leq n-1,
\end{align*}
for any polynomial $h \in \mathbb{C}[x_1,\dots,x_n]$.

Let us introduce a further operator $\tilde\omega$, defined as follows:
\begin{align}
\label{rev-omega}
\tilde\omega = \s_{n-1} \cdots \s_1 \tau_1,
\end{align}
where $\tau_1$ denotes a $q$-shift operator with action $\tau_1 \cdot h(x_1,\dots,x_n) = h(q x_1, x_2,\dots,x_n)$. One has
\begin{align*}
\tilde\omega \cdot
h(x_1,\dots,x_n)
=
h(q x_n,x_1,\dots,x_{n-1}),
\end{align*}
for any polynomial $h \in \mathbb{C}[x_1,\dots,x_n]$. Collectively, the operators $\tilde{T}_1,\dots,\tilde{T}_{n-1},\tilde\omega$ give a polynomial representation of the affine Hecke algebra of type $A_{n-1}$. The {\it Cherednik--Dunkl operators} $\tilde{Y}_i$ generate an Abelian subalgebra of the affine Hecke algebra. They are given by
\begin{align}
\label{rev-Yi}
\tilde{Y}_i = 
\tilde{T}_i\cdots \tilde{T}_{n-1}
\cdot
\tilde\omega 
\cdot
\tilde{T}_{1}^{-1}
\cdots 
\tilde{T}_{i-1}^{-1},
\qquad
1 \leq i \leq n,
\end{align}
and satisfy the commutation relations
\begin{align}
\label{Y-commute}
[\tilde{Y}_i,\tilde{Y}_j] = 0,
\qquad
\forall\
1 \leq i,j \leq n.
\end{align}
In view of their commutativity, one can seek to jointly diagonalize the operators \eqref{rev-Yi}. This brings us to the definition of the nonsymmetric Macdonald polynomials:

\begin{defn}
Introduce the following two orders on length-$n$ compositions 
$\mu = (\mu_1,\dots,\mu_n)$ and $\nu = (\nu_1,\dots,\nu_n)$. The first is the dominance order, denoted $<$:
\begin{align*}
\nu < \mu 
\iff 
\sum_{i=1}^{j} \nu_i 
< 
\sum_{i=1}^{j} \mu_i, 
\qquad \forall\ 1 \leq j \leq n.
\end{align*}
The second order is denoted $\prec$:
\begin{align*}
\nu \prec \mu 
\iff
\Big(
\nu^{+} < \mu^{+}
\quad
\text{or}
\quad
\nu^{+} = \mu^{+},
\
\nu < \mu
\Big),
\end{align*}
where $\lambda^{+}$ is the unique partition that can be obtained by permuting the parts of a composition $\lambda$. The nonsymmetric Macdonald polynomials $E_{\mu}(x_1,\dots,x_n;q,t)$ are the unique family of polynomials in $\mathbb{C}[x_1,\dots,x_n]$ which satisfy the triangularity property
\begin{align}
\label{monic}
E_{\mu}(x_1,\dots,x_n;q,t) &= 
x^{\mu}
+ 
\sum_{\nu \prec \mu} \tilde{c}_{\mu,\nu}(q,t)
x^{\nu},
\quad
x^{\mu}
=
\prod_{i=1}^{n} x_i^{\mu_i},
\quad
\tilde{c}_{\mu,\nu}(q,t)
\in
\mathbb{Q}(q,t),
\end{align}
as well as the eigenvalue equation
\begin{align}
\tilde{Y}_i \cdot
E_{\mu}(x_1,\dots,x_n;q,t) 
&= 
\tilde{y}_i(\mu;q,t)E_{\mu}(x_1,\dots,x_n;q,t),
\quad
\forall\ 1 \leq i \leq n,
\label{eig-rev-Yi}
\end{align}
with the eigenvalues on the right-hand side of \eqref{eig-rev-Yi} given by
\begin{align}
\label{tilde-rho}
\tilde{y}_i(\mu;q,t)= q^{\mu_i} t^{\tilde\eta_i(\mu)+n-i},
\quad
\tilde\eta_i(\mu)= -\#\{j < i : \mu_j \geq \mu_i\} - \#\{j > i : \mu_j > \mu_i\}.
\end{align}
\end{defn}

\subsection{Reversing the alphabet and compositions}
\label{ssec:reverse}

So far we adhered to the conventions of \cite{MimachiN} in defining the nonsymmetric Macdonald polynomials, but it will turn out to be convenient in our subsequent proofs to deal with a modified version of the polynomials in which both the order of the alphabet $(x_1,\dots,x_n)$ and the composition $\mu = (\mu_1,\dots,\mu_n)$ get reversed:
\begin{defn}
Let $\mu = (\mu_1,\dots,\mu_n)$ be a composition and $\tilde{\mu} = (\mu_n,\dots,\mu_1)$ its reverse ordering. We make the definition 
\begin{align}
\label{f-E}
f_{\mu}(x_1,\dots,x_n)
:=
E_{\tilde\mu}(x_n,\dots,x_1;q,t).
\end{align}
The polynomials $f_{\mu}(x_1,\dots,x_n)$ will be the subject of Section \ref{sec:matrix-prod}; we will continue to refer to them simply as nonsymmetric Macdonald polynomials.
\end{defn}

Let us now translate the operator definitions \eqref{rev-Hecke}--\eqref{rev-Yi} and properties \eqref{monic}--\eqref{tilde-rho} into the reversed-alphabet setting. In doing so, we drop all tildes from the notation. We define new versions of the Hecke generators
\begin{align}
\label{Hecke}
T_i = t - \frac{x_i-t x_{i+1}}{x_i-x_{i+1}} (1-\s_i),
\quad
T^{-1}_i  
= 
t^{-1}\left(1 - \frac{x_i-t x_{i+1}}{x_i-x_{i+1}} (1-\s_i) \right),
\quad
1 \leq i \leq n-1,
\end{align}
as well as a new cyclic generator
\begin{align}
\label{omega}
\omega = \tau_1 \s_1 \cdots \s_{n-1},
\end{align}
and reintroduce the Cherednik--Dunkl operators as follows: 
\begin{align}
\label{Yi}
Y_i
=
T_{i-1} \cdots T_1 \cdot \omega \cdot T^{-1}_{n-1} \cdots T^{-1}_{i},
\qquad
1 \leq i \leq n.
\end{align}
It is easy to check that the two families of operators \eqref{rev-Yi}, \eqref{Yi} are related via the identity
\begin{align}
\label{YY}
Y_{n-i+1} \cdot h(x_1,\dots,x_n)
=
\tilde{Y}_i \cdot h(x_n,\dots,x_1),
\end{align}
for any polynomial $h \in \mathbb{C}[x_1,\dots,x_n]$.

\begin{prop}
The polynomials $f_{\mu}(x_1,\dots,x_n)$ are uniquely characterized by their monomial expansion
\begin{align}
\label{monic-f}
f_{\mu}(x_1,\dots,x_n) &= 
x^{\mu}
+ 
\sum_{\tilde\nu \prec \tilde\mu} c_{\mu,\nu}(q,t)
x^{\nu},
\quad
c_{\mu,\nu}(q,t)
\in
\mathbb{Q}(q,t),
\end{align}
with the sum taken over all compositions $\nu$ such that $\tilde{\nu} \prec \tilde{\mu}$, as well as the eigenvalue equation
\begin{align}
\label{eig-Yi}
Y_i\cdot
f_{\mu}(x_1,\dots,x_n)
=
y_i(\mu;q,t)
f_{\mu}(x_1,\dots,x_n),
\quad
\forall\ 1 \leq i \leq n,
\end{align}
with the eigenvalues on the right-hand side of \eqref{eig-Yi} given by
\begin{align}
\label{rho}
y_i(\mu;q,t)
=
q^{\mu_i}
t^{\eta_i(\mu)+i-1},
\quad
\eta_i(\mu)= -\#\{j < i : \mu_j > \mu_i\} - \#\{j > i : \mu_j \geq \mu_i\}.
\end{align}
\end{prop}

\begin{proof}
We check how the defining properties \eqref{monic}--\eqref{tilde-rho} of the nonsymmetric Macdonald polynomials translate under the substitution \eqref{f-E}. One immediately sees that \eqref{monic} maps over to \eqref{monic-f}, where the coefficients in the two expansions are related via 
\begin{align*}
\tilde{c}_{\tilde{\mu},\tilde{\nu}}(q,t) = c_{\mu,\nu}(q,t).
\end{align*} 
Similarly, substituting $E_{\mu}(x_1,\dots,x_n;q,t) = f_{\tilde{\mu}}(x_n,\dots,x_1)$ into \eqref{eig-rev-Yi}, we read
\begin{align}
\label{prop-1}
\tilde{Y}_i
\cdot
f_{\tilde{\mu}}(x_n,\dots,x_1)
=
\tilde{y}_i(\mu;q,t)
f_{\tilde{\mu}}(x_n,\dots,x_1),
\end{align}
and using \eqref{YY} we can reverse the order of the alphabets that appear in \eqref{prop-1}. Reversing also the order of the composition $\mu$, one obtains
\begin{align}
Y_{n-i+1}
\cdot
f_{\mu}(x_1,\dots,x_n)
=
\tilde{y}_{i}(\tilde{\mu};q,t)
f_{\mu}(x_1,\dots,x_n),
\end{align}
or equivalently,
\begin{align}
Y_i
\cdot
f_{\mu}(x_1,\dots,x_n)
=
\tilde{y}_{n-i+1}(\tilde{\mu};q,t)
f_{\mu}(x_1,\dots,x_n).
\end{align}
This matches the equation \eqref{eig-Yi} provided that $y_i(\mu;q,t) = 
\tilde{y}_{n-i+1}(\tilde{\mu};q,t)$, which is easily verified using the 
form of the eigenvalues \eqref{tilde-rho} and \eqref{rho}.

\end{proof}

\begin{rmk}
\label{rmk:alt-def}
In the current work, it will be convenient to make use of an alternative characterization of the nonsymmetric Macdonald polynomials. Let $\mathbb{C}^{D}[x_1,\dots,x_n]$ be the set of all polynomials in 
$(x_1,\dots,x_n)$ of total degree $\leq D$. Viewing the latter as a vector space, we can write
\begin{align*}
\mathbb{C}^{D}[x_1,\dots,x_n]
=
{\rm Span}\{x^{\mu}\}_{|\mu| \leq D},
\qquad
|\mu| := \sum_{i=1}^{n} \mu_i.
\end{align*}
Introduce a generating series $Y(w) = \sum_{i=1}^{n} Y_i w^{i-1}$ of the operators \eqref{Yi}, which can be viewed as a linear operator $Y(w) : \mathbb{C}^{D}[x_1,\dots,x_n] \rightarrow \mathbb{C}^{D}[x_1,\dots,x_n]$. The corresponding eigenvalues $y(\mu;q,t;w) = \sum_{i=1}^{n} y_i(\mu;q,t) w^{i-1}$ are pairwise distinct for all $|\mu| \leq D$, with $q,t,w$ living inside some open set in $\mathbb{C}$.

The simplicity of the spectrum of $Y(w)$ uniquely determines each eigenvector $f_{\mu}(x_1,\dots,x_n)$ (with $|\mu| \leq D$) up to a multiplicative constant. One can therefore define $f_{\mu}$ as the unique polynomial solution of the equations \eqref{eig-Yi} such that $f_{\mu} \in \mathbb{C}^{|\mu|}[x_1,\dots,x_n]$ and ${\rm Coeff}[f_{\mu};x^{\mu}] = 1$, and we will tacitly assume this definition in what follows.

Simplicity of the spectrum plays an important role in Cherednik's theory of double affine Hecke algebras, \cf\ \cite[Theorem 8.2(i)]{Cherednik3}.
\end{rmk}

\section{Vertex models related to $\mathfrak{sl}_{n+1}$}

\subsection{Fundamental $R$-matrix}

The lattice models which appear in this work are based on certain representations of the quantized affine algebra 
$U_t(\wh{\mathfrak{sl}_{n+1}})$.\footnote{Throughout the paper, we use $t$ to indicate the intrinsic quantum parameter of our lattice models, while $q$ will appear extrinsically in our choice of boundary conditions in partition functions. We make these choices in order to be consistent with the canonical use of $(q,t)$ parameters in Macdonald polynomial theory \cite{Macdonald}.} In this setting, the basic object is the {\it fundamental $\mathfrak{sl}_{n+1}$ $R$-matrix} \cite{Jimbo1,Jimbo2,Bazhanov,FRT}, which takes the form
\begin{align}
\label{Rmat}
R_z
&=
\sum_{i=0}^{n}
\left(
R_z(i,i;i,i)
E^{(ii)} \otimes E^{(ii)}
\right)
\\
\nonumber
&+
\sum_{0 \leq i < j \leq n}
\left(
R_z(j,i;j,i)
E^{(ii)} \otimes E^{(jj)}
+
R_z(i,j;i,j)
E^{(jj)} \otimes E^{(ii)}
\right)
\\
\nonumber
&+
\sum_{0 \leq i < j \leq n}
\left(
R_z(j,i;i,j)
E^{(ij)} \otimes E^{(ji)}
+
R_z(i,j;j,i)
E^{(ji)} \otimes E^{(ij)}
\right)
\end{align}
where $E^{(ij)}$ denotes the $(n+1) \times (n+1)$ elementary matrix with a $1$ at position $(i,j)$ and $0$ everywhere else. The matrix entries of \eqref{Rmat} are given explicitly by
\begin{align}
\label{R-weights-a}
&
\left.
R_z(i,i;i,i)
=
1,
\quad
i \in \{0,1,\dots,n\},
\right.
\\
\nonumber
\\
\label{R-weights-bc}
&
\left.
\begin{array}{ll}
R_z(j,i;j,i)
=
\dfrac{t(1-z)}{1-tz},
&
\quad
R_z(i,j;i,j)
=
\dfrac{1-z}{1-tz}
\\ \\
R_z(j,i;i,j)
=
\dfrac{1-t}{1-tz},
&
\quad
R_z(i,j;j,i)
=
\dfrac{(1-t)z}{1-tz}
\end{array}
\right\}
\quad
i,j \in \{0,1,\dots,n\},
\quad  i<j.
\end{align}
All other matrix entries $R_z(i,j;k,\ell)$ which do not fall into a category listed above are by definition equal to $0$.

We shall denote the entries of the $R$-matrix pictorially using {\it vertices}. A vertex is the intersection of an oriented horizontal and vertical line, with a state variable $i \in \{0,1,\dots,n\}$ assigned to each of the four surrounding horizontal and vertical line segments. The $R$-matrix entries are identified with such vertices, as shown below: 
\begin{align}
\label{R-vert}
R_z(i,j; k,\ell)
=
\tikz{0.7}{
\draw[densely dotted] (-0.5,0) arc (180:270:0.5);
\node[below,left] at (0.1,-0.2) {\tiny $z$};
\draw[lgray,line width=1.5pt,->] (-1,0) -- (1,0);
\draw[lgray,line width=1.5pt,->] (0,-1) -- (0,1);
\node[left] at (-1,0) {\tiny $j$};\node[right] at (1,0) {\tiny $\ell$};
\node[below] at (0,-1) {\tiny $i$};\node[above] at (0,1) {\tiny $k$};
},
\quad
i,j,k,\ell \in \{0,1,\dots,n\},
\end{align}
where the angle between the lines takes the value $z$. We will often omit this angle from pictures when the {\it spectral parameter} $z$ of the $R$-matrix is clear from context. Comparing \eqref{R-weights-a}, \eqref{R-weights-bc} against \eqref{R-vert}, we obtain five different types of non-vanishing vertices, with corresponding weights recorded below:
\begin{align}
\label{fund-vert}
\begin{tabular}{|c|c|c|c|c|}
\hline
\quad
\tikz{0.6}{
\draw[lgray,line width=1.5pt,->] (-1,0) -- (1,0);
\draw[lgray,line width=1.5pt,->] (0,-1) -- (0,1);
\node[left] at (-1,0) {\tiny $i$};\node[right] at (1,0) {\tiny $i$};
\node[below] at (0,-1) {\tiny $i$};\node[above] at (0,1) {\tiny $i$};
}
\quad
&
\quad
\tikz{0.6}{
\draw[lgray,line width=1.5pt,->] (-1,0) -- (1,0);
\draw[lgray,line width=1.5pt,->] (0,-1) -- (0,1);
\node[left] at (-1,0) {\tiny $i$};\node[right] at (1,0) {\tiny $i$};
\node[below] at (0,-1) {\tiny $j$};\node[above] at (0,1) {\tiny $j$};
}
\quad
&
\quad
\tikz{0.6}{
\draw[lgray,line width=1.5pt,->] (-1,0) -- (1,0);
\draw[lgray,line width=1.5pt,->] (0,-1) -- (0,1);
\node[left] at (-1,0) {\tiny $j$};\node[right] at (1,0) {\tiny $j$};
\node[below] at (0,-1) {\tiny $i$};\node[above] at (0,1) {\tiny $i$};
}
\quad
&
\quad
\tikz{0.6}{
\draw[lgray,line width=1.5pt,->] (-1,0) -- (1,0);
\draw[lgray,line width=1.5pt,->] (0,-1) -- (0,1);
\node[left] at (-1,0) {\tiny $i$};\node[right] at (1,0) {\tiny $j$};
\node[below] at (0,-1) {\tiny $j$};\node[above] at (0,1) {\tiny $i$};
}
\quad
&
\quad
\tikz{0.6}{
\draw[lgray,line width=1.5pt,->] (-1,0) -- (1,0);
\draw[lgray,line width=1.5pt,->] (0,-1) -- (0,1);
\node[left] at (-1,0) {\tiny $j$};\node[right] at (1,0) {\tiny $i$};
\node[below] at (0,-1) {\tiny $i$};\node[above] at (0,1) {\tiny $j$};
}
\\[1.3cm]
\quad
$1$
\quad
& 
\quad
$\dfrac{t(1-z)}{1-tz}$
\quad
&
\quad
$\dfrac{1-z}{1-tz}$
\quad
& 
\quad
$\dfrac{1-t}{1-tz}$
\quad
&
\quad
$\dfrac{(1-t)z}{1-tz}$
\quad 
\\[0.7cm]
\hline
\end{tabular}
\end{align}
where the indices satisfy the constraint $0 \leq i < j \leq n$. This model can be viewed as a multi-coloured (or higher-rank) extension of the {\it stochastic six-vertex model} \cite{Baxter,GwaS}; it reduces to the latter when $n=1$.

\subsection{Bosonic $L$-matrix}

Our subsequent partition functions will be constructed in terms of another object, which we term the {\it $L$-matrix}. It is an infinite-dimensional matrix, whose components will be represented as {\it faces} with labels assigned to each of the four edges:
\begin{align}
\label{generic-L}
L_x(\I,j; \K,\ell)
\index{L1@$L_x(\I,j; \K,\ell)$; vertex weights}
=
\tikz{0.7}{
\draw[lgray,line width=1.5pt] (-1,-1) -- (1,-1) -- (1,1) -- (-1,1) -- (-1,-1);
\node[left] at (-1,0) {\tiny $j$};\node[right] at (1,0) {\tiny $\ell$};
\node[below] at (0,-1) {\tiny $\I$};\node[above] at (0,1) {\tiny $\K$};
\node[text centered] at (0,0) {$x$};
},
\quad
j,\ell \in \{0,1,\dots,n\},
\quad
\I,\K \in \mathbb{N}^n,
\end{align}
where $\mathbb{N} = \{0,1,2,\dots\}$. Note that the labels assigned to the left and right vertical edges continue to take values in $\{0,1,\dots,n\}$, as in the case of the $R$-matrix vertices \eqref{R-vert}. In contrast, the labels assigned to the bottom and top horizontal edges are now $n$-dimensional vectors of nonnegative integers, with no bound imposed on the size of these integers.\footnote{Our use of faces instead of vertices is for notational convenience only; we are not considering ``solid-on-solid'' (SOS) or ``Interaction Round a Face'' (IRF) models, where the weights depend on values at vertices around a face rather than at their edges.}

Before stating the face weights \eqref{generic-L} explicitly, we develop some useful vector notation. For all $1\leq i \leq n$, let $\bm{e}_i$ denote the $i$-th Euclidean unit vector. For any vector $\I = (I_1,\dots,I_n) \in \mathbb{N}^n$ and indices $i,j \in \{1,\dots,n\}$ we define 
\begin{align*}
\I^{+}_{i}
=
\I + \bm{e}_i,
\quad
\I^{-}_{i}
=
\I - \bm{e}_i,
\quad
\I^{+-}_{ij}
=
\I + \bm{e}_i - \bm{e}_j,
\quad
|\I|
=
\sum_{k=1}^{n} I_k,
\quad
\Is{i}{j}
=
\sum_{k=i}^{j} I_k,
\end{align*}
where in the final case it is assumed that $i \leq j$; by agreement, we choose $\Is{i}{j} = 0$ for $i>j$. In terms of these notations, we tabulate the face weights \eqref{generic-L} below:
\begin{align}
\label{bos-weights}
\begin{tabular}{|c|c|c|}
\hline
\quad
\tikz{0.7}{
\draw[lgray,line width=1.5pt] (-1,-1) -- (1,-1) -- (1,1) -- (-1,1) -- (-1,-1);
\node[left] at (-1,0) {\tiny $0$};\node[right] at (1,0) {\tiny $0$};
\node[below] at (0,-1) {\tiny $\I$};\node[above] at (0,1) {\tiny $\I$};
}
\quad
&
\quad
\tikz{0.7}{
\draw[lgray,line width=1.5pt] (-1,-1) -- (1,-1) -- (1,1) -- (-1,1) -- (-1,-1);
\node[left] at (-1,0) {\tiny $i$};\node[right] at (1,0) {\tiny $i$};
\node[below] at (0,-1) {\tiny $\I$};\node[above] at (0,1) {\tiny $\I$};
}
\quad
&
\quad
\tikz{0.7}{
\draw[lgray,line width=1.5pt] (-1,-1) -- (1,-1) -- (1,1) -- (-1,1) -- (-1,-1);
\node[left] at (-1,0) {\tiny $0$};\node[right] at (1,0) {\tiny $i$};
\node[below] at (0,-1) {\tiny $\I$};\node[above] at (0,1) {\tiny $\I^{-}_i$};
}
\quad
\\[1.3cm]
\quad
$1$
\quad
& 
\quad
$xt^{\Is{i+1}{n}}$
\quad
& 
\quad
$x(1-t^{I_i}) t^{\Is{i+1}{n}}$
\quad
\\[0.7cm]
\hline
\quad
\tikz{0.7}{
\draw[lgray,line width=1.5pt] (-1,-1) -- (1,-1) -- (1,1) -- (-1,1) -- (-1,-1);
\node[left] at (-1,0) {\tiny $i$};\node[right] at (1,0) {\tiny $0$};
\node[below] at (0,-1) {\tiny $\I$};\node[above] at (0,1) {\tiny $\I^{+}_i$};
}
\quad
&
\quad
\tikz{0.7}{
\draw[lgray,line width=1.5pt] (-1,-1) -- (1,-1) -- (1,1) -- (-1,1) -- (-1,-1);
\node[left] at (-1,0) {\tiny $i$};\node[right] at (1,0) {\tiny $j$};
\node[below] at (0,-1) {\tiny $\I$};\node[above] at (0,1) 
{\tiny $\I^{+-}_{ij}$};
}
\quad
&
\quad
\tikz{0.7}{
\draw[lgray,line width=1.5pt] (-1,-1) -- (1,-1) -- (1,1) -- (-1,1) -- (-1,-1);
\node[left] at (-1,0) {\tiny $j$};\node[right] at (1,0) {\tiny $i$};
\node[below] at (0,-1) {\tiny $\I$};\node[above] at (0,1) {\tiny $\I^{+-}_{ji}$};
}
\quad
\\[1.3cm] 
\quad
$1$
\quad
& 
\quad
$x(1-t^{I_j}) t^{\Is{j+1}{n}}$
\quad
&
\quad
$0$
\quad
\\[0.7cm]
\hline
\end{tabular} 
\end{align}
where it is assumed that $1 \leq i < j \leq n$. All of the cases of $L_x(\I,j;\K,\ell)$ listed in \eqref{bos-weights} satisfy 
$\I +{\bm e}_j = \K + {\bm e}_{\ell}$; whenever this condition is violated, we shall set $L_x(\I,j;\K,\ell) = 0$. In view of this conservation property, we are able to think of each of the faces in \eqref{bos-weights} in terms of coloured lattice paths which enter via the left/bottom edges and exit via the right/top edges: whenever a left/right edge assumes the value $i \geq 1$, this means a path of colour $i$ is incident at that edge, and similarly when a bottom/top edge assumes the value 
$\I = (I_1,\dots,I_n)$, this means that there are $I_i$ paths of colour $i$ incident at that edge (for all $1 \leq i \leq n$).

\begin{rmk}
\label{rmk:x-depend}
Note that the face weights \eqref{bos-weights} have a very simple dependence on the parameter $x$. If $\ell \geq 1$, $L_x(\I,j; \K,\ell)$ is a homogeneous polynomial of degree $1$ in $x$; if $\ell = 0$, 
$L_x(\I,j; \K,\ell)$ does not depend on $x$. Put another way, we obtain a factor of $x$ every time a path passes through the right vertical edge of a face.
\end{rmk}

\begin{ex}
Let $n=3$ and consider $L_x(\I,j; \K,\ell)$ for $\I = (1,1,1)$, $j=1$, 
$\K = (2,0,1)$, $\ell = 2$. It has the pictorial representation and corresponding weight
\begin{align*}
L_x(\I,j; \K,\ell)
=
\tikz{0.7}{
\draw[lgray,line width=1.5pt] (-1,-1) -- (1,-1) -- (1,1) -- (-1,1) -- (-1,-1);
\node[left] at (-1,0) {\tiny $1$};\node[right] at (1,0) {\tiny $2$};
\node[below] at (0,-1) {\tiny $(1,1,1)$};\node[above] at (0,1) 
{\tiny $(2,0,1)$};
\draw[red,line width=1pt,->] (-1,0) -- (-0.5,0) -- (-0.5,1);
\draw[red,line width=1pt,->] (-0.3,-1) -- (-0.3,1);
\draw[green,line width=1pt,->] (0,-1) -- (0,0) -- (1,0);
\draw[blue,line width=1pt,->] (0.3,-1) -- (0.3,1);
}
=
x(1-t)t.
\end{align*}
\end{ex}

\begin{ex}
Let $n=3$ and consider $L_x(\I,j; \K,\ell)$ for $\I = (2,1,2)$, $j=0$, 
$\K = (1,1,2)$, $\ell = 1$. It has the pictorial representation and corresponding weight
\begin{align*}
L_x(\I,j; \K,\ell)
=
\tikz{0.7}{
\draw[lgray,line width=1.5pt] (-1,-1) -- (1,-1) -- (1,1) -- (-1,1) -- (-1,-1);
\node[left] at (-1,0) {\tiny $0$};\node[right] at (1,0) {\tiny $1$};
\node[below] at (0,-1) {\tiny $(2,1,2)$};\node[above] at (0,1) 
{\tiny $(1,1,2)$};
\draw[red,line width=1pt,->] (-0.5,-1) -- (-0.5,1);
\draw[red,line width=1pt,->] (-0.3,-1) -- (-0.3,0) -- (1,0);
\draw[green,line width=1pt,->] (0,-1) -- (0,1);
\draw[blue,line width=1pt,->] (0.3,-1) -- (0.3,1);
\draw[blue,line width=1pt,->] (0.5,-1) -- (0.5,1);
}
=
x(1-t^2)t^3.
\end{align*}
\end{ex}

\subsection{Yang--Baxter equation}

Now we come to the key property of the $R$ and $L$-matrices thus introduced, the {\it Yang--Baxter equation}.

\begin{prop}
\label{prop:RLL}
For any fixed integers $i_1,i_2,i_3,j_1,j_2,j_3 \in \{0,1,\dots,n\}$ and vectors $\I,\J \in \mathbb{N}^n$, the weights \eqref{R-weights-a}, \eqref{R-weights-bc}, \eqref{bos-weights} satisfy the relation
\begin{multline}
\label{RLL}
\sum_{0 \leq k_1,k_2 \leq n}
\
\sum_{\K \in \mathbb{N}^n}
R_{y/x}(i_2,i_1;k_2,k_1)
L_x(\I,k_1;\K,j_1)
L_y(\K,k_2;\J,j_2)
\\
=
\sum_{0 \leq k_1,k_2 \leq n}
\
\sum_{\K \in \mathbb{N}^n}
L_y(\I,i_2;\K,k_2)
L_x(\K,i_1;\J,k_1)
R_{y/x}(k_2,k_1;j_2,j_1).
\end{multline}
In terms of the graphical conventions adopted in \eqref{R-vert} and \eqref{generic-L}, the equation \eqref{RLL} translates into
\begin{align}
\label{graph-RLL}
\sum_{0 \leq k_1,k_2 \leq n}
\
\sum_{\K \in \mathbb{N}^n}
\tikz{1.7}{
\draw[densely dotted] (-1.75,0.75) arc (135:225:{1/sqrt(8)});
\node[left] at (-1.47,0.5) {\tiny $y/x$};
\draw[lgray,line width=1.5pt] (-1,-0.5) -- (0,-0.5) -- (0,1.5) -- (-1,1.5) -- (-1,-0.5);
\draw[lgray,line width=1.5pt] (-1,0.5) -- (0,0.5);
\node[text centered] at (-0.5,0) {$x$};
\node[text centered] at (-0.5,1) {$y$};
\node[right] at (0,1) {\fs $j_2$};
\node[right] at (0,0) {\fs $j_1$};
\node[below] at (-0.5,-0.5) {\fs $\I$};
\node at (-0.5,0.5) {\fs $\K$};
\node[above] at (-0.5,1.5) {\fs $\J$};
\draw[lgray,line width=1.5pt,->]
(-2,1) node[left] {\color{black} \fs $i_1$} -- (-1,0) node[left] {\color{black} \fs $k_1\phantom{.}$};
\draw[lgray,line width=1.5pt,->] 
(-2,0) node[left] {\color{black} \fs $i_2$} -- (-1,1) node[left] {\color{black} \fs $k_2\phantom{.}$ };
}
\quad
=
\quad
\sum_{0 \leq k_1,k_2 \leq n}
\
\sum_{\K \in \mathbb{N}^n}
\tikz{1.7}{
\draw[densely dotted] (0.25,0.75) arc (135:225:{1/sqrt(8)});
\node[left] at (0.53,0.5) {\tiny $y/x$};
\draw[lgray,line width=1.5pt] (-1,-0.5) -- (0,-0.5) -- (0,1.5) -- (-1,1.5) -- (-1,-0.5);
\draw[lgray,line width=1.5pt] (-1,0.5) -- (0,0.5);
\node[text centered] at (-0.5,0) {$y$};
\node[text centered] at (-0.5,1) {$x$};
\node[left] at (-1,1) {\fs $i_1$};
\node[left] at (-1,0) {\fs $i_2$};
\node[below] at (-0.5,-0.5) {\fs $\I$};
\node at (-0.5,0.5) {\fs $\K$};
\node[above] at (-0.5,1.5) {\fs $\J$};
\draw[lgray,line width=1.5pt,->] 
(0,1) node[right] {\color{black} \fs $k_1$} -- (1,0) node[right] {\color{black} \fs $j_1$};
\draw[lgray,line width=1.5pt,->] 
(0,0) node[right] {\color{black} \fs $k_2$} -- (1,1) node[right] {\color{black} \fs $j_2$};
}
\end{align}
in which we explicitly show the parameter dependence of the $R$ and $L$-matrices.
\end{prop}

\begin{rmk}
Equation \eqref{RLL} is recovered as the $q \mapsto t$, $s=0$ reduction of \cite[Section 2.3, Equation (2.3.1)]{BorodinW2}. It also appeared, in terms of operators in the $t$-boson algebra, in \cite{CantiniGW} and \cite{GarbaliGW}.
\end{rmk}

\subsection{Row operators and commutation relations}

Let $V$ be the infinite-dimensional vector space obtained by taking the linear span of all $n$-tuples of nonnegative integers:  
\begin{align*}
V:= {\rm Span}\{\ket{\I}\} =
{\rm Span}\{\ket{i_1,\dots,i_n}\}_{i_1,\dots,i_n \in \mathbb{N}}.
\end{align*}
Fix a nonnegative integer\footnote{Later, we will specify the precise value of $N$ that we require. For now, it is treated as arbitrary.} $N$ and
consider an $(N+1)$-fold tensor product of such spaces, 
\begin{align*}
\mathbb{V}(N) := V_0 \otimes V_1 \otimes \cdots \otimes V_N,
\end{align*}
where each $V_i$ denotes a copy of $V$. We introduce a family of linear operators acting on $\mathbb{V}(N)$ as follows:
\begin{align}
\label{C-row}
\C_i(x)
:
\bigotimes_{k=0}^{N}
\ket{\J(k)}_k
\mapsto
\sum_{\I(0),\ldots,\I(N) \in \mathbb{N}^n}
\left(
\tikz{1.1}{
\draw[lgray,line width=1.5pt] (0.5,-0.5) -- (6.5,-0.5) -- (6.5,0.5) -- (0.5,0.5) -- (0.5,-0.5);
\foreach\x in {1,...,5}{
\draw[lgray,line width=1.5pt] (0.5+\x,-0.5) -- (0.5+\x,0.5);
}
\node[text centered] at (1,0) {$x$};
\node[text centered] at (2,0) {$x$};
\node[text centered] at (3,0) {$x$};
\node[text centered] at (4,0) {$x$};
\node[text centered] at (5,0) {$x$};
\node[text centered] at (6,0) {$x$};
\node[left] at (0.5,0) {\fs $i$};\node[right] at (6.5,0) {\fs $0$};
\node[below] at (6,-0.5) {\fs $\I(N)$};\node[above] at (6,0.5) {\fs $\J(N)$};
\node[below] at (5,-0.5) {\fs $\cdots$};\node[above] at (5,0.5) {\fs $\cdots$};
\node[below] at (4,-0.5) {\fs $\cdots$};\node[above] at (4,0.5) {\fs $\cdots$};
\node[below] at (3,-0.5) {\fs $\cdots$};\node[above] at (3,0.5) {\fs $\cdots$};
\node[below] at (2,-0.5) {\fs $\I(1)$};\node[above] at (2,0.5) {\fs $\J(1)$};
\node[below] at (1,-0.5) {\fs $\I(0)$};\node[above] at (1,0.5) {\fs $\J(0)$};
}
\right)
\bigotimes_{k=0}^{N}
\ket{\I(k)}_k,
\end{align}
where $1 \leq i \leq n$. Note that each face weight appearing on the right-hand side of \eqref{C-row} depends on the same parameter, $x$. The quantity
\begin{align*}
\tikz{1.1}{
\draw[lgray,line width=1.5pt] (0.5,-0.5) -- (6.5,-0.5) -- (6.5,0.5) -- (0.5,0.5) -- (0.5,-0.5);
\foreach\x in {1,...,5}{
\draw[lgray,line width=1.5pt] (0.5+\x,-0.5) -- (0.5+\x,0.5);
}
\node[text centered] at (1,0) {$x$};
\node[text centered] at (2,0) {$x$};
\node[text centered] at (3,0) {$x$};
\node[text centered] at (4,0) {$x$};
\node[text centered] at (5,0) {$x$};
\node[text centered] at (6,0) {$x$};
\node[left] at (0.5,0) {\fs $i$};\node[right] at (6.5,0) {\fs $0$};
\node[below] at (6,-0.5) {\fs $\I(N)$};\node[above] at (6,0.5) {\fs $\J(N)$};
\node[below] at (5,-0.5) {\fs $\cdots$};\node[above] at (5,0.5) {\fs $\cdots$};
\node[below] at (4,-0.5) {\fs $\cdots$};\node[above] at (4,0.5) {\fs $\cdots$};
\node[below] at (3,-0.5) {\fs $\cdots$};\node[above] at (3,0.5) {\fs $\cdots$};
\node[below] at (2,-0.5) {\fs $\I(1)$};\node[above] at (2,0.5) {\fs $\J(1)$};
\node[below] at (1,-0.5) {\fs $\I(0)$};\node[above] at (1,0.5) {\fs $\J(0)$};
}
\end{align*}
is a one-row partition function in the model \eqref{bos-weights}, and can be calculated by multiplying the weights of each face from left to right, noting that the integer values prescribed to all internal vertical edges are fixed by the local conservation property of the $L$-matrix.

By virtue of the Yang--Baxter equation \eqref{RLL}, one can now derive the following set of commutation relations of the row operators 
\eqref{C-row}:
\begin{prop}
\label{prop:CC}
Fix two positive integers $i,j$ such that $1 \leq i,j \leq n$. The row operators \eqref{C-row} satisfy the exchange relations
\begin{align}
\label{CC=}
\C_i(x) \C_i(y) &= \C_i(y) \C_i(x),
\\
\label{CC<}
t\, \C_i(x) \C_j(y) &= 
\frac{x-ty}{x-y}\, \C_j(y) \C_i(x) 
- 
\frac{(1-t)x}{x-y}\, \C_j(x) \C_i(y),
\quad\quad
i<j,
\\
\label{CC>}
\C_i(x) \C_j(y) &= 
\frac{x-ty}{x-y}\, \C_j(y) \C_i(x) 
- 
\frac{(1-t)y}{x-y}\, \C_j(x) \C_i(y),
\quad\quad
i>j.
\end{align}
\end{prop}

\begin{proof}
These relations were obtained in \cite[Theorem 3.2.1]{BorodinW2}. Their proof is based on writing an extension of \eqref{graph-RLL} to $N+1$ sites, and then choosing appropriate boundary conditions of the resulting relation. For the full details of the proof, see \cite[Section 3.2]{BorodinW2}. Let us note that we will only require the relation \eqref{CC>} in what follows.
\end{proof}

\subsection{Generic and composition states}

We now set up some useful definitions regarding the vector space $\mathbb{V}(N)$. Let $\{m_{i,j}\}_{1 \leq i \leq n, 0 \leq j \leq N}$ be a collection of nonnegative integers. We associate to such a set the vector
\begin{align}
\label{generic-vec}
\ket{\{m\}}
=
\bigotimes_{j=0}^{N}
\ket{\bm{M}(j)}_j
\in
\mathbb{V}(N),
\qquad
\text{with}\ \
\bm{M}(j) = \sum_{i=1}^{n} m_{i,j} \bm{e}_i,
\end{align}
as well as a state in the dual vector space
\begin{align}
\label{generic-dual-vec}
\bra{\{m\}}
=
\bigotimes_{j=0}^{N}
\bra{\bm{M}(j)}_j
\in
\mathbb{V}(N)^{*},
\qquad
\text{with}\ \
\bm{M}(j) = \sum_{i=1}^{n} m_{i,j} \bm{e}_i,
\end{align}
where we assume the inner product
\begin{align}
\label{inner}
\langle \{m\} | \{m'\} \rangle = 
\prod_{i=1}^{n} \prod_{j=0}^{N} 
{\bm 1}(m_{i,j} = m'_{i,j}).
\end{align}
We refer to the vectors $\ket{\{m\}}$ and $\bra{\{m\}}$ as generic states; they provide a natural basis for $\mathbb{V}(N)$ and $\mathbb{V}(N)^{*}$. A particular subclass of the vectors \eqref{generic-vec}--\eqref{generic-dual-vec} will play an important role in what follows. Let $\mu \in \mathbb{N}^n$ be a composition, and choose $N$ such that $\max_{1 \leq i \leq n}(\mu_i) \leq N$. We define
\begin{align}
\label{comp-state}
\ket{\mu}
=
\bigotimes_{j=0}^{N}
\ket{\bm{\mu}(j)}_j
\in
\mathbb{V}(N),
\qquad
\text{with}\ \
\bm{\mu}(j) = \sum_{i=1}^{n} \bm{1}(\mu_i = j) \bm{e}_i,
\end{align}
and the corresponding dual vector
\begin{align}
\bra{\mu}
=
\bigotimes_{j=0}^{N}
\bra{\bm{\mu}(j)}_j
\in
\mathbb{V}(N)^{*},
\qquad
\text{with}\ \
\bm{\mu}(j) = \sum_{i=1}^{n} \bm{1}(\mu_i = j) \bm{e}_i.
\end{align}
We refer to the vectors $\ket{\mu}$ and $\bra{\mu}$ as composition states; we will use them in our partition function representation of the nonsymmetric Macdonald polynomials $f_{\mu}(x_1,\dots,x_n)$.

\section{Partition function construction of the nonsymmetric Macdonald polynomials}
\label{sec:matrix-prod}

In this section we present our main result, namely an explicit partition function expression for the nonsymmetric Macdonald polynomial $f_{\mu}(x_1,\dots,x_n)$, for any composition $\mu \in \mathbb{N}^n$; this is the content of Theorem \ref{thm:f-formula}. The proof of this theorem will be given progressively over the course of the section; the elements of the proof are showing that the proposed formula \eqref{f-ansatz} satisfies {\bf 1.} certain exchange relations under the action of Hecke generators \eqref{Hecke}, and {\bf 2.} a particular cyclic relation under the action of the operator \eqref{omega}. Combining {\bf 1} and {\bf 2} allows one to show that \eqref{f-ansatz} satisfies the eigenvalue equation \eqref{eig-Yi}, which uniquely characterizes the nonsymmetric Macdonald polynomials (up to normalization).

Theorem \ref{thm:f-formula} expresses the nonsymmetric Macdonald polynomials in terms of a certain linear form which acts on a product of the row operators \eqref{C-row} and returns a scalar quantity; this is often termed a ``matrix product formula'' in the literature. In particular, the formula 
\eqref{f-ansatz} was inspired by the earlier work of \cite{CantiniGW}, where a matrix product expression was obtained for a different family of nonsymmetric polynomials, which also yield (symmetric) Macdonald polynomials under appropriate symmetrization. We will comment more on the difference between the present work and that of \cite{CantiniGW} in Remark \ref{rmk:cdgw}.

\subsection{Matrix product formula for $f_{\mu}(x_1,\dots,x_n)$}

\begin{defn}
Let $\mu$ be a composition and denote its largest part by $N = \max_{1 \leq i \leq n}(\mu_i)$. Fix a linear operator $X \in {\rm End}(\mathbb{V}(N))$, and a set of indeterminates $\{v_{i,j}\}_{1 \leq i \leq n, 0 \leq j \leq N}$. We define the linear form
\begin{align}
\label{form}
\langle X \rangle_{\mu}
:=
\sum_{\{m\}}
\prod_{i=1}^{n}
\prod_{j=0}^{N}
v_{i,j}^{m_{i,j}}
\bra{\{m\}} X \ket{\{m\} + \mu},
\end{align}
where the dual state $\bra{\{m\}}$ is given by \eqref{generic-dual-vec}, and where we have defined
\begin{align}
\ket{\{m\} + \mu}
=
\bigotimes_{j=0}^{N}
\ket{\bm{M}(j) + \bm{\mu}(j)}_j,
\qquad
\text{with}\ \ 
\bm{M}(j) + \bm{\mu}(j) = 
\sum_{i=1}^{n} (m_{i,j} + \bm{1}(\mu_i = j)) \bm{e}_i.
\end{align}
The summation in \eqref{form} is taken over all nonnegative integers $m_{i,j} \geq 0$, for each $1 \leq i \leq n$ and $0 \leq j \leq N$. These sums are to be interpreted as convergent geometric series (\ie\ we assume that $|v_{i,j}|$ is sufficiently small, for all $i,j$).
\end{defn}

\begin{thm}
\label{thm:f-formula}
Fix a composition $\mu$, set $N = \max_{1 \leq i \leq n}(\mu_i)$, and recall the definition \eqref{C-row} of the row operators $\C_i(x)$. The nonsymmetric Macdonald polynomial $f_{\mu}(x_1,\dots,x_n)$ is given in terms of the linear form \eqref{form} as follows:
\begin{align}
\label{f-ansatz}
f_{\mu}(x_1,\dots,x_n)
=
\Omega_{\mu}(q,t)
\langle \C_1(x_1) \cdots \C_n(x_n) \rangle_{\mu},
\end{align}
where the parameters $v_{i,j}$ are chosen to be
\begin{align}
\label{v_ij}
v_{i,j}
=
q^{\mu_i-j}
t^{\gamma_{i,j}(\mu)}
\bm{1}(\mu_i > j),
\qquad
1 \leq i \leq n,
\ \ 
0 \leq j \leq N,
\end{align}
with exponents $\gamma_{i,j}(\mu)$ given by
\begin{align}
\label{gamma}
\gamma_{i,j}(\mu)
=
-\#\{k < i : \mu_k > \mu_i\} + \#\{k > i : j \leq \mu_k < \mu_i\}.
\end{align}
The normalization factor $\Omega_{\mu}(q,t)$ appearing on the right-hand side of \eqref{f-ansatz} is a factorized polynomial in $(q,t)$, given explicitly as
\begin{align}
\label{Omega}
\Omega_{\mu}(q,t)
=
\prod_{i=1}^{n}
\prod_{j=0}^{\mu_i-1}
\left(1-q^{\mu_i-j} t^{\alpha_{i,j}(\mu)}\right),
\end{align}
where we have defined
\begin{align}
\label{alpha}
\alpha_{i,j}(\mu)
=
\#\{k < i: \mu_k = \mu_i\}
+
\#\{k \not= i: j < \mu_k < \mu_i\}
+
\#\{k > i: j = \mu_k\}.
\end{align}
\end{thm}

The proof of Theorem \ref{thm:f-formula} will be dealt with over the course of Sections \ref{ssec:recast}--\ref{ssec:norm}. Before moving on to the proof, let us comment on two special cases of the formula \eqref{f-ansatz} which were known prior to this work.

\begin{rmk}
\label{rmk:cdgw}
When the composition $\mu$ is anti-dominant, \ie\ in the case $\mu = \delta$ where $\delta_1 \leq \cdots \leq \delta_n$, the formula \eqref{f-ansatz} previously appeared in \cite{CantiniGW}. In that work the focus was on the family of polynomials
\begin{align}
\label{f-delta}
f_{\delta}^{\rho}(x_1,\dots,x_n)
=
\Omega_{\delta}(q,t)
\langle \C_{\rho_1}(x_1) \cdots \C_{\rho_n}(x_n) \rangle_{\delta},
\qquad
v_{i,j}
=
q^{\delta_i-j} \bm{1}(\delta_i > j),
\end{align}
where $\delta$ is a fixed anti-dominant composition and $\rho = (\rho_1,\dots,\rho_n)$ is allowed to vary (we have assumed that $\rho$ is a permutation, meaning that its parts are pairwise distinct, but \cite{CantiniGW} also catered for the situation when $\rho$ is a generic composition). The polynomials thus defined by \eqref{f-delta} are closely related to the stationary distribution of the multi-species ASEP on a ring; indeed, the formula \eqref{f-delta} originally arose by incorporating spectral parameters $(x_1,\dots,x_n)$ into known matrix product algorithms for such probabilities \cite{FerrariM,EvansFM,ProlhacEM,AritaAMP}.
\end{rmk}

\begin{rmk}
When $q=0$, all of the parameters \eqref{v_ij} vanish, either due to the positive exponents of $q$ or the indicator function $\bm{1}(\mu_i > j)$. This, in turn, drastically simplifies the linear form \eqref{form}; the summation over $\{m\}$ is collapsed to a single non-vanishing term, corresponding to the choice $m_{i,j} = 0$ for all $1 \leq i \leq n$, $0 \leq j \leq N$. One thus finds that
\begin{align*}
\langle X \rangle_{\mu}
\Big|_{v_{i,j} = 0}
=
\bra{\varnothing} X \ket{\mu},
\qquad
\text{where}\ \ 
\bra{\varnothing}
=
\bigotimes_{j=0}^{N} \bra{\bm{0}}_j,
\quad
\bm{0}
=
(0,\dots,0) \in \mathbb{N}^n,
\end{align*}
with $\ket{\mu}$ given by \eqref{comp-state}, as usual. Now using the fact that $\Omega_{\mu}(0,t) = 1$, we find that
\begin{align}
\label{hl-limit}
f_{\mu}(x_1,\dots,x_n)
\Big|_{q=0}
=
\bra{\varnothing} 
\C_1(x_1) \cdots \C_n(x_n) 
\ket{\mu}.
\end{align}
Equation \eqref{hl-limit} thus provides an algebraic formula for the {\it nonsymmetric Hall--Littlewood polynomials}. This result was recently obtained in \cite{BorodinW2}; see, in particular, Definition 3.4.3 and Section 5.9 therein\footnote{Note that \cite{BorodinW2} uses $q$ instead of $t$ to denote the Hall--Littlewood parameter.}. The present paper thus refines the work of \cite{BorodinW2} to the full Macdonald level, by the introduction of a nonzero $q$ parameter.
\end{rmk}

\subsection{Recasting the eigenvalue relation}
\label{ssec:recast}

We begin by modifying \eqref{eig-Yi} slightly, by transferring the string of Hecke generators $T_{i-1} \cdots T_1$ (contained within the Cherednik--Dunkl operator $Y_i$) to the right-hand side of the relation. One obtains
\begin{align}
\label{eig-recast}
\omega \cdot T^{-1}_{n-1} \cdots T^{-1}_i
\cdot
f_{\mu}(x_1,\dots,x_n)
=
y_i(\mu;q,t)
T^{-1}_1 \cdots T^{-1}_{i-1}
\cdot
f_{\mu}(x_1,\dots,x_n).
\end{align}
We wish to substitute the Ansatz \eqref{f-ansatz} into \eqref{eig-recast}, and check its validity. In order to do so, we require the following auxiliary result:
\begin{prop} 
Fix two integers $1 \leq j<k \leq n$ and a further integer $1 \leq i \leq n-1$ (unrelated to the first two). The operator identity
\begin{align}
\label{T-CC}
T^{-1}_i
\cdot
\C_j(x_i) \C_k(x_{i+1})
=
t^{-1}
\C_k(x_i) \C_j(x_{i+1})
\end{align}
holds in ${\rm End}(\mathbb{V}(N))$.
\end{prop}

\begin{proof}
This follows by direct calculation, using the explicit form \eqref{Hecke} of $T^{-1}_i$. We compute
\begin{align*}
T^{-1}_i
\cdot
\C_j(x_i) \C_k(x_{i+1})
&=
t^{-1}\left(
\frac{(t-1)x_{i+1}}{x_i-x_{i+1}}
\C_j(x_i) \C_k(x_{i+1})
+
\frac{x_i-tx_{i+1}}{x_i-x_{i+1}}
\C_j(x_{i+1}) \C_k(x_i)
\right)
\\
&=
t^{-1}
\C_k(x_i) \C_j(x_{i+1}),
\end{align*}
where the final equality is a consequence of the commutation relation \eqref{CC>} with $i \mapsto k$, $x \mapsto x_i$, $y \mapsto x_{i+1}$. Note that the identity \eqref{T-CC} gives an immediate proof of Proposition \ref{prop1-intro}, stated in the introduction.
\end{proof}

\subsection{Left-hand side of eigenvalue equation}

Let us now examine the left-hand side of \eqref{eig-recast}, assuming the substitution \eqref{f-ansatz}. It becomes
\begin{align*}
\omega \cdot T^{-1}_{n-1} \cdots T^{-1}_i
\cdot
f_{\mu}(x_1,\dots,x_n)
=
\Omega_{\mu}(q,t)
\cdot
\omega
\cdot
T^{-1}_{n-1} \cdots T^{-1}_i
\cdot
\langle \C_1(x_1) \cdots \C_n(x_n) \rangle_{\mu},
\end{align*}
and by $n-i$ applications of the identity \eqref{T-CC} we transfer the operator $\C_i$ towards the right of the operator product, to obtain
\begin{multline*}
\omega \cdot T^{-1}_{n-1} \cdots T^{-1}_i
\cdot
f_{\mu}(x_1,\dots,x_n)
\\
=
\Omega_{\mu}(q,t) t^{i-n}
\cdot
\omega
\cdot
\langle \C_1(x_1) \cdots \C_{i-1}(x_{i-1})
\C_{i+1}(x_i) \cdots \C_n(x_{n-1})
\C_i(x_n) \rangle_{\mu}.
\end{multline*}
Finally we act explicitly with $\omega$ (\cf\ equation \eqref{omega}), which yields
\begin{multline}
\label{lhs}
\omega \cdot T^{-1}_{n-1} \cdots T^{-1}_i
\cdot
f_{\mu}(x_1,\dots,x_n)
\\
=
\Omega_{\mu}(q,t) t^{i-n}
\cdot
\langle \C_1(x_2) \cdots \C_{i-1}(x_i)
\C_{i+1}(x_{i+1}) \cdots \C_n(x_n)
\C_i(qx_1) \rangle_{\mu}.
\end{multline}

\subsection{Right-hand side of eigenvalue equation}

Repeating the exercise for the right-hand side of \eqref{eig-recast}, we have
\begin{align*}
y_i(\mu;q,t)
T^{-1}_1 \cdots T^{-1}_{i-1}
\cdot
f_{\mu}(x_1,\dots,x_n)
=
y_i(\mu;q,t)
\Omega_{\mu}(q,t)
\cdot
T^{-1}_1 \cdots T^{-1}_{i-1}
\cdot
\langle \C_1(x_1) \cdots \C_n(x_n) \rangle_{\mu},
\end{align*}
and by $i-1$ applications of the identity \eqref{T-CC} we send the operator $\C_i$ towards the left of the operator product. One finds that
\begin{multline}
\label{rhs}
y_i(\mu;q,t)
T^{-1}_1 \cdots T^{-1}_{i-1}
\cdot
f_{\mu}(x_1,\dots,x_n)
\\
=
t^{1-i}
y_i(\mu;q,t)
\Omega_{\mu}(q,t)
\cdot
\langle 
\C_i(x_1)
\C_1(x_2) \cdots \C_{i-1}(x_i)
\C_{i+1}(x_{i+1}) \cdots \C_n(x_n)
\rangle_{\mu}.
\end{multline}

\subsection{Cyclic relation}

Matching the right-hand sides of \eqref{lhs} and \eqref{rhs}, and applying $\mathfrak{s}_{i-1} \cdots \mathfrak{s}_1$ to both sides of the resulting equation, we are left with the task of proving that
\begin{align*}
\langle \C_1(x_1) \cdots \wh{\C_i(x_i)} \cdots \C_n(x_n) \C_i(qx_i) \rangle_{\mu}
=
t^{n-2i+1} y_i(\mu;q,t)
\langle \C_i(x_i) \C_1(x_1) \cdots \wh{\C_i(x_i)} \cdots \C_n(x_n) \rangle_{\mu},
\end{align*}
where we use the circumflex to indicate omission of the operator $\C_i(x_i)$ at that location. Using the explicit form of the eigenvalue \eqref{rho}, this we can write as
\begin{align}
\label{cyclic-rel}
\langle \C_1(x_1) \cdots \wh{\C_i(x_i)} \cdots \C_n(x_n) \C_i(qx_i) \rangle_{\mu}
=
q^{\mu_i}
t^{\gamma_{i,0}(\mu)}
\langle \C_i(x_i) \C_1(x_1) \cdots \wh{\C_i(x_i)} \cdots \C_n(x_n) \rangle_{\mu},
\end{align}
where we have used the fact that
\begin{align*}
n-i+\eta_i(\mu)
=
-\#\{j < i : \mu_j > \mu_i\} + \#\{j > i : \mu_j < \mu_i\}
\equiv
\gamma_{i,0}(\mu).
\end{align*}
We now devote our efforts to proving \eqref{cyclic-rel}, under the assumption that the twist parameters associated to the linear form \eqref{form} are given by \eqref{v_ij}--\eqref{gamma}.

\subsection{Column operators}

Our main tool for proving \eqref{cyclic-rel} will be another class of objects which we call {\it column operators}. These are constructed in terms of ``towers'' of the face weights, as follows. Fix two sets of nonnegative integers $\{i_1,\dots,i_n\}$, $\{j_1,\dots,j_n\}$ with 
$i_k,j_k \in \{0,1,\dots,n\}$ for all $1 \leq k \leq n$, as well as two compositions $\I, \J \in \mathbb{N}^n$. We define the following extension of the faces \eqref{generic-L}:
\begin{align}
\label{tower}
L(\I, i_1,\dots,i_n; \J, j_1,\dots, j_n)
:=
\tikz{1.2}{
\foreach\y in {1,...,4}{
\draw[lgray,line width=1.5pt] (1.5,0.5+\y) -- (2.5,0.5+\y);
}
\draw[lgray,line width=1.5pt] (1.5,0.5) -- (2.5,0.5) -- (2.5,5.5) -- (1.5,5.5) -- (1.5,0.5);
\node at (2,1.5) {\footnotesize$\bullet$}; \draw (2,1.5) -- (4,1.5); \node[right] at (4,1.5) {\footnotesize$\K(1)$};
\node at (2,2.5) {\footnotesize$\bullet$}; \draw (2,2.5) -- (4,2.5); \node[right] at (4,2.5) {\footnotesize$\K(2)$}; 
\node at (2,4.5) {\footnotesize$\bullet$}; \draw (2,4.5) -- (4,4.5);
\node[right] at (4,4.5) {\footnotesize$\K(n-1)$};
\node[text centered] at (2,1) {$x_1$};
\node[text centered] at (2,2) {$x_2$};
\node[text centered] at (2,3.1) {$\vdots$};
\node[text centered] at (2,4.1) {$\vdots$};
\node[text centered] at (2,5) {$x_n$};
\node[below] at (2,0.5) {\footnotesize$\bm{I}$};
\node[above] at (2,5.5) {\footnotesize$\bm{J}$};
\node[right] at (2.5,1) {$j_1$};
\node[right] at (2.5,2) {$j_2$};
\node[right] at (2.5,3.1) {$\vdots$};
\node[right] at (2.5,4.1) {$\vdots$};
\node[right] at (2.5,5) {$j_n$};
\node[left] at (1.5,1) {$i_1$};
\node[left] at (1.5,2) {$i_2$};
\node[left] at (1.5,3.1) {$\vdots$};
\node[left] at (1.5,4.1) {$\vdots$};
\node[left] at (1.5,5) {$i_n$};
}
\end{align}
where every face weight in the tower is prescribed a different spectral parameter. The states $\K(k)$ assigned to internal horizontal edges (marked by dot points in the picture) are fixed by conservation; namely, $L(\I, i_1,\dots,i_n; \J, j_1,\dots, j_n)$ vanishes identically unless
\begin{align*}
\I + \sum_{\ell=1}^{k} \bm{e}_{i_\ell}
=
\K(k) + \sum_{\ell=1}^{k} \bm{e}_{j_\ell},
\qquad
\forall\ 1 \leq k \leq n,
\qquad
\K(n) \equiv \J.
\end{align*}
In what follows, let $W \cong \mathbb{C}^{n+1}$ be an 
$(n+1)$-dimensional vector space with basis $\{\ket{0},\ket{1},\dots,\ket{n}\}$; or equivalently, we write
\begin{align*}
W
=
{\rm Span}\{\ket{i}\}_{0 \leq i \leq n}.
\end{align*}
The basis of the dual vector space $W^{*}$ will be written as $\{\bra{0},\bra{1},\dots,\bra{n}\}$ with $\langle i | j \rangle = \delta_{i,j}$. We wish to consider an $n$-fold tensor product of $W$, denoted by
\begin{align*}
\mathbb{W}(n)
=
W_1 \otimes \cdots \otimes W_n
=
{\rm Span}\left\{
\bigotimes_{k=1}^{n}
\ket{i_k}_k
\right\}_{0 \leq i_1,\dots,i_n \leq n}
\equiv
{\rm Span}\Big\{
\ket{i_1,\dots,i_n}
\Big\}_{0 \leq i_1,\dots,i_n \leq n}
\end{align*}
and now define operators acting on $\mathbb{W}(n)$.

\begin{defn}
Let $\{v_1,\dots,v_n\}$ be a set of complex parameters, and fix a composition $\AA = (a_1,\dots,a_n) \in \mathbb{N}^n$. We introduce the family of operators $\Lambda_{\{v_1,\dots,v_n\}}(\AA) \in {\rm End}(\mathbb{W}(n))$ with explicit action
\begin{multline}
\label{Q-def}
\Lambda_{\{v_1,\dots,v_n\}}(\AA)
:
\ket{j_1,\dots,j_n}
\mapsto
\\
\sum_{0 \leq i_1,\dots,i_n \leq n}
\left(
\sum_{\M \in \mathbb{N}^n}
\prod_{k=1}^{n}
v_k^{m_k}
L(\M, i_1, \dots, i_n; \M+\AA, j_1, \dots, j_n)
\right)
\ket{i_1,\dots,i_n}.
\end{multline}
The second sum in \eqref{Q-def} runs over all compositions $\M = (m_1,\dots,m_n) \in \mathbb{N}^n$, while in its summand $L(\M, i_1, \dots, i_n; \M+\AA, j_1, \dots, j_n)$ denotes a tower of the form \eqref{tower}. We assume that this sum always converges (which can be ensured by choosing $\{v_1,\dots,v_n\}$ to have small moduli), and thus the matrix elements of $\Lambda_{\{v_1,\dots,v_n\}}(\AA)$ can be viewed as elements of $\mathbb{Q}(v_1,\dots,v_n,t)$.
\end{defn}

The reason for introducing the column operators \eqref{Q-def} is that they offer an alternative algebraic setup for matrix products of the form \eqref{f-ansatz}, which will be very convenient in studying the relation \eqref{cyclic-rel}. In particular, we note the following:
\begin{prop}
Fix a composition $\mu$ and associate to it, in the same way as in \eqref{comp-state}, the binary strings 
$\bm{\mu}(j)$ for each $0 \leq j \leq 
\max_{1\leq i \leq n}(\mu_i) \equiv \max(\mu)$. Then one has
\begin{align}
\label{column-decomp}
\langle \C_1(x_1) \cdots \C_n(x_n) \rangle_{\mu}
=
\bra{1,\dots,n}
\prod_{0 \leq j \leq \max(\mu)}^{\rightarrow}
\Lambda_{\{v_{1,j},\dots,v_{n,j}\}}(\bm{\mu}(j))
\ket{0,\dots,0},
\end{align}
where the product is ordered so that the index $j$ increases as one moves from left to right, and the parameters $v_{i,j}$ are as in \eqref{v_ij}.
\end{prop}

\begin{proof}
By the definition of the row operators \eqref{C-row} and the linear form \eqref{form}, we obtain the following partition function representation of $\langle \C_1(x_1) \cdots \C_n(x_n) \rangle_{\mu}$: 
\begin{align}
\label{lattice}
\langle \C_1(x_1) \cdots \C_n(x_n) \rangle_{\mu}
=
\sum_{\{m\}}
\prod_{i=1}^{n}
\prod_{j=0}^{N}
v_{i,j}^{m_{i,j}}
\times
\tikz{1.1}{
\foreach\y in {0,...,5}{
\draw[lgray,line width=1.5pt] (1.5,0.5+\y) -- (7.5,0.5+\y);
}
\foreach\x in {0,...,6}{
\draw[lgray,line width=1.5pt] (1.5+\x,0.5) -- (1.5+\x,5.5);
}
\node[text centered] at (2,1) {$x_1$};
\node[text centered] at (4,1) {$\cdots$};
\node[text centered] at (5,1) {$\cdots$};
\node[text centered] at (7,1) {$x_1$};
\node[text centered] at (2,2) {$x_2$};
\node[text centered] at (4,2) {$\cdots$};
\node[text centered] at (5,2) {$\cdots$};
\node[text centered] at (7,2) {$x_2$};
\node[text centered] at (2,5) {$x_n$};
\node[text centered] at (4,5) {$\cdots$};
\node[text centered] at (5,5) {$\cdots$};
\node[text centered] at (7,5) {$x_n$};
\node[text centered] at (2,3.1) {$\vdots$};
\node[text centered] at (2,4.1) {$\vdots$};
\node[text centered] at (7,3.1) {$\vdots$};
\node[text centered] at (7,4.1) {$\vdots$};
\node[below] at (2,0.5) {\footnotesize$\bm{M}(0)$};
\node[above,text centered] at (4,0) {$\cdots$};
\node[above,text centered] at (5,0) {$\cdots$};
\node[below] at (7,0.5) {\footnotesize$\bm{M}(N)$};
\node[above] at (2,5.5) {\footnotesize$\bm{M}(0)+\bm{\mu}(0)$};
\node[above,text centered] at (4,5.5) {$\cdots$};
\node[above,text centered] at (5,5.5) {$\cdots$};
\node[above] at (7,5.5) {\footnotesize$\bm{M}(N)+\bm{\mu}(N)$};
\node[right] at (7.5,1) {$0$};
\node[right] at (7.5,2) {$0$};
\node[right] at (7.5,3.1) {$\vdots$};
\node[right] at (7.5,4.1) {$\vdots$};
\node[right] at (7.5,5) {$0$};
\node[left] at (1.5,1) {$1$};
\node[left] at (1.5,2) {$2$};
\node[left] at (1.5,3.1) {$\vdots$};
\node[left] at (1.5,4.1) {$\vdots$};
\node[left] at (1.5,5) {$n$};
}
\end{align}
where we have denoted $N \equiv \max(\mu)$, with the states at the bottom/top of each column given by
\begin{align*}
\bm{M}(j) = \sum_{i=1}^{n} m_{i,j} \bm{e}_i,
\qquad
\bm{\mu}(j) = \sum_{i=1}^{n} \bm{1}(\mu_i = j) \bm{e}_i,
\qquad
\forall\
0 \leq j \leq N,
\end{align*}
and where the summation is over all nonnegative integers $m_{i,j} \geq 0$, for each $1 \leq i \leq n$ and $0 \leq j \leq N$. Summation over all possible states is implied at each internal lattice edge.

Now we read the partition function \eqref{lattice} column-by-column, starting from the leftmost and working to the right. We distribute the product $\prod_{i=1}^{n} \prod_{j=0}^{N} v_{i,j}^{m_{i,j}}$ over the columns by assigning the factor $\prod_{i=1}^{n} v_{i,j}^{m_{i,j}}$ to column $j$, and for each $0 \leq j \leq N$ we compute separately the sums over $\{m_{1,j},\dots,m_{n,j}\}$ (which only affects the boundary states of column $j$). Referring to the definition \eqref{Q-def} of 
$\Lambda_{\{v_{1,j},\dots,v_{n,j}\}}(\bm{\mu}(j))$, equation \eqref{column-decomp} is nothing more than the algebraic realization of this column-by-column decomposition.
\end{proof}

The partition function \eqref{lattice} is quite useful for visualizing the structure of $\langle \C_1(x_1) \dots \C_n(x_n) \rangle_{\mu}$, but it is still not in the most convenient form. In particular, we would like to blend the summation over $\{m\}$ into our graphical definitions. With that in mind, we set up the following notation.
\begin{defn}
Let $\Lambda_{\{v_1,\dots,v_n\}}(\bm{A})$ be a column operator as defined in \eqref{Q-def}. We shall simplify the graphical representation of its components by writing
\begin{align*}
\bra{i_1,\dots,i_n}
\Lambda_{\{v_1,\dots,v_n\}}(\bm{A})
\ket{j_1,\dots,j_n}
=
\sum_{\M \in \mathbb{N}^n}
\prod_{k=1}^{n}
v_k^{m_k}
\tikz{1.2}{
\foreach\y in {1,...,4}{
\draw[lgray,line width=1.5pt] (1.5,0.5+\y) -- (2.5,0.5+\y);
}
\draw[lgray,line width=1.5pt] (1.5,0.5) -- (2.5,0.5) -- (2.5,5.5) -- (1.5,5.5) -- (1.5,0.5);
\node[text centered] at (2,1) {$x_1$};
\node[text centered] at (2,2) {$x_2$};
\node[text centered] at (2,3.1) {$\vdots$};
\node[text centered] at (2,4.1) {$\vdots$};
\node[text centered] at (2,5) {$x_n$};
\node[below] at (2,0.5) {\footnotesize$\bm{M}$};
\node[above] at (2,5.5) {\footnotesize$\bm{M}+\bm{A}$};
\node[right] at (2.5,1) {$j_1$};
\node[right] at (2.5,2) {$j_2$};
\node[right] at (2.5,3.1) {$\vdots$};
\node[right] at (2.5,4.1) {$\vdots$};
\node[right] at (2.5,5) {$j_n$};
\node[left] at (1.5,1) {$i_1$};
\node[left] at (1.5,2) {$i_2$};
\node[left] at (1.5,3.1) {$\vdots$};
\node[left] at (1.5,4.1) {$\vdots$};
\node[left] at (1.5,5) {$i_n$};
}
=
\tikz{1.2}{
\foreach\y in {1,...,4}{
\draw[lgray,line width=1.5pt] (1.5,0.5+\y) -- (2.5,0.5+\y);
}
\draw[lgray,line width=1.2pt] (2.5,0.5) -- (2.5,5.5);
\draw[double,lgray,line width=1.7pt] (2.5,5.5) -- (1.5,5.5);
\draw[lgray,line width=1.2pt] (1.5,5.5) -- (1.5,0.5);
\draw[double,lgray,line width=1.7pt] (1.5,0.5) -- (2.5,0.5);
\node[text centered] at (2,1) {$x_1$};
\node[text centered] at (2,2) {$x_2$};
\node[text centered] at (2,3.1) {$\vdots$};
\node[text centered] at (2,4.1) {$\vdots$};
\node[text centered] at (2,5) {$x_n$};
\node[below] at (2,0.5) {\phantom{\footnotesize$\bm{M}$}};
\node[above] at (2,5.5) {\footnotesize$\bm{A}$};
\node[right] at (2.5,1) {$j_1$};
\node[right] at (2.5,2) {$j_2$};
\node[right] at (2.5,3.1) {$\vdots$};
\node[right] at (2.5,4.1) {$\vdots$};
\node[right] at (2.5,5) {$j_n$};
\node[left] at (1.5,1) {$i_1$};
\node[left] at (1.5,2) {$i_2$};
\node[left] at (1.5,3.1) {$\vdots$};
\node[left] at (1.5,4.1) {$\vdots$};
\node[left] at (1.5,5) {$i_n$};
}
\end{align*}
where the dependence on $\{v_1,\dots,v_n\}$ is kept implicit in the object appearing on the right-hand side.
\end{defn}

\subsection{Explicit computation of column operator components}

In this subsection we turn to the question of explicitly computing the matrix elements of the linear operators \eqref{Q-def}. We begin with some auxiliary definitions:

\begin{defn}[Admissible vectors]
\label{defn:admiss}
Let $\mathcal{I} = (i_1,\dots,i_n) \in \{0,1,\dots,n\}^n$ and 
$\mathcal{J} = (j_1,\dots,j_n) \in \{0,1,\dots,n\}^n$ be two vectors of integers. We say that the pair $(\mathcal{I},\mathcal{J})$ is {\it admissible} provided that
\begin{align}
\label{mult-admiss}
1 \geq \text{mult}_k(\mathcal{I}) \geq \text{mult}_k(\mathcal{J}) \geq 0,
\qquad
\forall\ 1 \leq k \leq n,
\end{align}
with $\text{mult}_k(\mathcal{I})$ and $\text{mult}_k(\mathcal{J})$ denoting the multiplicities of $k$ in $\mathcal{I}$ and $\mathcal{J}$, respectively.
\end{defn}

\begin{defn}[Colour data]
\label{defn:data}
Given an admissible pair of vectors $(\mathcal{I},\mathcal{J})$, as in Definition \ref{defn:admiss}, we introduce two disjoint sets $\mathcal{P} \subset \{1,\dots,n\}$ and $\mathcal{Q} \subset \{1,\dots,n\}$, where $p \in \mathcal{P}$ if and only if $\text{mult}_p(\mathcal{I}) = 1$, 
$\text{mult}_p(\mathcal{J}) = 0$,  while $p \in \mathcal{Q}$ if and only if 
$\text{mult}_p(\mathcal{I}) = \text{mult}_p(\mathcal{J}) = 1$. We refer to $(\mathcal{P},\mathcal{Q})$ as the {\it colour data} associated to $(\mathcal{I},\mathcal{J})$.
\end{defn}

\begin{defn}[Coordinates]
\label{defn:coord}
Let $\mathcal{I} = (i_1,\dots,i_n)$ and $\mathcal{J} = (j_1,\dots,j_n)$ be admissible vectors as in Definition \ref{defn:admiss}, and $(\mathcal{P},\mathcal{Q})$ their associated colour data, as in Definition \ref{defn:data}. Introduce another two vectors
\begin{align}
\label{coord-def}
\mathcal{A} = (a_p)_{p \in \mathcal{P} \cup \mathcal{Q}},
\qquad 
\mathcal{B} = (b_p)_{p \in \mathcal{Q}},
\end{align} 
such that
\begin{align}
\label{coord-def2}
i_{a_p} = p, \quad \forall\ p \in \mathcal{P} \cup \mathcal{Q},
\qquad\qquad
j_{b_p} = p, \quad \forall\ p \in \mathcal{Q}.
\end{align}
We shall call $(\mathcal{A},\mathcal{B})$ the {\it coordinates} of 
$(\mathcal{I},\mathcal{J})$.
\end{defn}

\begin{ex}
In order to illustrate the above definitions, let us choose $n=5$, and $\mathcal{I} = (0,2,1,5,3)$, $\mathcal{J} = (2,0,0,0,5)$. These vectors clearly satisfy the requirements \eqref{mult-admiss}, and are thus admissible. We find that $\text{mult}_p(\mathcal{I}) = 1$, $\text{mult}_p(\mathcal{J}) = 0$ for $p \in \{1,3\}$, while $\text{mult}_p(\mathcal{I}) = \text{mult}_p(\mathcal{J}) = 1$ for $p \in \{2,5\}$. Hence, $\mathcal{P} = \{1,3\}$ and $\mathcal{Q} = \{2,5\}$ is the colour data associated to $(\mathcal{I},\mathcal{J})$. Finally, one finds $\mathcal{A} = (a_1,a_2,a_3,a_5) = (3,2,5,4)$ by reading the positions of $\{1,2,3,5\}$ in $\mathcal{I}$, and $\mathcal{B} = (b_2,b_5) = (1,5)$ by reading the positions of $\{2,5\}$ in $\mathcal{J}$.

To illustrate how we will make use of such quantities, consider the following example of a tower \eqref{tower} for $n=5$, $\I = (0,1,0,0,0)$, $\J = (1,1,1,0,0)$, $(i_1,\dots,i_5) = \mathcal{I}$, $(j_1,\dots,j_5) = \mathcal{J}$ (with the same $\mathcal{I}$ and $\mathcal{J}$ as above):
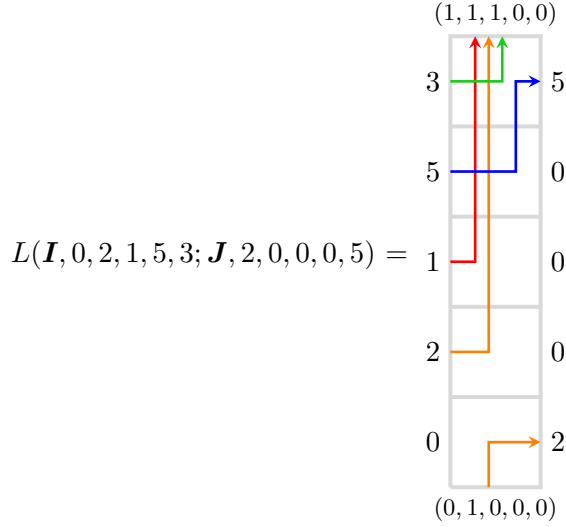
\begin{align*}
L(\I,0,2,1,5,3;\J,2,0,0,0,5)
=
\tikz{1.2}{
\foreach\y in {1,...,4}{
\draw[lgray,line width=1.5pt] (1.5,0.5+\y) -- (2.5,0.5+\y);
}
\draw[lgray,line width=1.5pt] (1.5,0.5) -- (2.5,0.5) -- (2.5,5.5) -- (1.5,5.5) -- (1.5,0.5);
\draw[red,line width=1pt,->] (1.5,3) -- (1.775,3) -- (1.775,5.5);
\draw[orange,line width=1pt,->] (1.5,2) -- (1.925,2) -- (1.925,5.5);
\draw[orange,line width=1pt,->] (1.925,0.5) -- (1.925,1) -- (2.5,1);
\draw[green,line width=1pt,->] (1.5,5) -- (2.075,5) -- (2.075,5.5);
\draw[blue,line width=1pt,->] (1.5,4) -- (2.225,4) -- (2.225,5) -- (2.5,5);
\node[below] at (2,0.5) {\footnotesize$(0,1,0,0,0)$};
\node[above] at (2,5.5) {\footnotesize$(1,1,1,0,0)$};
\node[right] at (2.5,1) {$2$};
\node[right] at (2.5,2) {$0$};
\node[right] at (2.5,3) {$0$};
\node[right] at (2.5,4) {$0$};
\node[right] at (2.5,5) {$5$};
\node[left] at (1.5,1) {$0$};
\node[left] at (1.5,2) {$2$};
\node[left] at (1.5,3) {$1$};
\node[left] at (1.5,4) {$5$};
\node[left] at (1.5,5) {$3$};
}
\end{align*}
In this picture, $\mathcal{I}$ and $\mathcal{J}$ label the states (read from bottom to top) which live on the left and right edges of the tower, respectively; the admissibility of $(\mathcal{I},\mathcal{J})$ translates into the fact that each colour $\{1,\dots,5\}$ appears at most once on the left/right of the tower, with the requirement that all colours that appear on the right must have also appeared on the left. $\mathcal{P} \cup \mathcal{Q}$ gives the set of all colours entering via the left edges of the tower, $\mathcal{Q}$ gives the set of all colours leaving via the right edges.
\end{ex}

\begin{thm}
\label{thm:comp}
Let $\mathcal{I} = (i_1,\dots,i_n)$, $\mathcal{J} = (j_1,\dots,j_n)$ be admissible vectors as in Definition \ref{defn:admiss}, and associate to them the colour data $(\mathcal{P},\mathcal{Q})$ in the same way as in Definition \ref{defn:data}. Let $(\mathcal{A},\mathcal{B})$ be the coordinates of $(\mathcal{I},\mathcal{J})$, as in equations \eqref{coord-def}--\eqref{coord-def2}. Fix $n$ complex parameters $\{v_1,\dots,v_n\}$ such that $v_r = 0$ for all $r \not\in \mathcal{P} \cup \mathcal{Q}$, and define the binary string
\begin{align*}
\bm{e}_{\mathcal{P}}
=
\sum_{p \in \mathcal{P}} 
\bm{e}_p.
\end{align*}
Under the above set of assumptions, we have
\begin{multline}
\label{components}
\bra{i_1,\dots,i_n}
\Lambda_{\{v_1,\dots,v_n\}}
\left( \bm{e}_{\mathcal{P}} \right)
\ket{j_1,\dots,j_n}
\\
=
\prod_{\substack{
p > \ell
\\
p \in \mathcal{P} \cup \mathcal{Q}
\\
\ell \in \mathcal{Q}
}}
\bm{1}(a_p \not= b_{\ell})
\prod_{p \in \mathcal{P}} t^{g(p)}
\prod_{p \in \mathcal{Q}} x_{b_p}
\prod_{p \in \mathcal{P} \cup \mathcal{Q}}
\frac{1}{1-v_p t^{f(p)}}
\prod_{\substack{
p\in \mathcal{Q} \\ a_p \not= b_p
}}
\frac{ v_p^{\bm{1}(a_p > b_p)} t^{h(p)}(1-t)}
{1-v_p t^{f(p)+1}},
\end{multline}
where we have defined the combinatorial exponents
\begin{align}
\label{exponents}
\begin{split}
f(p) 
&= 
\# \{\ell \in \mathcal{Q} : \ell < p\},
\\
g(p)
&=
\# \{\ell \in \mathcal{Q}: \ell < p, a_p < b_{\ell}\},
\\
h(p)
&=
\# \{\ell \in \mathcal{Q}: \ell < p, b_{\ell} \in (a_p,b_p)\},
\end{split}
\end{align}
with the interval $(a_p,b_p)$ appearing in $h(p)$ to be interpreted in a cyclic sense; namely, for all integers $1 \leq a,b \leq n$ we define
\begin{align*}
(a,b) 
:=
\left\{
\begin{array}{ll}
\{a+1,\dots,b-1\},
&
a < b,
\\
\{a+1,\dots,n\} \cup \{1,\dots,b-1\},
&
a > b,
\\
\varnothing,
&
a = b.
\end{array}
\right.
\end{align*}

\end{thm} 

\begin{proof}
We compute the components by calculating
\begin{align}
\label{comp-calc1}
\bra{i_1,\dots,i_n}
\Lambda_{\{v_1,\dots,v_n\}}
\left( \bm{e}_{\mathcal{P}} \right)
\ket{j_1,\dots,j_n}
=
\sum_{\M \in \mathbb{N}^n}
\prod_{k=1}^{n}
v_k^{m_k}
L(\M, i_1, \dots, i_n; 
\M+\bm{e}_{\mathcal{P}}, j_1, \dots, j_n),
\end{align}
where $L(\M, i_1, \dots, i_n; \M+\bm{e}_{\mathcal{P}}, j_1, \dots, j_n)$ is given by the tower of vertices \eqref{tower} with $\I = \M$ and $\J = \M + \bm{e}_{\mathcal{P}}$. We begin by remarking that this tower has weight zero if $i_k > j_k \geq 1$ for any $1 \leq k \leq n$, due to the vanishing of the sixth weight in \eqref{bos-weights}. This gives rise to the product of indicator functions present in \eqref{components}. In the rest of the proof, we assume that the constraints imposed by these indicators are always obeyed.

Noting that the Boltzmann weights \eqref{bos-weights} are factorized across different colours, \ie\ over the $n$ components of $\I$, we find that the $n$ sums in \eqref{comp-calc1}, over $m_1,\dots,m_n \geq 0$,  can be computed independently of each other (once $\mathcal{I}$ and $\mathcal{J}$ are fixed). This leads to the factorization
\begin{align}
\bra{i_1,\dots,i_n}
\Lambda_{\{v_1,\dots,v_n\}}
\left( \bm{e}_{\mathcal{P}} \right)
\ket{j_1,\dots,j_n}
=
\prod_{p \in \mathcal{P}}
\phi_p(v_p)
\prod_{\substack{p \in \mathcal{Q} \\ a_p = b_p}}
\chi_p(v_p)
\prod_{\substack{p \in \mathcal{Q} \\ a_p \not= b_p}}
\psi_p(v_p),
\end{align}
where we have introduced the following functions:
\begin{align}
\label{phi-sum}
\phi_p(v)
=
t^{g(p)}
\sum_{m=0}^{\infty}
\left(v t^{f(p)}\right)^m
=
\frac{t^{g(p)}}{1-v t^{f(p)}},
\end{align}
\begin{align}
\label{chi-sum}
\chi_p(v)
=
x_{b_p}
\sum_{m = 0}^{\infty}
\left(v t^{f(p)}\right)^m
=
\frac{x_{b_p}}{1-v t^{f(p)}},
\end{align}
\begin{align}
\label{psi-sum}
\psi_p(v) 
=
x_{b_p}
v^{\bm{1}(a_p > b_p)}
t^{h(p)}
\sum_{m=0}^{\infty}
(1-t^{m+1})
\left(v t^{f(p)}\right)^m
=
\frac{x_{b_p} v^{\bm{1}(a_p > b_p)} t^{h(p)}(1-t)}
{\left(1-v t^{f(p)} \right) \left(1-v t^{f(p)+1} \right)}.
\end{align}
The sums \eqref{phi-sum}--\eqref{psi-sum} can be obtained by tracing paths of a fixed colour; see Figure \ref{fig:phi-chi-psi}.
\begin{figure}
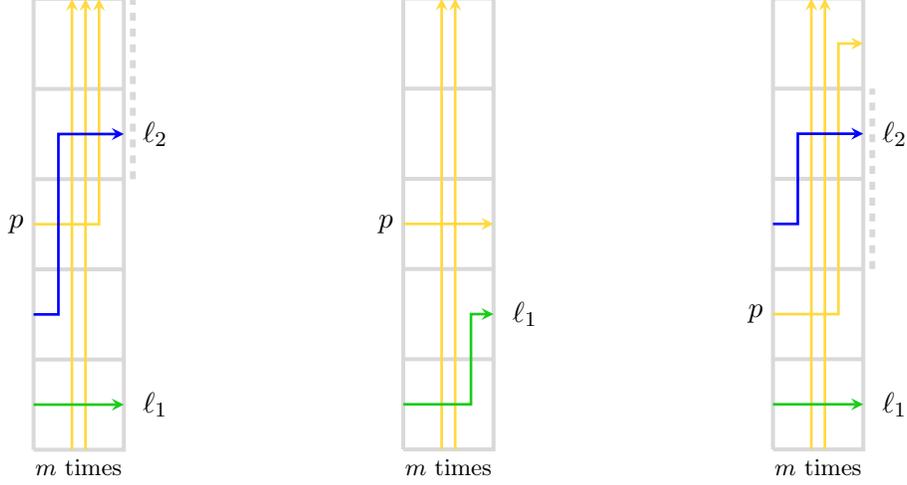

\tikz{1.2}{
\foreach\y in {1,...,4}{
\draw[lgray,line width=1.5pt] (1.5,0.5+\y) -- (2.5,0.5+\y);
}
\draw[lgray,line width=1.5pt] (1.5,0.5) -- (2.5,0.5) -- (2.5,5.5) -- (1.5,5.5) -- (1.5,0.5);
\draw[lgray,dashed,line width=2.2pt] (2.6,3.5) -- (2.6,5.5);
\draw[yellow,line width=1pt,->] (1.5,3) -- (2.225,3) -- (2.225,5.5);
\draw[yellow,line width=1pt,->] (1.925,0.5) -- (1.925,5.5);
\draw[yellow,line width=1pt,->] (2.075,0.5) -- (2.075,5.5);
\draw[blue,line width=1pt,->] (1.5,2) -- (1.775,2) -- (1.775,4) -- (2.5,4);
\draw[green,line width=1pt,->] (1.5,1) -- (2.5,1);
\node[below] at (2,0.5) {\footnotesize$m$\ {\rm times}};
\node[right] at (2.6,4) {$\ell_2$};
\node[right] at (2.6,1) {$\ell_1$};
\node[left] at (1.5,3) {$p$};
}
\qquad\qquad\qquad
\tikz{1.2}{
\foreach\y in {1,...,4}{
\draw[lgray,line width=1.5pt] (1.5,0.5+\y) -- (2.5,0.5+\y);
}
\draw[lgray,line width=1.5pt] (1.5,0.5) -- (2.5,0.5) -- (2.5,5.5) -- (1.5,5.5) -- (1.5,0.5);
\draw[yellow,line width=1pt,->] (1.5,3) -- (2.5,3);
\draw[yellow,line width=1pt,->] (1.925,0.5) -- (1.925,5.5);
\draw[yellow,line width=1pt,->] (2.075,0.5) -- (2.075,5.5);
\draw[green,line width=1pt,->] (1.5,1) -- (2.25,1) -- (2.25,2) -- (2.5,2);
\node[below] at (2,0.5) {\footnotesize$m$\ {\rm times}};
\node[right] at (2.6,2) {$\ell_1$};
\node[left] at (1.5,3) {$p$};
}
\qquad\qquad\qquad
\tikz{1.2}{
\foreach\y in {1,...,4}{
\draw[lgray,line width=1.5pt] (1.5,0.5+\y) -- (2.5,0.5+\y);
}
\draw[lgray,line width=1.5pt] (1.5,0.5) -- (2.5,0.5) -- (2.5,5.5) -- (1.5,5.5) -- (1.5,0.5);
\draw[lgray,dashed,line width=2.2pt] (2.6,2.5) -- (2.6,4.5);
\draw[yellow,line width=1pt,->] (1.5,2) -- (2.225,2) -- (2.225,5) -- (2.5,5);
\draw[yellow,line width=1pt,->] (1.925,0.5) -- (1.925,5.5);
\draw[yellow,line width=1pt,->] (2.075,0.5) -- (2.075,5.5);
\draw[blue,line width=1pt,->] (1.5,3) -- (1.775,3) -- (1.775,4) -- (2.5,4);
\draw[green,line width=1pt,->] (1.5,1) -- (2.5,1);
\node[below] at (2,0.5) {\footnotesize$m$\ {\rm times}};
\node[right] at (2.6,4) {$\ell_2$};
\node[right] at (2.6,1) {$\ell_1$};
\node[left] at (1.5,2) {$p$};
}
\caption{Left panel: a path of colour $p$ enters via the left of the tower, and circles $m$ times around it before exiting via the top. $f(p)$ counts the number of colours $\ell_1 < p$ which enter from the left of the tower and leave anywhere on its right. $g(p)$ counts the number of colours $\ell_2 < p$ that exit through one of the shaded right edges, which are situated strictly above the position where $p$ entered.
\\
{\color{white} .}
\hspace{0.1em} Middle panel: a path of colour $p$ enters via the left of the tower, and circles $m$ times around it before exiting via the right edge directly opposite its entry point. In this case, only the $f(p)$ statistic is of interest (it is defined as above).
\\
{\color{white} .}
\hspace{0.1em} Right panel: a path of colour $p$ enters via the left of the tower, circles $m$ times around it, and then exits via a right edge in a different row to its starting point. $h(p)$ counts the number of colours $\ell_2 < p$ that exit through one of the shaded right edges, which are situated strictly between the initial and final locations of path $p$.}
\label{fig:phi-chi-psi}
\end{figure}
Pictorially, the first sum \eqref{phi-sum} corresponds to a path that entered on the left of the tower, circled $m$ times over the column and then exited through its top. The second sum \eqref{chi-sum} corresponds to a path that entered on the left of the tower, circled $m$ times over the column, and finally exited on its right through the {\it same} row as where it entered. The third sum \eqref{psi-sum} corresponds to a path that entered on the left of the tower, circled $m$ times over the column, and then exited on its right but via a {\it different} row from the one where it entered. The precise form of the summand in all three cases \eqref{phi-sum}--\eqref{psi-sum} is obtained by extracting factors from the weights \eqref{bos-weights} that correspond to the colour being fixed; the summation index counts the number of circles that the path makes around the column.

The result \eqref{components} now follows.

\end{proof}

\subsection{Column rotation}
\label{ssec:rotate}

Next, we define a specific ``rotation'' operation which acts on the column operator components. The key result will be the fact that, up to a simple overall factor depending on $\{v_1,\dots,v_n\}$ and $t$, the components \eqref{components} remain invariant under this operation: 
\begin{prop}
\label{prop:rotate}
Assuming the same setup as in Theorem \ref{thm:comp}, consider the effect of rotating the edge states of the column, which is equivalent to sending $(i_1,\dots,i_{n-1},i_n) \mapsto (i_n,i_1,\dots,i_{n-1})$ and $(j_1,\dots,j_{n-1},j_n) \mapsto (j_n,j_1,\dots,j_{n-1})$ in \eqref{components}, together with the corresponding cyclic shift of the row parameters. We find that
\begin{align}
\label{rotation-const}
\frac{
\bra{i_1,\dots,i_n}
\Lambda_{\{v_1,\dots,v_n\}}(\bm{e}_{\mathcal{P}})
\ket{j_1,\dots,j_n}
}
{
\bra{i_n,i_1,\dots,i_{n-1}}
\mathfrak{s}_{n-1} \cdots \mathfrak{s}_1 \cdot
\Lambda_{\{v_1,\dots,v_n\}}
\left( \bm{e}_{\mathcal{P}} \right)
\ket{j_n,j_1,\dots,j_{n-1}}
}
=
\kappa_{\{v_1,\dots,v_n;t\}}(\mathcal{A},\mathcal{B}),
\end{align}
where the right-hand side is given by
\begin{align}
\label{rotation-const2}
\kappa_{\{v_1,\dots,v_n;t\}}(\mathcal{A},\mathcal{B})
=
\prod_{\substack{(p,\ell) \in \mathcal{P} \times \mathcal{Q} \\ p > \ell}}
t^{\bm{1}(b_{\ell}=n)-\bm{1}(a_p = n)}
\prod_{p \in \mathcal{Q}}
v_p^{\bm{1}(a_p=n)-\bm{1}(b_p = n)}.
\end{align}
\end{prop}

\begin{proof}
We can explicitly compute the ratio on the left-hand side of \eqref{rotation-const} by using the formula \eqref{components} for the column components. The main observation is that, under the proposed rotation, almost all of the factors present in \eqref{components} remain invariant.

Let us recall that $\mathcal{P}\cup\mathcal{Q}$ is the set of colours entering on the $n$ left edges of the column, while $\mathcal{Q}$ is the set of colours exiting via the $n$ right edges of the column. The vector $(a_p)_{p \in \mathcal{P} \cup \mathcal{Q}}$ records the positions where colours traverse left edges, while $(b_p)_{p \in \mathcal{Q}}$ records the positions where colours traverse right edges.

First, we examine the term $\prod_{p \in \mathcal{Q}} x_{b_p}$ in \eqref{components}, where the product ranges over all colours which exit via the right edges of the column. The required cyclic rotation of the exiting positions is achieved by the replacement $b_p \mapsto b_p+1\ (\text{mod}\ n)$, which is clearly negated by the cyclic permutation $\mathfrak{s}_{n-1} \cdots \mathfrak{s}_1$, which sends $x_i \mapsto x_{i-1\ (\text{mod}\ n)}$. Therefore this term cancels out in the ratio on the left-hand side of \eqref{rotation-const}. 

Similarly, it is clear from their definition that both of the exponents $f(p)$ and $h(p)$ in \eqref{exponents} remain unchanged under the rotation $a_p \mapsto a_p+1\ (\text{mod}\ n)$ for $p \in \mathcal{P} \cup \mathcal{Q}$ and $b_p \mapsto b_p+1\ (\text{mod}\ n)$ for $p \in \mathcal{Q}$. This allows us to cancel all factors involving those exponents in the ratio \eqref{rotation-const}, and we then read
\begin{align}
\label{rotation-const2}
\frac{
\bra{i_1,\dots,i_n}
\Lambda_{\{v_1,\dots,v_n\}}(\bm{e}_{\mathcal{P}})
\ket{j_1,\dots,j_n}
}
{
\bra{i_n,i_1,\dots,i_{n-1}}
\mathfrak{s}_{n-1} \cdots \mathfrak{s}_1 \cdot
\Lambda_{\{v_1,\dots,v_n\}}
\left( \bm{e}_{\mathcal{P}} \right)
\ket{j_n,j_1,\dots,j_{n-1}}
}
=
\frac{
\prod_{p \in \mathcal{P}} t^{g(p)} 
\prod_{p \in \mathcal{Q}} v_p^{\bm{1}(a_p > b_p)}
}
{
\prod_{p \in \mathcal{P}} t^{\tilde{g}(p)}
\prod_{p \in \mathcal{Q}} v_p^{\bm{1}(\tilde{a}_p > \tilde{b}_p)}
},
\end{align}
where we have defined $\tilde{a}_p = a_p+1\ (\text{mod}\ n)$, 
$\tilde{b}_p = b_p+1\ (\text{mod}\ n)$ and 
\begin{align}
\tilde{g}(p)
=
\#\{\ell \in \mathcal{Q} : \ell < p, \tilde{a}_p < \tilde{b}_{\ell}\}.
\end{align}
We are able to simplify the right-hand side of \eqref{rotation-const2} yet further, by noticing that it is only in the case $a_p = n$ for some $p \in \mathcal{P} \cup \mathcal{Q}$, or $b_p = n$ for some $p \in \mathcal{Q}$, where we see a discrepancy between the numerator and denominator. In those cases, one has $\tilde{a}_p=1$ or $\tilde{b}_p=1$, which can cause inequalities of the form $a_p > b_{\ell}$ or $a_{\ell} < b_p$ that previously held to now be violated. Analysing these cases yields the final result:
\begin{multline*}
\frac{
\bra{i_1,\dots,i_n}
\Lambda_{\{v_1,\dots,v_n\}}(\bm{e}_{\mathcal{P}})
\ket{j_1,\dots,j_n}
}
{
\bra{i_n,i_1,\dots,i_{n-1}}
\mathfrak{s}_{n-1} \cdots \mathfrak{s}_1 \cdot
\Lambda_{\{v_1,\dots,v_n\}}
\left( \bm{e}_{\mathcal{P}} \right)
\ket{j_n,j_1,\dots,j_{n-1}}
}
\\
=
\prod_{\substack{(p,\ell) \in \mathcal{P} \times \mathcal{Q} \\ p > \ell}}
t^{\bm{1}(b_{\ell}=n)-\bm{1}(a_p = n)}
\prod_{p \in \mathcal{Q}}
v_p^{\bm{1}(n=a_p > b_p)-\bm{1}(a_p < b_p = n)},
\end{multline*}
where we can simplify the exponent in the second product by noting that 
\begin{align*}
\bm{1}(n=a_p > b_p)-\bm{1}(a_p < b_p = n) = 
\bm{1}(a_p=n) - \bm{1}(b_p=n).
\end{align*}
\end{proof}

\begin{rmk}
In the calculations that follow, it is more useful to recast the right-hand side of \eqref{rotation-const} in terms of the colour data $\mathcal{P}$, $\mathcal{Q}$ and the vector components $i_n$, $j_n$, without making reference to the coordinates of Definition \ref{defn:coord}. We find that 
\begin{align}
\label{rot-const}
\kappa_{\{v_1,\dots,v_n;t\}}
=
\frac{
t^{\#\{ a \in \mathcal{P} : a > j_n \} \bm{1}(j_n \geq 1)}
}
{
t^{\#\{ a \in \mathcal{Q} : i_n > a\} \bm{1}(i_n \in \mathcal{P})}
}
\cdot
\frac{
(v_{i_n})^{\bm{1}(i_n \in \mathcal{Q})}
}
{
(v_{j_n})^{\bm{1}(j_n \geq 1)}
},
\end{align}
with the same $\mathcal{P}$ and $\mathcal{Q}$ as in Definition \ref{defn:data}.
\end{rmk}

\subsection{Proof of the relation \eqref{cyclic-rel}}
\label{ssec:proof}

We are now ready to return to the proof of \eqref{cyclic-rel}. Fix $N$ vectors $(k^{(a)}_1,\dots,k^{(a)}_n)$, $1 \leq a \leq N$, where $0 \leq k^{(a)}_b \leq n$ for all $1 \leq b \leq n$. We introduce a function $Z_{\rm l}$ obtained by concatenating $N$ column operators as shown below:
\begin{align}
\label{Z_l}
Z_{\rm l}\Big[(k^{(1)}_1,\dots,k^{(1)}_n),\dots,(k^{(N)}_1,\dots,k^{(N)}_n)\Big]
=
\tikz{1.1}{
\foreach\y in {0}{
\draw[lgray,line width=1.7pt,double] (1.5,0.5+\y) -- (7.5,0.5+\y);
}
\foreach\y in {5}{
\draw[lgray,line width=1.7pt,double] (1.5,0.5+\y) -- (7.5,0.5+\y);
}
\foreach\y in {1,...,4}{
\draw[lgray,line width=1.5pt] (1.5,0.5+\y) -- (7.5,0.5+\y);
}
\foreach\x in {0,...,6}{
\draw[lgray,line width=1.5pt] (1.5+\x,0.5) -- (1.5+\x,5.5);
}
\node[text centered] at (2,1) {$x_1$};
\node[text centered] at (4,1) {$\cdots$};
\node[text centered] at (5,1) {$\cdots$};
\node[text centered] at (7,1) {$x_1$};
\node[text centered] at (2,4) {$x_n$};
\node[text centered] at (4,4) {$\cdots$};
\node[text centered] at (5,4) {$\cdots$};
\node[text centered] at (7,4) {$x_n$};
\node[text centered] at (2,5) {$qx_i$};
\node[text centered] at (4,5) {$\cdots$};
\node[text centered] at (5,5) {$\cdots$};
\node[text centered] at (7,5) {$qx_i$};
\node[text centered] at (2,2.1) {$\vdots$};
\node[text centered] at (2,3.1) {$\vdots$};
\node[text centered] at (7,2.1) {$\vdots$};
\node[text centered] at (7,3.1) {$\vdots$};
\node[above] at (2,5.5) {\footnotesize$\bm{\mu}(0)$};
\node[above,text centered] at (4,5.5) {$\cdots$};
\node[above,text centered] at (5,5.5) {$\cdots$};
\node[above] at (7,5.5) {\footnotesize$\bm{\mu}(N)$};
\node[right] at (7.5,1) {$0$};
\node[right] at (7.5,2.1) {$\vdots$};
\node[right] at (7.5,3.1) {$\vdots$};
\node[right] at (7.5,4) {$0$};
\node[right] at (7.5,5) {$0$};
\node[left] at (1.5,1) {$1$};
\node[left] at (1.5,2.1) {$\vdots$};
\node[left] at (1.5,3.1) {$\vdots$};
\node[left] at (1.5,4) {$n$};
\node[left] at (1.5,5) {$i$};
\node[right, text centered] at (2.5,5) {$k^{(1)}_i$};
\node at (2.5,5) {$\bullet$};
\node[left, text centered] at (6.5,5) {$k^{(N)}_i$};
\node at (6.5,5) {$\bullet$};
\node[right, text centered] at (2.5,4) {$k^{(1)}_n$};
\node at (2.5,4) {$\bullet$};
\node[left, text centered] at (6.5,4) {$k^{(N)}_n$};
\node at (6.5,4) {$\bullet$};
\node[right, text centered] at (2.5,1) {$k^{(1)}_1$};
\node at (2.5,1) {$\bullet$};
\node[left, text centered] at (6.5,1) {$k^{(N)}_1$};
\node at (6.5,1) {$\bullet$};
}
\end{align}
where each $\bm{\mu}(j)$ is given by \eqref{comp-state}. In this picture, the left edge states are ordered sequentially from bottom to top, with the exception of the state $i$; it is omitted from the sequence $(1,\dots,n)$ and instead appears in the top row. We continue to label rows according to the state which enters on the left edge; accordingly, the top row will continue to be called the $i$-th row. For any $1 \leq b \leq n$, the internal edge states of the $b$-th row are $(k^{(1)}_b,\dots,k^{(N)}_b)$. All internal edge states are fixed, not summed.

Let us introduce a further function $Z_{\rm r}$, with a similar definition:
\begin{align}
\label{Z_r}
Z_{\rm r}\Big[
(k^{(1)}_1,\dots,k^{(1)}_n),\dots,(k^{(N)}_1,\dots,k^{(N)}_n)\Big]
=
\tikz{1.1}{
\foreach\y in {0}{
\draw[lgray,line width=1.7pt,double] (1.5,0.5+\y) -- (7.5,0.5+\y);
}
\foreach\y in {5}{
\draw[lgray,line width=1.7pt,double] (1.5,0.5+\y) -- (7.5,0.5+\y);
}
\foreach\y in {1,...,4}{
\draw[lgray,line width=1.5pt] (1.5,0.5+\y) -- (7.5,0.5+\y);
}
\foreach\x in {0,...,6}{
\draw[lgray,line width=1.5pt] (1.5+\x,0.5) -- (1.5+\x,5.5);
}
\node[text centered] at (2,1) {$x_i$};
\node[text centered] at (4,1) {$\cdots$};
\node[text centered] at (5,1) {$\cdots$};
\node[text centered] at (7,1) {$x_i$};
\node[text centered] at (2,2) {$x_1$};
\node[text centered] at (4,2) {$\cdots$};
\node[text centered] at (5,2) {$\cdots$};
\node[text centered] at (7,2) {$x_1$};
\node[text centered] at (2,5) {$x_n$};
\node[text centered] at (4,5) {$\cdots$};
\node[text centered] at (5,5) {$\cdots$};
\node[text centered] at (7,5) {$x_n$};
\node[text centered] at (2,3.1) {$\vdots$};
\node[text centered] at (2,4.1) {$\vdots$};
\node[text centered] at (7,3.1) {$\vdots$};
\node[text centered] at (7,4.1) {$\vdots$};
\node[above] at (2,5.5) {\footnotesize$\bm{\mu}(0)$};
\node[above,text centered] at (4,5.5) {$\cdots$};
\node[above,text centered] at (5,5.5) {$\cdots$};
\node[above] at (7,5.5) {\footnotesize$\bm{\mu}(N)$};
\node[right] at (7.5,1) {$0$};
\node[right] at (7.5,2) {$0$};
\node[right] at (7.5,3.1) {$\vdots$};
\node[right] at (7.5,4.1) {$\vdots$};
\node[right] at (7.5,5) {$0$};
\node[left] at (1.5,1) {$i$};
\node[left] at (1.5,2) {$1$};
\node[left] at (1.5,3.1) {$\vdots$};
\node[left] at (1.5,4.1) {$\vdots$};
\node[left] at (1.5,5) {$n$};
\node[right, text centered] at (2.5,5) {$k^{(1)}_n$};
\node at (2.5,5) {$\bullet$};
\node[left, text centered] at (6.5,5) {$k^{(N)}_n$};
\node at (6.5,5) {$\bullet$};
\node[right, text centered] at (2.5,2) {$k^{(1)}_1$};
\node at (2.5,2) {$\bullet$};
\node[left, text centered] at (6.5,2) {$k^{(N)}_1$};
\node at (6.5,2) {$\bullet$};
\node[right, text centered] at (2.5,1) {$k^{(1)}_i$};
\node at (2.5,1) {$\bullet$};
\node[left, text centered] at (6.5,1) {$k^{(N)}_i$};
\node at (6.5,1) {$\bullet$};
}
\end{align}
Once again, the left edge states are ordered sequentially from bottom to top, with the exception of state $i$; it is omitted from $(1,\dots,n)$ and transferred to the bottom row. This time, the bottom row is called the $i$-th row. As before, we assign the states $(k^{(1)}_b,\dots,k^{(N)}_b)$ to the internal edges of the $b$-th row, for each $1 \leq b \leq n$. 

We note that $Z_{\rm l}$ and $Z_{\rm r}$ are equivalent under application of the rotation operation of Section \ref{ssec:rotate} to each of the $N+1$ columns that are present (up to the variable shift $qx_i \mapsto x_i$). $Z_{\rm l}$ and $Z_{\rm r}$ also serve as refinements of the left and right-hand sides of \eqref{cyclic-rel}, respectively. One sees that
\begin{align}
\label{left-refined}
\langle \C_1(x_1) \cdots \wh{\C_i(x_i)} \cdots \C_n(x_n) \C_i(qx_i) \rangle_{\mu}
&=
\sum
Z_{\rm l}\Big[
(k^{(1)}_1,\dots,k^{(1)}_n),\dots,(k^{(N)}_1,\dots,k^{(N)}_n)\Big],
\\
\label{right-refined}
\langle \C_i(x_i) \C_1(x_1) \cdots \wh{\C_i(x_i)} \cdots \C_n(x_n) \rangle_{\mu}
&=
\sum
Z_{\rm r}\Big[
(k^{(1)}_1,\dots,k^{(1)}_n),\dots,(k^{(N)}_1,\dots,k^{(N)}_n)\Big],
\end{align}
where the sum in both cases is taken over all $0 \leq k_b^{(a)} \leq n$, $1 \leq a \leq N$ and $1 \leq b \leq n$, such that the resulting lattice configurations \eqref{Z_l} and \eqref{Z_r} have a non-vanishing weight. 

\begin{ex} Let us give an explicit example of the correspondence between \eqref{left-refined} and \eqref{right-refined}, in the case $n=5$, $i=3$ and $\mu =(0,4,4,1,5)$. In this case we have 
\begin{align*}
\bm{\mu}(0) = \bm{e}_1,\ \ 
\bm{\mu}(1) = \bm{e}_4,\ \ 
\bm{\mu}(2) = \bm{\mu}(3) = \bm{0},\ \ 
\bm{\mu}(4) = \bm{e}_2 + \bm{e}_3,\ \ 
\bm{\mu}(5) = \bm{e}_5.
\end{align*}
We extract a sample term from the sums \eqref{left-refined} and 
\eqref{right-refined}:
\begin{align*}
\tikz{1}{
\foreach\y in {0}{
\draw[lgray,line width=1.5pt,double] (1.5,0.5+\y) -- (7.5,0.5+\y);
}
\foreach\y in {5}{
\draw[lgray,line width=1.5pt,double] (1.5,0.5+\y) -- (7.5,0.5+\y);
}
\foreach\y in {1,...,4}{
\draw[lgray,line width=1.5pt] (1.5,0.5+\y) -- (7.5,0.5+\y);
}
\foreach\x in {0,...,6}{
\draw[lgray,line width=1.5pt] (1.5+\x,0.5) -- (1.5+\x,5.5);
}
\node[above] at (2,5.5) {\footnotesize$\bm{e}_1$};
\node[above] at (3,5.5) {\footnotesize$\bm{e}_4$};
\node[above] at (4,5.5) {\footnotesize$\bm{0}$};
\node[above] at (5,5.5) {\footnotesize$\bm{0}$};
\node[above] at (6,5.5) {\footnotesize$\bm{e}_2+\bm{e}_3$};
\node[above] at (7,5.5) {\footnotesize$\bm{e}_5$};
\node[right] at (7.5,1) {$0$};
\node[right] at (7.5,2) {$0$};
\node[right] at (7.5,3.1) {$0$};
\node[right] at (7.5,4.1) {$0$};
\node[right] at (7.5,5) {$0$};
\node[left] at (1,1) {$x_1$};
\node[left] at (1,2) {$x_2$};
\node[left] at (1,3) {$x_4$};
\node[left] at (1,4) {$x_5$};
\node[left] at (1,5) {$qx_3$};
\node[left] at (1.5,1) {$1$};
\node[left] at (1.5,2) {$2$};
\node[left] at (1.5,3) {$4$};
\node[left] at (1.5,4) {$5$};
\node[left] at (1.5,5) {$3$};
\draw[ultra thick,yellow,->] (1.5,5) -- (2.925,5) -- (2.925,5.5);
\draw[ultra thick,yellow,->] (2.925,0.5) -- (2.925,2) -- (4,2) -- (4,3) -- (5,3) -- (5,5.5);
\draw[ultra thick,yellow,->] (5,0.5) -- (5,1) -- (6.075,1) -- (6.075,5.5);
\draw[ultra thick,red,->] (1.5,4) -- (4,4) -- (4,5) -- (7,5) -- (7,5.5);
\draw[ultra thick,orange,->] (1.5,3) -- (3.075,3) -- (3.075,5.5);
\draw[ultra thick,green,->] (1.5,2) -- (2.075,2) -- (2.075,5.5); 
\draw[ultra thick,green,->] (2.075,0.5) -- (2.075,1) -- (5,1) -- (5,2) -- (5.925,2) -- (5.925,5.5);
\draw[ultra thick,blue,->] (1.5,1) -- (1.925,1) -- (1.925,5.5);
}
\quad\quad
\tikz{1}{
\foreach\y in {0}{
\draw[lgray,line width=1.5pt,double] (1.5,0.5+\y) -- (7.5,0.5+\y);
}
\foreach\y in {5}{
\draw[lgray,line width=1.5pt,double] (1.5,0.5+\y) -- (7.5,0.5+\y);
}
\foreach\y in {1,...,4}{
\draw[lgray,line width=1.5pt] (1.5,0.5+\y) -- (7.5,0.5+\y);
}
\foreach\x in {0,...,6}{
\draw[lgray,line width=1.5pt] (1.5+\x,0.5) -- (1.5+\x,5.5);
}
\node[above] at (2,5.5) {\footnotesize$\bm{e}_1$};
\node[above] at (3,5.5) {\footnotesize$\bm{e}_4$};
\node[above] at (4,5.5) {\footnotesize$\bm{0}$};
\node[above] at (5,5.5) {\footnotesize$\bm{0}$};
\node[above] at (6,5.5) {\footnotesize$\bm{e}_2+\bm{e}_3$};
\node[above] at (7,5.5) {\footnotesize$\bm{e}_5$};
\node[right] at (7.5,1) {$0$};
\node[right] at (7.5,2) {$0$};
\node[right] at (7.5,3.1) {$0$};
\node[right] at (7.5,4.1) {$0$};
\node[right] at (7.5,5) {$0$};
\node[left] at (1,1) {$x_3$};
\node[left] at (1,2) {$x_1$};
\node[left] at (1,3) {$x_2$};
\node[left] at (1,4) {$x_4$};
\node[left] at (1,5) {$x_5$};
\node[left] at (1.5,1) {$3$};
\node[left] at (1.5,2) {$1$};
\node[left] at (1.5,3) {$2$};
\node[left] at (1.5,4) {$4$};
\node[left] at (1.5,5) {$5$};
\draw[ultra thick,yellow,->] (1.5,1) -- (2.925,1) -- (2.925,3) -- (4,3) -- (4,4) -- (5,4) -- (5,5.5);
\draw[ultra thick,yellow,->] (5,0.5) -- (5,2) -- (6.075,2) -- (6.075,5.5);
\draw[ultra thick,red,->] (1.5,5) -- (4,5) -- (4,5.5); 
\draw[ultra thick,red,->] (4,0.5)-- (4,1) -- (7,1) -- (7,5.5);
\draw[ultra thick,orange,->] (1.5,4) -- (3.075,4) -- (3.075,5.5);
\draw[ultra thick,green,->] (1.5,3) -- (2.075,3) -- (2.075,5.5); 
\draw[ultra thick,green,->] (2.075,0.5) -- (2.075,2) -- (5,2) -- (5,3) -- (5.925,3) -- (5.925,5.5);
\draw[ultra thick,blue,->] (1.5,2) -- (1.925,2) -- (1.925,5.5);
}
\end{align*}
where the configuration on the left is one term in the sum \eqref{left-refined}, while the configuration on the right is one term in the sum \eqref{right-refined}. The configuration on the right is obtained from the configuration on the left by rotating each of the $N+1$ columns; namely, by bumping each row upwards by one step, and transferring the top row to the bottom of the lattice.
\end{ex}

\begin{prop}
\label{prop:refined-cyclic-rel}
Fix any $nN$ integers $0 \leq k_b^{(a)} \leq n$, $1 \leq a \leq N$ and $1 \leq b \leq n$, such that $Z_{\rm r}$ is not identically zero. Then $Z_{\rm l}$ is also non-vanishing, and we have
\begin{align}
\label{refined-cyclic-rel}
\frac{Z_{\rm l}\Big[
(k^{(1)}_1,\dots,k^{(1)}_n),\dots,(k^{(N)}_1,\dots,k^{(N)}_n)\Big]}
{Z_{\rm r}\Big[
(k^{(1)}_1,\dots,k^{(1)}_n),\dots,(k^{(N)}_1,\dots,k^{(N)}_n)\Big]}
=
q^{\mu_i} t^{\gamma_{i,0}(\mu)}.
\end{align}
\end{prop}

\begin{proof}
The proof is based on $N+1$ applications of \eqref{rotation-const} and \eqref{rot-const}, noting that any of the $N+1$ columns present in $Z_{\rm l}$ differs from the corresponding column in $Z_{\rm r}$ up to the rotation operation of Proposition \ref{prop:rotate}. The ratio $Z_{\rm l} / Z_{\rm r}$ will thus be of the form $q^a t^b$ for some $a,b \in \mathbb{Z}$; let us proceed to calculate these exponents.

For the purpose of calculating the right-hand side of \eqref{refined-cyclic-rel}, Proposition \ref{prop:rotate} tells us that it is sufficient to focus on the internal edge states of the $i$-th row; in particular we will be interested in those $k^{(j)}_i$ which take non-zero values, since these contribute non-trivially to the right-hand side of \eqref{rotation-const}. As we will not need to specify $k^{(j)}_b$ for $b \not= i$, we hereafter lighten the notation by writing $k^{(j)}_i = K_j$ for all $1 \leq j \leq N$. We also set $K_0 = i$ and $K_{N+1} = 0$.

By iterating \eqref{rot-const} over the $N+1$ columns and assigning a factor of $q$ to each integer $1 \leq j \leq N$ for which $K_j \geq 1$,\footnote{We obtain a power of $q$ for every step by a path in the $i$-th row of $Z_{\rm l}$, due to the $q$-shifted argument of $\mathcal{C}_i(qx_i)$ and Remark \ref{rmk:x-depend}.} we find that
\begin{align}
\label{N-columns}
\frac{Z_{\rm l}}{Z_{\rm r}}
=
\prod_{j=0}^{N}
\frac{
(v_{K_j,j})^{\bm{1}(K_j \in \mathcal{Q}_j)}
}
{
(v_{K_{j+1},j})^{\bm{1}(K_{j+1} \geq 1)}
}
\cdot
\frac{
t^{\#\{a \in \mathcal{P}_j: a>K_{j+1}\} \bm{1}(K_{j+1} \geq 1)}
}
{
t^{\#\{a \in \mathcal{Q}_j: a<K_{j}\} \bm{1}(K_j \in \mathcal{P}_j)}
}
\cdot
\prod_{j=1}^{N}
q^{\bm{1}(K_j \geq 1)},
\end{align}
where $\mathcal{P}_j$, $\mathcal{Q}_j$ are the colour data associated to the $j$-th column of $Z_{\rm l}$. Note that $\mathcal{P}_j$ is the set of colours which exit column $j$ via its top edge, while $\mathcal{Q}_j$ is the set of colours exiting via columns $j+1,\dots,N$. Hence we can write
\begin{align}
\label{PQ_j}
\mathcal{P}_j = \{a: \mu_a = j\},
\qquad
\mathcal{Q}_j = \{a: \mu_a > j\}.
\end{align}
The right-hand side of \eqref{N-columns} becomes
\begin{align}
\label{N-columns-2}
\frac{Z_{\rm l}}{Z_{\rm r}}
&=
\frac{
(v_{i,0})^{\bm{1}(i \in \mathcal{Q}_0)}
}
{t^{\#\{ a<i : \mu_a > 0 \} \bm{1}(i \in \mathcal{P}_0)}}
\prod_{j=1}^{N}
\frac{
(v_{K_j,j})^{\bm{1}(K_j \in \mathcal{Q}_j)}
}
{
(v_{K_j,j-1})^{\bm{1}(K_j \geq 1)}
}
\cdot
\frac{
t^{\#\{a>K_{j} : \mu_a=j-1\} \bm{1}(K_{j} \geq 1)}
}
{
t^{\#\{a<K_{j} : \mu_a > j\} \bm{1}(K_j \in \mathcal{P}_j)}
}
\cdot
q^{\bm{1}(K_j \geq 1)},
\end{align}
where we have used the fact that $K_0 = i$ and $K_{N+1} = 0$ to redistribute the factors in the product.  

Now let us compute the $j$-th term in the product on the right-hand side of \eqref{N-columns-2}. It is clearly sufficient to restrict our attention to $j$ such that $K_j \geq 1$, since for each $j$ such that $K_j = 0$ the factors inside the product are all equal to $1$; we tacitly assume $K_j \geq 1$ in what follows. Invoking (for the first time) the explicit form of the parameters \eqref{v_ij}, we find that
\begin{align}
\label{j-ratio}
\frac{
(v_{K_j,j})^{\bm{1}(K_j \in \mathcal{Q}_j)}
}
{
(v_{K_j,j-1})^{\bm{1}(K_j \geq 1)}
}
=
\left\{
\begin{array}{ll}
q^{-1} t^{\gamma_{K_j,j}(\mu)-\gamma_{K_j,j-1}(\mu)}, \quad
& \mu_{K_j} > j,
\\
\\
q^{-1} t^{-\gamma_{K_j,j-1}(\mu)}, \quad
& \mu_{K_j} = j,
\end{array}
\right.
\end{align}
where the case $\mu_{K_j} < j$ never appears (it would mean that the colour $K_j$ has traversed beyond column $\mu_{K_j}$, which is forbidden). From \eqref{gamma}, we can write down the $t$ exponents appearing in \eqref{j-ratio}. In the case $\mu_{K_j} > j$ (equivalent to $K_j \in \mathcal{Q}_j$), one has
\begin{align*}
\gamma_{K_j,j}(\mu)
-
\gamma_{K_j,j-1}(\mu)
&=
\#\{a > K_j : j \leq \mu_a < \mu_{K_j} \}
-
\#\{a > K_j : j-1 \leq \mu_a < \mu_{K_j} \}
\\
&=
-\#\{a > K_j : j-1 = \mu_a \}.
\end{align*}
Similarly, for $\mu_{K_j} = j$ (equivalent to $K_j \in \mathcal{P}_j$), we find that
\begin{align*}
-\gamma_{K_j,j-1}(\mu)
&=
\#\{a < K_j : \mu_a > \mu_{K_j} \}
-
\#\{a > K_j : j-1 \leq \mu_a < \mu_{K_j} \}
\\
&=
\#\{a < K_j : \mu_a > j \}
-
\#\{a > K_j : j-1 = \mu_a \}.
\end{align*}
Using these facts in \eqref{j-ratio}, we read
\begin{align}
\label{j-ratio-2}
\frac{
(v_{K_j,j})^{\bm{1}(K_j \in \mathcal{Q}_j)}
}
{
(v_{K_j,j-1})^{\bm{1}(K_j \geq 1)}
}
=
q^{-1}
\cdot
\left\{
\begin{array}{ll}
t^{-\#\{a > K_j : \mu_a = j-1 \}}, \quad\quad
& K_j \in \mathcal{Q}_j,
\\
\\
t^{\#\{a < K_j : \mu_a > j \}
- \#\{a > K_j : \mu_a = j-1 \}}, \quad\quad
& K_j \in \mathcal{P}_j.
\end{array}
\right.
\end{align}
We can now see that in either case, the $t$ exponents cancel perfectly with the remaining factors in the $j$-th term of the product \eqref{N-columns-2}, and hence
\begin{align*}
\frac{
(v_{K_j,j})^{\bm{1}(K_j \in \mathcal{Q}_j)}
}
{
(v_{K_j,j-1})^{\bm{1}(K_j \geq 1)}
}
\cdot
\frac{
t^{\#\{a>K_{j} : \mu_a=j-1\} \bm{1}(K_{j} \geq 1)}
}
{
t^{\#\{a<K_{j} : \mu_a > j\} \bm{1}(K_j \in \mathcal{P}_j)}
}
=
q^{-\bm{1}(K_j \geq 1)}.
\end{align*}
Returning to the expression \eqref{N-columns-2}, we have shown that
\begin{align}
\frac{Z_{\rm l}}{Z_{\rm r}}
&=
\frac{
(v_{i,0})^{\bm{1}(i \in \mathcal{Q}_0)}
}
{t^{\#\{ a<i : \mu_a > 0 \} \bm{1}(i \in \mathcal{P}_0)}},
\end{align}
which expresses the remarkable fact that the ratio $Z_{\rm l}/Z_{\rm r}$ does not depend on any of the values of the states $k^{(j)}_b$; not even those for which $b=i$. Finally, we check that
\begin{align}
\frac{
(v_{i,0})^{\bm{1}(i \in \mathcal{Q}_0)}
}
{t^{\#\{ a<i : \mu_a > 0 \} \bm{1}(i \in \mathcal{P}_0)}}
=
\left\{
\begin{array}{ll}
v_{i,0} = q^{\mu_i} t^{\gamma_{i,0}(\mu)}, \quad\quad
& i \in \mathcal{Q}_0,
\\
\\
t^{-\#\{ a<i : \mu_a > 0 \}}
=
q^{\mu_i} t^{\gamma_{i,0}(\mu)}, \quad\quad
& i \in \mathcal{P}_0,
\end{array}
\right.
\end{align}
where in the latter case, we have noted that $i \in \mathcal{P}_0$ implies $\mu_i = 0$.

\end{proof}

\begin{cor}
Equation \eqref{cyclic-rel} holds.
\end{cor}

\begin{proof}
Since we have a one-to-one pairing of each non-vanishing term in the sum \eqref{left-refined} with a corresponding term in \eqref{right-refined}, and have shown that the two terms have the correct proportionality constant \eqref{refined-cyclic-rel} irrespective of the values of the internal states $k^{(j)}_b$, \eqref{cyclic-rel} follows immediately. 

Note that, by repeating the above arguments with ``permuted'' boundary conditions at the left edges of the lattice, we prove Proposition \ref{prop2-intro} from the introduction.
\end{proof}

\subsection{Fixing the normalization}
\label{ssec:norm}

So far we have determined that $\langle \C_1(x_1) \cdots \C_n(x_n) \rangle_{\mu}$ satisfies the eigenvalue equation
\begin{align}
\label{eig-Yi-CCC}
Y_i \cdot \langle \C_1(x_1) \cdots \C_n(x_n) \rangle_{\mu}
=
y_i(\mu;q,t) \langle \C_1(x_1) \cdots \C_n(x_n) \rangle_{\mu},
\qquad
\forall\ 1 \leq i \leq n.
\end{align}
It remains only for us to resolve the normalization of 
$\langle \C_1(x_1) \cdots \C_n(x_n) \rangle_{\mu}$ by introducing a factor $\Omega_{\mu}(q,t)$ such that
\begin{align}
\label{Omega-compute}
{\rm Coeff}\Big[ 
\Omega_{\mu}(q,t)
\langle \C_1(x_1) \cdots \C_n(x_n) \rangle_{\mu};
x^{\mu}
\Big]
=
1,
\qquad
x^{\mu}
=
\prod_{i=1}^{n} x_i^{\mu_i}.
\end{align}
In this section we will show that $\Omega_{\mu}(q,t)$ takes the form \eqref{Omega}.

It is a straightforward task to compute the coefficient of the monomial $x^{\mu}$ in $\langle \C_1(x_1) \cdots \C_n(x_n) \rangle_{\mu}$. In the process of computing this coefficient we also verify that $ \langle \C_1(x_1) \cdots \C_n(x_n) \rangle_{\mu}$ does not vanish identically (identical zero would serve as the trivial solution of \eqref{eig-Yi-CCC}). We make the following claim: 
\begin{prop}
The quantity ${\rm Coeff}[\langle \C_1(x_1) \cdots \C_n(x_n) \rangle_{\mu}; x^{\mu}]$ is non-vanishing and is given by the weight of the unique lattice configuration of the form
\begin{align}
\label{frozen-xmu}
\sum_{\{m\}}
\prod_{i=1}^{n}
\prod_{j=0}^{N}
v_{i,j}^{m_{i,j}}
\times
\tikz{1.1}{
\foreach\y in {0,...,5}{
\draw[lgray,line width=1.5pt] (1.5,0.5+\y) -- (7.5,0.5+\y);
}
\foreach\x in {0,...,6}{
\draw[lgray,line width=1.5pt] (1.5+\x,0.5) -- (1.5+\x,5.5);
}
\node[below] at (2,0.5) {\footnotesize$\bm{M}(0)$};
\node[above,text centered] at (4,0) {$\cdots$};
\node[above,text centered] at (5,0) {$\cdots$};
\node[below] at (7,0.5) {\footnotesize$\bm{M}(N)$};
\node[above] at (2,5.5) {\footnotesize$\bm{M}(0)+\bm{\mu}(0)$};
\node[above,text centered] at (4,5.5) {$\cdots$};
\node[above,text centered] at (5,5.5) {$\cdots$};
\node[above] at (7,5.5) {\footnotesize$\bm{M}(N)+\bm{\mu}(N)$};
\node[right] at (7.5,1) {$0$};
\node[right] at (7.5,2) {$0$};
\node[right] at (7.5,3.1) {$\vdots$};
\node[right] at (7.5,4.1) {$\vdots$};
\node[right] at (7.5,5) {$0$};
\node[left] at (1,1) {$x_1$};
\node[left] at (1,2) {$x_2$};
\node[left] at (1,3.1) {$\vdots$};
\node[left] at (1,4.1) {$\vdots$};
\node[left] at (1,5) {$x_n$};
\node[left] at (1.5,1) {$1$};
\node[left] at (1.5,2) {$2$};
\node[left] at (1.5,3.1) {$\vdots$};
\node[left] at (1.5,4.1) {$\vdots$};
\node[left] at (1.5,5) {$n$};
\draw[ultra thick,red,->] (1.5,5) -- (2.5,5) -- (4,5) -- (4,5.5);
\draw[ultra thick,orange,->] (1.5,4) -- (2,4) -- (2,5.5);
\draw[ultra thick,yellow,->] (1.5,3) -- (2.5,3) -- (5.1,3) -- (5.1,5.5);
\draw[ultra thick,green,->] (1.5,2) -- (7,2) -- (7,5.5); 
\draw[ultra thick,blue,->] (1.5,1) -- (2.5,1) -- (4.9,1) -- (4.9,5.5);
}
\end{align}
in which path $i$ travels straight for $\mu_i$ consecutive horizontal steps, before turning and exiting at the top of the $\mu_i$-th column, for all 
$1 \leq i \leq n$.
\end{prop}

\begin{proof}
In what follows let $N = \max_{1 \leq i \leq n}(\mu_i)$. Let $i$ be the largest integer such that $\mu_i = N$ (there may be other integers $j$ such that $\mu_j = N$, and we assume that $i>j$ for all such $j$). Then there is no other path which travels further, horizontally, than the path of colour $i$ does in order to reach its final destination in column $N$. 

Given that we wish to calculate ${\rm Coeff}[\langle \C_1(x_1) \cdots \C_n(x_n) \rangle_{\mu}; x^{\mu}]$, we restrict our attention to lattice configurations which give rise to a factor of $x_i^{\mu_i} = x_i^N$. We claim that the only possible way in which this factor is obtained is when the path of colour $i$ travels horizontally for $N$ consecutive steps before turning into the $N$-th column. In other words, isolating the $i$-th row of the lattice,
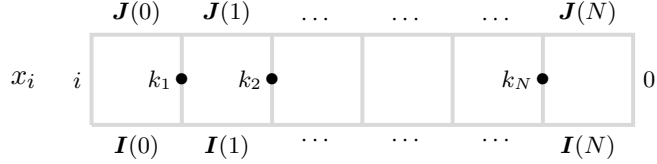
\begin{align*}
\tikz{1.2}{
\draw[lgray,line width=1.5pt] (0.5,-0.5) -- (6.5,-0.5) -- (6.5,0.5) -- (0.5,0.5) -- (0.5,-0.5);
\foreach\x in {1,...,5}{
\draw[lgray,line width=1.5pt] (0.5+\x,-0.5) -- (0.5+\x,0.5);
}
\node[left] at (0,0) {$x_i$};
\node[left] at (1.5,0) {\fs $k_1$}; \node at (1.5,0) {$\bullet$};
\node[left] at (2.5,0) {\fs $k_2$}; \node at (2.5,0) {$\bullet$};
\node[left] at (5.5,0) {\fs $k_N$}; \node at (5.5,0) {$\bullet$};
\node[left] at (0.5,0) {\fs $i$};\node[right] at (6.5,0) {\fs $0$};
\node[below] at (6,-0.5) {\fs $\I(N)$};\node[above] at (6,0.5) {\fs $\J(N)$};
\node[below] at (5,-0.5) {\fs $\cdots$};\node[above] at (5,0.5) {\fs $\cdots$};
\node[below] at (4,-0.5) {\fs $\cdots$};\node[above] at (4,0.5) {\fs $\cdots$};
\node[below] at (3,-0.5) {\fs $\cdots$};\node[above] at (3,0.5) {\fs $\cdots$};
\node[below] at (2,-0.5) {\fs $\I(1)$};\node[above] at (2,0.5) {\fs $\J(1)$};
\node[below] at (1,-0.5) {\fs $\I(0)$};\node[above] at (1,0.5) {\fs $\J(0)$};
}
\end{align*} 
and labelling its internal edges as $k_1,\dots,k_N$ as shown above, we claim that the only possible contribution to $x_i^N$ comes in the case $k_1 = \cdots = k_N = i$. To prove this, we will assume that some other lattice configuration exists that gives rise to the factor $x_i^N$, in which not all $k_1,\dots,k_N$ are equal to $i$, and show that we obtain a contradiction.

To begin, we note that all $k_1,\dots,k_N$ must be nonzero, or else we cannot recover the required degree in $x_i$. Assume that for some integer $p_1$ one has a lattice configuration in which $k_{p_1} = i$ and 
$k_{p_1+1} = j_1$, where $i < j_1$ (note that one cannot have $i > j_1 \geq 1$, since the corresponding Boltzmann weight in the table \eqref{bos-weights} vanishes). Now the path of colour $j_1$ must turn out of the $i$-th row somewhere before the $N$-th column, since by assumption $i$ is the largest colour such that $\mu_i = N$. Let the column where this turning happens be labelled $p_2$. Then we must have $k_{p_2} = j_1$ and $k_{p_2+1} = j_2$, where $j_1 < j_2$. We again find that the path of colour $j_2$ must turn out of the $i$-th row somewhere before the $N$-th column. One may now iterate this reasoning, ultimately arriving at the contradiction that $k_N$ will be forced to assume a value greater than $i$, which is impossible.

We have thus proved the required statement about the path of colour $i$, completely constraining its motion to the $i$-th row, and it plays no role in configurations of the remaining rows. Then let us define 
$\hat{\mu} = (\hat\mu_1,\dots,\hat\mu_{n-1}) = (\mu_1,\dots,\mu_{i-1},\mu_{i+1},\dots,\mu_n)$, \ie\ we omit the $i$-th part from $\mu$. One may now apply a similar reasoning to the largest integer $j$ such that 
$\mu_j = \max_{1 \leq i \leq n-1}(\hat\mu_i)$, arriving at the same conclusion for the path of colour $j$; namely, it is forced to take 
$\max_{1 \leq i \leq n-1}(\hat\mu_i)$ consecutive horizontal steps before turning into the $\mu_j$-th column, where it terminates. By iterating this argument over each of the colours, we arrive at the statement of the proposition.
\end{proof}

Now let us compute the weight of the lattice configuration shown in \eqref{frozen-xmu}. To do this, we use the formula \eqref{components} to calculate the weight of each of the columns in \eqref{frozen-xmu}, and multiply them together. Each of the columns that appear in the frozen configuration \eqref{frozen-xmu} take a much simpler form than the generic column components \eqref{components}; we denote the weight of the $j$-th column by $F_j$, and remark that it is a function only of the colours $\{a: \mu_a = j\}$ which exit via its top and those $\{b: \mu_b > j\}$ which horizontally traverse it. We are able to write
\begin{align}
\label{something}
{\rm Coeff}\Big[ 
\langle \C_1(x_1) \cdots \C_n(x_n) \rangle_{\mu};
x^{\mu}
\Big]
=
\prod_{j=0}^{N}
F_j(\mathcal{P}_j ; \mathcal{Q}_j),
\end{align}
with $\mathcal{P}_j$ and $\mathcal{Q}_j$ as given by \eqref{PQ_j}, and where we have defined the function
\begin{align}
\label{coeff-prod}
F_j(\mathcal{P} ; \mathcal{Q})
=
\prod_{p \in \mathcal{P} \cup \mathcal{Q}}
\frac{1}{1-v_{p,j} t^{f(p)}}
=
\prod_{p \in \mathcal{P} \cup \mathcal{Q}}
\frac{1}{1-v_{p,j} t^{\#\{\ell \in \mathcal{Q} : \ell < p\}}}.
\end{align}
Rewriting the product \eqref{coeff-prod} more explicitly, we have
\begin{align}
\label{coeff-prod2}
{\rm Coeff}\Big[ 
\langle \C_1(x_1) \cdots \C_n(x_n) \rangle_{\mu};
x^{\mu}
\Big]
=
\prod_{j=0}^{N}
\prod_{i: \mu_i \geq j}
\frac{1}{1-v_{i,j} t^{\#\{\ell < i : j < \mu_{\ell}\}}}
=
\prod_{i=1}^{n}
\prod_{j=0}^{\mu_i}
\frac{1}{1-v_{i,j} t^{\#\{\ell < i : j < \mu_{\ell}\}}}.
\end{align}
To conclude our calculation, we recall the explicit form \eqref{v_ij} of the parameters $v_{i,j}$. First, the indicator function $\bm{1}(\mu_i > j)$ present in $v_{i,j}$ allows us to replace the product on the right-hand side of \eqref{coeff-prod2} with $\prod_{i=1}^{n} \prod_{j=0}^{\mu_i-1}$. Second, the parameters $v_{i,j}$ contribute a factor of $t^{\gamma_{i,j}(\mu)}$ with $\gamma_{i,j}(\mu)$ given by \eqref{gamma}, which can be combined with the $t^{\#\{\ell < i : j < \mu_{\ell}\}}$ term appearing in \eqref{coeff-prod2}. Collecting the exponents proceeds as follows:
\begin{align*}
\#\{\ell < i : j < \mu_{\ell}\}
+
\gamma_{i,j}(\mu)
&=
\#\{\ell < i : j < \mu_\ell\}
-
\#\{k < i : \mu_k > \mu_i\}
+
\#\{k>i : j \leq \mu_k < \mu_i\}
\\
&=
\#\{k < i : j < \mu_k \leq \mu_i\}
+
\#\{k > i : j \leq \mu_k < \mu_i\}
\\
&=
\#\{k < i : \mu_k = \mu_i\}
+
\#\{k \not=i : j < \mu_k < \mu_i\}
+
\#\{k > i : j = \mu_k \}
\\
&=
\alpha_{i,j}(\mu),
\end{align*}
with $\alpha_{i,j}(\mu)$ given by \eqref{alpha}. Putting everything together, we have shown that
\begin{align}
\label{coeff-prod3}
{\rm Coeff}\Big[ 
\langle \C_1(x_1) \cdots \C_n(x_n) \rangle_{\mu};
x^{\mu}
\Big]
=
\prod_{i=1}^{n}
\prod_{j=0}^{\mu_i-1}
\frac{1}{1-q^{\mu_i-j} t^{\alpha_{i,j}(\mu)}}
=
\frac{1}{\Omega_{\mu}(q,t)},
\end{align}
with $\Omega_{\mu}(q,t)$ given by \eqref{Omega}.

\subsection{Completing the proof of Theorem \ref{thm:f-formula}}

We have shown that our Ansatz \eqref{f-ansatz} satisfies both the eigenvalue equation \eqref{eig-Yi-CCC} and the normalization constraint \eqref{Omega-compute}, with $\Omega_{\mu}(q,t)$ given by \eqref{Omega}. This lands us in the territory of Remark \ref{rmk:alt-def}, where we defined the nonsymmetric Macdonald polynomials $f_{\mu}$. The proof of Theorem \ref{thm:f-formula} is complete.

\section{Connection with the formula of Haglund--Haiman--Loehr}

One of the nice consequences of the formula \eqref{f-ansatz} is that it completely elaborates the monomial structure of $f_{\mu}(x_1,\dots,x_n;q,t)$; namely, it allows one to calculate the coefficient of the monomial $x^{\nu}$ within this polynomial, for any composition $\nu$ such that $|\nu| = |\mu|$. In particular, we will show that with a small amount of work, \eqref{f-ansatz} leads to a combinatorial formula for nonsymmetric Macdonald polynomials of the same type as that obtained in \cite{HaglundHL2}.

In Sections \ref{ssec:diag}--\ref{ssec:inv} we gather some necessary combinatorial definitions. In Section \ref{ssec:formula} we state a combinatorial formula for $f_{\mu}(x_1,\dots,x_n;q,t)$, which can be seen to match with Theorem 3.5.1 of \cite{HaglundHL2} up to sending $(q,t) \mapsto (q^{-1},t^{-1})$ and reorganization of the $t$ exponents within the summand. Sections \ref{ssec:fill-config}--\ref{ssec:weight-match} will be devoted to building a bridge between the formula \eqref{hhl-formula} and the partition function formalism of the present paper.

\subsection{Diagrams, leg and arm lengths, attacking squares}
\label{ssec:diag}

Fix a composition $\mu = (\mu_1,\dots,\mu_n)$ and associate to it a diagram ${\rm dg}(\mu)$, which is a collection of $n$ columns of squares, where the $i$-th column (counted from the left) consists of $\mu_i$ squares. More precisely, the diagram of $\mu$ is the following set of coordinates:
\begin{align*}
{\rm dg}(\mu)
=
\{(i,j) : 1 \leq i \leq n,\ 1 \leq j \leq \mu_i \},
\end{align*}
or in other words, the square $s$ with coordinates $(i,j)$ is a member of ${\rm dg}(\mu)$ if and only if $1 \leq j \leq \mu_i$. ${\rm dg}(\mu)$ extends the notion of a Young diagram, written in French notation, to the case of compositions.

The extended diagram $\widehat{\rm dg}(\mu)$ takes the same form as above, up to abridging an extra row of squares along the bottom edge of 
${\rm dg}(\mu)$:
\begin{align}
\label{extend-diag}
\widehat{\rm dg}(\mu)
=
\{(i,j) : 1 \leq i \leq n,\ 0 \leq j \leq \mu_i \}.
\end{align}
Note that the square $(i,0)$, $1 \leq i \leq n$, is always present in 
$\widehat{\rm dg}(\mu)$.

\begin{defn}[Legs]
\label{def:leg}
Let $s = (i,j)$ be a square in ${\rm dg}(\mu)$. We define its {\it leg length} $l(s)$ as follows:
\begin{align}
l(s) = \mu_i-j.
\end{align}
\end{defn}

\begin{defn}[Arms]
\label{def:arm}
Let $s = (i,j)$ be a square in ${\rm dg}(\mu)$. We define its {\it arm length} $a(s)$ as follows:
\begin{align}
a(s) = \alpha_{i,j-1}(\mu),
\end{align}
where the right-hand side is given by \eqref{alpha}. Note that one has
\begin{align*}
a(s) 
&= 
\#\{k < i : j \leq \mu_k \leq \mu_i\} 
+
\#\{k > i: j-1 \leq \mu_k < \mu_i\}
=
|\text{arm}^{\text{left}}(s)|
+
|\text{arm}^{\text{right}}(s)|,
\end{align*}
with precisely the same definition of the sets $\text{arm}^{\text{left}}(s)$ and $\text{arm}^{\text{right}}(s)$ as in \cite[Section 2.3]{HaglundHL2}.
\end{defn}

\begin{defn}[Attacking squares]
Given an extended diagram of the form \eqref{extend-diag}, we say that two squares $(i,j) \in \widehat{\rm dg}(\mu)$ and $(i',j') \in \widehat{\rm dg}(\mu)$ {\it attack}, and write $(i,j) \vdash (i',j')$, if {\bf 1.} $i<i'$ and 
$j=j'$, or {\bf 2.} $i< i'$ and $j=j'+1$.
\end{defn}

\subsection{Fillings and the non-attacking property}
\label{ssec:fill}

\begin{defn}[Fillings]
A filling $\sigma$ of shape $\mu$, written $(\mu;\sigma)$, assigns a positive integer in $\{1,\dots,n\}$ to each of the squares in $\widehat{\rm dg}(\mu)$. Namely, it is a set of integer triples of the following form:
\begin{align}
\label{extend-fill}
(\mu;\sigma)
=
\{ 
(i,j,\sigma_{i,j}) 
: 
1 \leq i \leq n,\
0 \leq j \leq \mu_i,\ 
1 \leq \sigma_{i,j} \leq n \},
\end{align}
where by agreement $\sigma_{i,0} = i$ for all $1 \leq i \leq n$.\footnote{Note that in \cite{HaglundHL2} such fillings were called {\it extended}. Here we avoid this qualification and always agree that our fillings apply to extended diagrams.}
\end{defn}

\begin{defn}[Non-attacking fillings]
\label{def:non-attack}
A filling $(\mu;\sigma)$ is said to be {\it non-attacking} if for all squares $(i,j) \in \widehat{\rm dg}(\mu)$ and $(i',j') \in \widehat{\rm dg}(\mu)$ such that $(i,j) \vdash (i',j')$, one has $\sigma_{i,j} \not= \sigma_{i',j'}$. Denote the set of all non-attacking fillings $(\mu;\sigma)$ as follows:
\begin{align}
\mathfrak{S}(\mu)
=
\{
(\mu;\sigma):
\sigma\ \text{non-attacking}
\}.
\end{align}
\end{defn}

\subsection{Ascents and descents}
\label{ssec:descent}

\begin{defn}[Ascent and descent sets]
\label{def:ascent}
Fixing a filling $(\mu; \sigma)$ and a square $s = (i,j) \in {\rm dg}(\mu)$, we say that $s$ is a {\it descent} if $\sigma_{i,j} > \sigma_{i,j-1}$. Similarly, 
$s = (i,j) \in {\rm dg}(\mu)$ is an {\it ascent} if $\sigma_{i,j} < \sigma_{i,j-1}$. We define the descent (resp. ascent) set of a filling as follows:
\begin{align*}
\mathcal{D}(\mu;\sigma)
\equiv
\mathcal{D}(\sigma)
=
\{(i,j) \in {\rm dg}(\mu) : \sigma_{i,j} > \sigma_{i,j-1}\},
\\
\mathcal{A}(\mu;\sigma)
\equiv
\mathcal{A}(\sigma)
=
\{(i,j) \in {\rm dg}(\mu) : \sigma_{i,j} < \sigma_{i,j-1}\}.
\end{align*}

\end{defn}

\subsection{Ordered triples}
\label{ssec:inv}

\begin{defn}[Ordered triples]
Given a filling $(\mu; \sigma)$, an {\it ordered triple}\footnote{The triples that we define are not supposed to match with the {\it inversion triples} defined in \cite{HaglundHL2}; accordingly, we have given them a different name.} is a collection of three squares $(i,j) \in {\rm dg}(\mu)$, $(i',j-1) \in \widehat{\rm dg}(\mu)$ and $(i',j)$ with $i < i'$, whose fillings satisfy one of the strings of inequalities
\begin{align}
\label{pos-inv}
\sigma_{i',j} > \sigma_{i,j} > \sigma_{i',j-1},
\\
\label{neg-inv}
\sigma_{i',j} < \sigma_{i,j} < \sigma_{i',j-1},
\end{align}
where by agreement $\sigma_{i',j} = \infty$ if $(i',j) \not\in {\rm dg}(\mu)$. Triples of the first type \eqref{pos-inv} will be called {\it positive}; similarly, triples of the second type \eqref{neg-inv} we term {\it negative}. We denote the total number of positive (resp. negative) ordered triples in $(\mu;\sigma)$ by ${\rm ord}_{\pm}(\mu;\sigma) \equiv {\rm ord}_{\pm}(\sigma)$, and write their difference as
\begin{align*}
\Delta(\sigma) = {\rm ord}_{+}(\sigma) - {\rm ord}_{-}(\sigma). 
\end{align*}

\end{defn}

\subsection{Combinatorial formula for $f_{\mu}(x_1,\dots,x_n;q,t)$}
\label{ssec:formula}

\begin{prop}
\label{prop:match}
The nonsymmetric Macdonald polynomials are given by the following formula:
\begin{align}
\label{hhl-formula}
f_{\mu}(x_1,\dots,x_n;q,t)
=
\sum_{\sigma \in \mathfrak{S}(\mu)}
x^{\sigma}
t^{\Delta(\sigma)}
\prod_{s \in \mathcal{D}(\sigma)}
\frac{1-t}{1-q^{l(s)+1} t^{a(s)+1}}
\prod_{s \in \mathcal{A}(\sigma)}
\frac{q^{l(s)+1} t^{a(s)} (1-t)}{1-q^{l(s)+1} t^{a(s)+1}},
\end{align}
where we have defined $x^{\sigma} =
\prod_{i=1}^{n} \prod_{j=1}^{\mu_i} x_{\sigma_{i,j}}$.

\end{prop}

\subsection{Bijection between lattice configurations and non-attacking fillings}
\label{ssec:fill-config}

The proof of Proposition \ref{prop:match} proceeds by finding a bijective mapping between the configurations appearing in the matrix product formula for $f_{\mu}(x_1,\dots,x_n;q,t)$, and non-attacking fillings. 

A lattice configuration $\xi$ is a collection of integer triples of the form
\begin{align}
\label{eta}
\xi
=
\left\{ \left(i,j,k^{(j)}_i\right): 
1 \leq i \leq n,\ 
0 \leq j \leq N,\
0 \leq k^{(j)}_i \leq n
\right\},
\end{align}
where $k^{(j)}_i$ is the colour that passes through the vertical edge situated at lateral coordinate $j$ and height coordinate $i$, as in Figure \ref{fig:xi}.
\begin{figure}
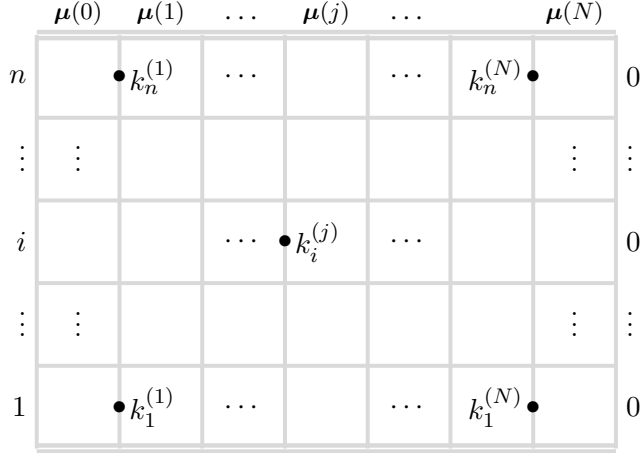

\begin{align*}
\tikz{1.1}{
\foreach\y in {0}{
\draw[lgray,line width=1.7pt,double] (1.5,0.5+\y) -- (8.5,0.5+\y);
}
\foreach\y in {5}{
\draw[lgray,line width=1.7pt,double] (1.5,0.5+\y) -- (8.5,0.5+\y);
}
\foreach\y in {1,...,4}{
\draw[lgray,line width=1.5pt] (1.5,0.5+\y) -- (8.5,0.5+\y);
}
\foreach\x in {0,...,7}{
\draw[lgray,line width=1.5pt] (1.5+\x,0.5) -- (1.5+\x,5.5);
}
\node[text centered] at (4,3) {$\cdots$};
\node[text centered] at (6,3) {$\cdots$};
\node[text centered] at (4,1) {$\cdots$};
\node[text centered] at (6,1) {$\cdots$};
\node[text centered] at (4,5) {$\cdots$};
\node[text centered] at (6,5) {$\cdots$};
\node[text centered] at (2,2.1) {$\vdots$};
\node[text centered] at (2,4.1) {$\vdots$};
\node[text centered] at (8,2.1) {$\vdots$};
\node[text centered] at (8,4.1) {$\vdots$};
\node[above] at (2,5.5) {\footnotesize$\bm{\mu}(0)$};
\node[above] at (3,5.5) {\footnotesize$\bm{\mu}(1)$};
\node[above,text centered] at (4,5.5) {$\cdots$};
\node[above] at (5,5.5) {\footnotesize$\bm{\mu}(j)$};
\node[above,text centered] at (6,5.5) {$\cdots$};
\node[above] at (8,5.5) {\footnotesize$\bm{\mu}(N)$};
\node[right] at (8.5,1) {$0$};
\node[right] at (8.5,3) {$0$};
\node[right] at (8.5,2.1) {$\vdots$};
\node[right] at (8.5,4.1) {$\vdots$};
\node[right] at (8.5,5) {$0$};
\node[left] at (1.5,3) {$i$};
\node[left] at (1.5,1) {$1$};
\node[left] at (1.5,2.1) {$\vdots$};
\node[left] at (1.5,4.1) {$\vdots$};
\node[left] at (1.5,5) {$n$};
\node[right, text centered] at (2.5,5) {$k^{(1)}_n$};
\node at (2.5,5) {$\bullet$};
\node[left, text centered] at (7.5,5) {$k^{(N)}_n$};
\node at (7.5,5) {$\bullet$};
\node[right, text centered] at (2.5,1) {$k^{(1)}_1$};
\node at (2.5,1) {$\bullet$};
\node[left, text centered] at (7.5,1) {$k^{(N)}_1$};
\node at (7.5,1) {$\bullet$};
\node[right, text centered] at (4.5,3) {$k^{(j)}_i$};
\node at (4.5,3) {$\bullet$};
}
\end{align*}
\caption{A lattice configuration $\xi$ on a cylinder. $\xi$ is $\mu$-legal provided the vertical edge states $k^{(j)}_i$ obey the constraints {\bf 1}--{\bf 3}.}
\label{fig:xi}
\end{figure}
We say that a configuration \eqref{eta} is $\mu$-legal, and write $\xi \in \Xi(\mu)$, provided that the following three types of constraints are obeyed: 

\begin{enumerate}[{\bf 1.}]
\item
For any colour $1 \leq a \leq n$ and lateral coordinate $0 \leq j \leq \mu_a$, there exists an $i$ such that $k^{(j)}_i = a$; while for $j > \mu_a$, there exists no $i$ such that $k^{(j)}_i = a$. This constraint comes from the continuity of path $a$ as it traverses the cylinder, and the fact that it terminates in column $\mu_a$ of the lattice. 

\item
One has $k^{(0)}_i = i$ for all $1 \leq i \leq n$. This comes from the boundary condition assigned to the left edges of the lattice.  

\item
For any pair of neighbouring colours $k^{(j)}_i \geq 1$ and $k^{(j+1)}_i \geq 1$, the inequality $k^{(j)}_i > k^{(j+1)}_i$ is forbidden. This restriction comes from the vanishing of the sixth Boltzmann weight in the table \eqref{bos-weights}.
\end{enumerate}

Given a configuration $\xi \in \Xi(\mu)$, let us now find a corresponding filling $\sigma \in \mathfrak{S}(\mu)$, \ie\ we seek a map $\mathfrak{M}: \Xi(\mu) \rightarrow \mathfrak{S}(\mu)$ such that $\mathfrak{M}(\xi) = \sigma$. We will specify this map algorithmically. From the $\mu$-legal configuration $\xi$ given by \eqref{eta} we construct $n$ {\it colour sets}:
\begin{align*}
\mathcal{S}_a 
=
\left\{(i,j) : 1 \leq i \leq n,\ 0 \leq j \leq N,\ k^{(j)}_i = a \right\},
\qquad
1 \leq a \leq n.
\end{align*}
In view of constraint {\bf 1} listed above, one has 
$|\mathcal{S}_a| = \mu_a + 1$ for all $1\leq a \leq n$, and $\mathcal{S}_a$ takes the explicit form
\begin{align}
\label{colour-sets-sigma}
\mathcal{S}_a
=
\Big\{
(\sigma_{a,0},0),
(\sigma_{a,1},1),
\dots,
(\sigma_{a,\mu_a},\mu_a)
\Big\},
\qquad
1\leq a \leq n,
\end{align}
where $\sigma_{a,0} = a$ (constraint {\bf 2}) and  $1 \leq \sigma_{a,j} \leq n$ for all $1 \leq j \leq \mu_a$. Hence, the integers $\sigma_{a,j}$ can be used to specify a filling $(\mu;\sigma)$ of the extended diagram $\widehat{\rm dg}(\mu)$, in the same sense as in \eqref{extend-fill}; namely, one uses the members of $\mathcal{S}_a$ to fill the boxes of the $a$-th column of $\widehat{\rm dg}(\mu)$. 

It remains to check that the filling $(\mu;\sigma)$ inherited in this way is non-attacking. For this, notice that the integers obtained by the prescription \eqref{colour-sets-sigma} must satisfy 
$\sigma_{a,j} \not= \sigma_{b,j}$ for all $a \not= b$ (were this not the case, we would have $k^{(j)}_i =a$ and $k^{(j)}_i = b$ for $i = \sigma_{a,j} = \sigma_{b,j}$, which leads to the contradiction that a vertical edge in Figure \ref{fig:xi} is occupied by two different colours). In a similar vein, the integers coming from \eqref{colour-sets-sigma} must satisfy $\sigma_{a,j+1} \not= \sigma_{b,j}$ for all $a<b $ (were this not the case, we would have $k^{(j+1)}_i = a < b = k^{(j)}_i$ for $i = \sigma_{a,j+1} = \sigma_{b,j}$, which violates constraint {\bf 3}). The properties $\sigma_{a,j} \not= \sigma_{b,j}$ for $a \not= b$ and $\sigma_{a,j+1} \not= \sigma_{b,j}$ for $a < b$ are precisely the same as those which characterize non-attacking fillings $(\mu;\sigma)$; compare with Definition \ref{def:non-attack}.

We denote the map described above, taking $\Xi(\mu)$ to 
$\mathfrak{S}(\mu)$, by $\mathfrak{M}$. It is straightforward to show that $\mathfrak{M}$ is not only injective but also surjective, \ie\ it is a bijection between the two sets. An explicit example of the map $\mathfrak{M}$ is given in Figure \ref{fig:map}.
\begin{figure}
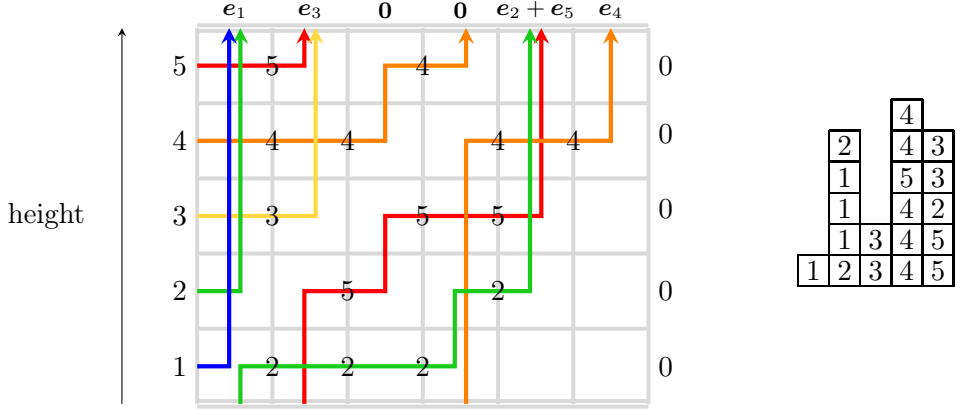

\begin{align*}
\tikz{1}{
\foreach\y in {0}{
\draw[lgray,line width=1.5pt,double] (1.5,0.5+\y) -- (7.5,0.5+\y);
}
\foreach\y in {5}{
\draw[lgray,line width=1.5pt,double] (1.5,0.5+\y) -- (7.5,0.5+\y);
}
\foreach\y in {1,...,4}{
\draw[lgray,line width=1.5pt] (1.5,0.5+\y) -- (7.5,0.5+\y);
}
\foreach\x in {0,...,6}{
\draw[lgray,line width=1.5pt] (1.5+\x,0.5) -- (1.5+\x,5.5);
}
\node[above] at (2,5.5) {\footnotesize$\bm{e}_1$};
\node[above] at (3,5.5) {\footnotesize$\bm{e}_3$};
\node[above] at (4,5.5) {\footnotesize$\bm{0}$};
\node[above] at (5,5.5) {\footnotesize$\bm{0}$};
\node[above] at (6,5.5) {\footnotesize$\bm{e}_2+\bm{e}_5$};
\node[above] at (7,5.5) {\footnotesize$\bm{e}_4$};
\node[right] at (7.5,1) {$0$};
\node[right] at (7.5,2) {$0$};
\node[right] at (7.5,3.1) {$0$};
\node[right] at (7.5,4.1) {$0$};
\node[right] at (7.5,5) {$0$};
\node[left] at (1.5,1) {$1$};
\node[left] at (1.5,2) {$2$};
\node[left] at (1.5,3) {$3$};
\node[left] at (1.5,4) {$4$};
\node[left] at (1.5,5) {$5$};
\draw[ultra thick,red,->] (1.5,5) -- (2.925,5) -- (2.925,5.5);
\draw[ultra thick,red,->] (2.925,0.5) -- (2.925,2) -- (4,2) -- (4,3) -- (6.075,3) -- (6.075,5.5);
\draw[ultra thick,orange,->] (1.5,4) -- (4,4) -- (4,5) -- (5.075,5) -- (5.075,5.5);
\draw[ultra thick,orange,->] (5.075,0.5) -- (5.075,4) -- (7,4) -- (7,5.5);
\draw[ultra thick,yellow,->] (1.5,3) -- (3.075,3) -- (3.075,5.5);
\draw[ultra thick,green,->] (1.5,2) -- (2.075,2) -- (2.075,5.5); 
\draw[ultra thick,green,->] (2.075,0.5) -- (2.075,1) -- (4.925,1) -- (4.925,2) -- (5.925,2) -- (5.925,5.5);
\draw[ultra thick,blue,->] (1.5,1) -- (1.925,1) -- (1.925,5.5);
\node at (2.5,1) {$2$};
\node at (3.5,1) {$2$};
\node at (4.5,1) {$2$};
\node at (5.5,2) {$2$};
\node at (2.5,3) {$3$};
\node at (2.5,4) {$4$};
\node at (3.5,4) {$4$};
\node at (4.5,5) {$4$};
\node at (5.5,4) {$4$};
\node at (6.5,4) {$4$};
\node at (2.5,5) {$5$};
\node at (3.5,2) {$5$};
\node at (4.5,3) {$5$};
\node at (5.5,3) {$5$};
\node at (-0.5,3) {height};
\draw[->] (0.5,0.5) -- (0.5,5.5);
}
\qquad\qquad
\tableau{
& & & {4} &
\\
& {2} & & {4} & {3}
\\
& {1} & & {5} & {3}
\\
& {1} & & {4} & {2}
\\ 
& {1} & {3} & {4} & {5}
\\ 
{1} & {2} & {3} & {4} & {5} 
}
\end{align*}
\caption{Left panel: a $\mu$-legal configuration $\xi$ in the case $\mu = (0,4,1,5,4)$. The trajectory of each coloured path, from the left boundary to its final column, is shown on the picture. Right panel: the corresponding non-attacking filling $\sigma = \mathfrak{M}(\xi)$. The numbers in column 
$1\leq i \leq 5$ are obtained by listing the height coordinates visited by path $i$ at each of its $\mu_i+1$ horizontal unit steps within $\xi$.}
\label{fig:map}
\end{figure}

\subsection{Matching weights}
\label{ssec:weight-match}

In this subsection we let $w : \Xi(\mu) \rightarrow \mathbb{C}$ and $W : \mathfrak{S}(\mu) \rightarrow \mathbb{C}$ be {\it weights}, \ie\ functions which map lattice configurations (resp. fillings) to the complex numbers. Given our existing weights $w$ prescribed to lattice configurations, our aim is to write down $W = w \circ \mathfrak{M}^{-1}$ explicitly, where $\mathfrak{M}^{-1}$ is the inverse of the map described in the previous subsection.

Let us begin by writing down the following generic expansion:
\begin{align*}
f_{\mu}(x_1,\dots,x_n;q,t)
=
\sum_{\xi \in \Xi(\mu)}
w(\xi; x_1,\dots,x_n; q,t),
\end{align*}
obtained by writing $\Omega_{\mu}(q,t) \langle \C_1(x_1) \cdots \C_n(x_n) \rangle_{\mu}$ as a sum over $\mu$-legal lattice configurations $\xi$. Let $\xi$ be such a configuration, of the form \eqref{eta}. For each fixed $0 \leq j \leq N$, the vector $\left(k_1^{(j)},\dots,k_n^{(j)}\right)$ tells us the occupation data of the $j$-th column of vertical edges of the lattice in Figure \ref{fig:xi}. We associate to it a vector of coordinates $\mathcal{A}_j$, in the same way as in Definition \ref{defn:coord}; namely,
\begin{align}
\label{A_j}
\mathcal{A}_j
=
(a_{p,j})_{p \in \mathcal{Q}_{j-1}},
\qquad
\text{such that}
\quad
k^{(j)}_{a_{p,j}} = p,
\quad
\forall\ p \in \mathcal{Q}_{j-1},
\end{align}
where we have defined the following sets, which depend implicitly on the composition $\mu$:
\begin{align}
\mathcal{P}_j = \{p : \mu_p = j\},
\qquad
\mathcal{Q}_j = \{p : \mu_p > j\},
\qquad
\mathcal{P}_j \cup \mathcal{Q}_j = \mathcal{Q}_{j-1}.
\end{align}
In words, $a_{p,j}$ is the (vertical) coordinate of colour $p$ as it passes through the $j$-th column of the lattice. The weight prescribed to $\xi$ is then the product of weights of the individual columns, as given by \eqref{components}:
\begin{multline*}
w(\xi; x_1,\dots,x_n; q,t)
=
\Omega_{\mu}(q,t)
\cdot
\prod_{j=0}^{N}
\prod_{\subalign{
\quad p &> \ell
\\
\quad p &\in \mathcal{Q}_{j-1}
\\
\quad \ell &\in \mathcal{Q}_j
}}
\bm{1}(a_{p,j} \not= a_{\ell,j+1})
\prod_{p \in \mathcal{P}_j} t^{g_j(p)}
\prod_{p \in \mathcal{Q}_j} x_{a_{p,j+1}}
\\
\times
\prod_{p \in \mathcal{Q}_{j-1}}
\frac{1}{1-v_{p,j} t^{f_j(p)}}
\prod_{\substack{
p\in \mathcal{Q}_j \\ a_{p,j} \not= a_{p,j+1}
}}
\frac{ v_{p,j}^{\bm{1}(a_{p,j} > a_{p,j+1})} t^{h_j(p)}(1-t)}
{1-v_{p,j} t^{f_j(p)+1}},
\end{multline*}
where we have defined column-dependent analogues of the exponents \eqref{exponents} in the obvious way;
\begin{align}
\label{exponents-j}
\begin{split}
f_j(p) 
&= 
\# \{\ell \in \mathcal{Q}_j : \ell < p\},
\\
g_j(p)
&=
\# \{\ell \in \mathcal{Q}_j : \ell < p, a_{p,j} < a_{\ell,j+1}\},
\\
h_j(p)
&=
\# \{\ell \in \mathcal{Q}_j : \ell < p, a_{\ell,j+1} \in (a_{p,j},a_{p,j+1})\}.
\end{split}
\end{align}
Let $\mathfrak{M} : \Xi(\mu) \rightarrow \mathfrak{S}(\mu)$ be the previously described bijection between lattice configurations and fillings, and fix a lattice configuration $\xi \in \Xi(\mu)$ of the form \eqref{eta} with associated coordinate vectors given by \eqref{A_j}. If $\sigma \in \mathfrak{S}(\mu)$ is the unique non-attacking filling such that $\sigma = \mathfrak{M}(\xi)$, we note that $\sigma_{i,j} = a_{i,j}$ for all $1 \leq i \leq n$, $0 \leq j \leq \mu_i$. Because of this, we see that
\begin{align*}
f_{\mu}(x_1,\dots,x_n;q,t)
=
\sum_{\sigma \in \mathfrak{S}(\mu)}
W(\sigma;x_1,\dots,x_n;q,t),
\end{align*}
with weights given by
\begin{multline}
\label{W-complicated}
W(\sigma;x_1,\dots,x_n;q,t)
=
\Omega_{\mu}(q,t)
\cdot
\prod_{j=0}^{N}
\prod_{\subalign{
\quad p &> \ell
\\
\quad p &\in \mathcal{Q}_{j-1}
\\
\quad \ell &\in \mathcal{Q}_j
}}
\bm{1}(\sigma_{p,j} \not= \sigma_{\ell,j+1})
\prod_{p \in \mathcal{P}_j} t^{g_j(p)}
\prod_{p \in \mathcal{Q}_j} x_{\sigma_{p,j+1}}
\\
\times
\prod_{p \in \mathcal{Q}_{j-1}}
\frac{1}{1-v_{p,j} t^{f_j(p)}}
\prod_{\substack{
p\in \mathcal{Q}_j \\ \sigma_{p,j} \not= \sigma_{p,j+1}
}}
\frac{ v_{p,j}^{\bm{1}(\sigma_{p,j} > \sigma_{p,j+1})} t^{h_j(p)}(1-t)}
{1-v_{p,j} t^{f_j(p)+1}},
\end{multline}
where each $v_{p,j}$ parameter appearing above is given by \eqref{v_ij}--\eqref{gamma}, and where the exponents $f_j(p)$, $g_j(p)$ and $h_j(p)$ are the same as in \eqref{exponents-j}, up to replacing the letter $a$ with $\sigma$. It remains to show that the right-hand side of \eqref{W-complicated} takes the same form as the summand appearing in \eqref{hhl-formula}; \ie\ we seek to show that
\begin{align}
\label{W-simple}
W(\sigma;x_1,\dots,x_n;q,t)
=
x^{\sigma}
t^{\Delta(\sigma)}
\prod_{s \in \mathcal{D}(\sigma)}
\frac{1-t}{1-q^{l(s)+1} t^{a(s)+1}}
\prod_{s \in \mathcal{A}(\sigma)}
\frac{q^{l(s)+1} t^{a(s)} (1-t)}{1-q^{l(s)+1} t^{a(s)+1}}.
\end{align}
The details of this routine matching procedure follow in Appendix \ref{app:match}.

\appendix
\section{Matching equations \eqref{W-complicated} and \eqref{W-simple}}
\label{app:match}

Throughout this appendix we shall make reference to the right-hand side of \eqref{W-complicated}, taking it apart piece-by-piece.

\addtocontents{toc}{\protect\setcounter{tocdepth}{1}}

\subsection{} We can ignore the product of indicator functions 
\begin{align*}
\prod_{j=0}^{N}
\prod_{\subalign{
\quad p &> \ell
\\
\quad p &\in \mathcal{Q}_{j-1}
\\
\quad \ell &\in \mathcal{Q}_j
}}
\bm{1}(\sigma_{p,j} \not= \sigma_{\ell,j+1});
\end{align*}
all such constraints are already fulfilled by virtue of the fact that 
$\sigma$ is non-attacking.

\subsection{} 
\label{app:x-sigma}

One easily sees that there holds
\begin{align*}
\prod_{j=0}^{N}
\prod_{p \in \mathcal{Q}_j} x_{\sigma_{p,j+1}}
=
\prod_{i=1}^{n}
\prod_{j=1}^{\mu_i} 
x_{\sigma_{i,j}}
=
x^{\sigma}.
\end{align*}

\subsection{} Combining equations \eqref{something}--\eqref{coeff-prod3} from earlier, we observe the cancellation
\begin{align*}
\Omega_{\mu}(q,t)
\cdot
\prod_{j=0}^{N}
\prod_{p \in \mathcal{Q}_{j-1}}
\frac{1}{1-v_{p,j} t^{f_j(p)}}
=
1.
\end{align*}

\subsection{} We move on to the factor
\begin{align*}
\prod_{j=0}^{N}
\prod_{\substack{
p\in \mathcal{Q}_j \\ \sigma_{p,j} \not= \sigma_{p,j+1}
}}
\frac{1-t}{1-v_{p,j} t^{f_j(p)+1}}
=
\prod_{j=0}^{N}
\prod_{\substack{
p: \mu_p > j \\ \sigma_{p,j} \not= \sigma_{p,j+1}
}}
\frac{1-t}{1-v_{p,j} t^{f_j(p)+1}}.
\end{align*}
Using the explicit form \eqref{v_ij}--\eqref{gamma} of $v_{p,j}$ and collecting $t$ exponents in the exact same way as in the computations at the end of Section \ref{ssec:norm}, we obtain (\cf\ Definitions \ref{def:leg}, \ref{def:arm} and \ref{def:ascent})
\begin{align*}
\prod_{j=0}^{N}
\prod_{\substack{
p\in \mathcal{Q}_j \\ \sigma_{p,j} \not= \sigma_{p,j+1}
}}
\frac{1-t}{1-v_{p,j} t^{f_j(p)+1}}
=
\prod_{i=1}^{n}
\prod_{\substack{
1 \leq j \leq \mu_i \\ \sigma_{i,j-1} \not= \sigma_{i,j}
}}
\frac{1-t}{1-q^{\mu_i-j+1} t^{\alpha_{i,j-1}(\mu)+1}}
=
\prod_{s \in \mathcal{D}(\sigma) \cup \mathcal{A}(\sigma)}
\frac{1-t}{1-q^{l(s)+1} t^{a(s)+1}}.
\end{align*}

\subsection{} 
\label{app:t^g}
The remaining factors require us to collect together various $t$ exponents. We have 
\begin{align*}
\prod_{j=0}^{N}
\prod_{p \in \mathcal{P}_j} t^{g_j(p)}
=
\prod_{j=0}^{N}
\prod_{p: \mu_p = j} t^{g_j(p)}
=
\prod_{j=0}^{N}
t^{\#\{(\ell,p):\ \ell < p,\ \mu_\ell > \mu_p = j,\ \sigma_{p,j} < 
\sigma_{\ell,j+1} \}}.
\end{align*}
This product assigns a factor of $t$ to every pair of boxes in the filling $\sigma$ of the form
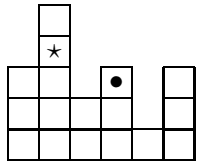
\begin{align*}
\tableau{
& {} & & & &
\\
& {\star} & & &
\\
{} & {} & & {\bullet} & & {}
\\ 
{} & {} & {} & {} & & {}
\\ 
{} & {} & {} & {} & {} &{}
}
\end{align*}
with square coordinates $\star = (\ell,j+1)$ and $\bullet = (p,j) = (p,\mu_p)$, and such that the square fillings satisfy $\sigma_{\ell,j+1} > \sigma_{p,j}$. The number of such pairs is equal to the number of positive ordered triples of the form \eqref{pos-inv} with $\sigma_{i',j} = \infty$.

\subsection{} The last factor is the most involved. It has the form
\begin{align}
\label{two-pieces}
\prod_{j=0}^{N}
\prod_{\substack{
p\in \mathcal{Q}_j \\ \sigma_{p,j} \not= \sigma_{p,j+1}
}}
v_{p,j}^{\bm{1}(\sigma_{p,j} > \sigma_{p,j+1})} t^{h_j(p)}
=
\prod_{j=0}^{N}
\left(
\prod_{\substack{
p\in \mathcal{Q}_j \\ \sigma_{p,j} < \sigma_{p,j+1}
}}
t^{h_j(p)}
\right)
\left(
\prod_{\substack{
p\in \mathcal{Q}_j \\ \sigma_{p,j} > \sigma_{p,j+1}
}}
v_{p,j} t^{h_j(p)}
\right),
\end{align}
which splits into two pieces. 

\subsection{} The first piece of \eqref{two-pieces} is given by
\begin{align*}
\prod_{j=0}^{N}
\prod_{\substack{
p\in \mathcal{Q}_j \\ \sigma_{p,j} < \sigma_{p,j+1}
}}
t^{h_j(p)}
=
\prod_{j=0}^{N}
\prod_{\substack{
p: \mu_p > j \\ \sigma_{p,j} < \sigma_{p,j+1}
}}
t^{h_j(p)}
=
\prod_{j=0}^{N}
t^{\#\{(\ell,p):\ \ell < p,\ \mu_\ell > j,\ \mu_p > j,\ \sigma_{p,j} < 
\sigma_{\ell,j+1} < \sigma_{p,j+1} \}},
\end{align*}
using the definition \eqref{exponents-j} of $h_j(p)$. This product assigns a factor of $t$ to every triple of boxes in the filling $\sigma$ of the form
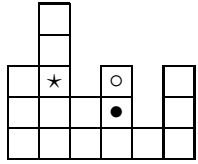
\begin{align*}
\tableau{
& {} & & & &
\\
& {} & & &
\\
{} & {\star} & & {\circ} & & {}
\\ 
{} & {} & {} & {\bullet} & & {}
\\ 
{} & {} & {} & {} & {} &{}
}
\end{align*}
with square coordinates $\circ = (p,j+1)$, $\star = (\ell,j+1)$ and 
$\bullet = (p,j)$, and such that the square fillings satisfy 
$\sigma_{p,j+1} > \sigma_{\ell,j+1} > \sigma_{p,j}$. Hence we are, in this case, counting positive ordered triples of the form \eqref{pos-inv}. Combining this with the contribution from Section \ref{app:t^g}, we recover precisely $t^{{\rm ord}_{+}(\sigma)}$.

\subsection{} 
\label{app:second-piece}

The second piece of \eqref{two-pieces} is given by
\begin{align}
\label{second-piece}
\prod_{j=0}^{N}
\prod_{\substack{
p\in \mathcal{Q}_j \\ \sigma_{p,j} > \sigma_{p,j+1}
}}
v_{p,j} t^{h_j(p)}
=
\prod_{j=0}^{N}
\prod_{\substack{
p: \mu_p > j \\ \sigma_{p,j} > \sigma_{p,j+1}
}}
v_{p,j} t^{h_j(p)}.
\end{align}
We note that
\begin{align*}
h_j(p)
&=
f_j(p) 
- 
\#\{\ell \in \mathcal{Q}_j : 
\ell < p,\
\sigma_{\ell,j+1} 
\not\in (\sigma_{p,j},\sigma_{p,j+1}) \}
\\
&=
f_j(p) 
- 
\#\{\ell \in \mathcal{Q}_j : 
\ell < p,\ 
\sigma_{p,j} > \sigma_{\ell,j+1} > \sigma_{p,j+1} \},
\end{align*}
where the second line follows from the fact that $\sigma_{p,j} > \sigma_{p,j+1}$ by assumption in \eqref{second-piece}, as well as the fact that the equalities $\sigma_{p,j} = \sigma_{\ell,j+1}$ and $ \sigma_{p,j+1} = \sigma_{\ell,j+1}$ are forbidden (they both violate the non-attacking property of fillings). Equation \eqref{second-piece} becomes (again, using the computation at the end of Section \ref{ssec:norm}) 
\begin{multline*}
\prod_{j=0}^{N}
\prod_{\substack{
p\in \mathcal{Q}_j \\ \sigma_{p,j} > \sigma_{p,j+1}
}}
v_{p,j} t^{h_j(p)}
\\
=
\left(
\prod_{i=1}^{n}
\prod_{\substack{
1 \leq j \leq \mu_i \\ \sigma_{i,j-1} > \sigma_{i,j}
}}
q^{\mu_i-j+1} t^{\alpha_{i,j-1}(\mu)}
\right)
\prod_{j=0}^{N}
t^{-\#\{(\ell,p):\ \ell<p,\ \mu_{\ell}>j,\ \mu_p>j,\
\sigma_{p,j} > \sigma_{\ell,j+1} > \sigma_{p,j+1}
\}}.
\end{multline*}
The final product assigns a factor of $t^{-1}$ to every triple of boxes in the filling $\sigma$ of the form
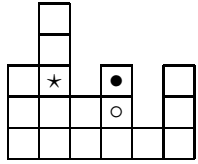
\begin{align*}
\tableau{
& {} & & & &
\\
& {} & & &
\\
{} & {\star} & & {\bullet} & & {}
\\ 
{} & {} & {} & {\circ} & & {}
\\ 
{} & {} & {} & {} & {} &{}
}
\end{align*}
with square coordinates $\bullet = (p,j+1)$, $\star = (\ell,j+1)$ and 
$\circ = (p,j)$, and such that the square fillings satisfy 
$\sigma_{p,j+1} < \sigma_{\ell,j+1} < \sigma_{p,j}$. We are thus counting negative ordered triples of the form \eqref{neg-inv}. We conclude that
\begin{align*}
\prod_{j=0}^{N}
\prod_{\substack{
p\in \mathcal{Q}_j \\ \sigma_{p,j} > \sigma_{p,j+1}
}}
v_{p,j} t^{h_j(p)}
=
t^{-{\rm ord}_{-}(\sigma)}
\prod_{s \in \mathcal{A}(\sigma)}
q^{l(s)+1} t^{a(s)}.
\end{align*}

\subsection{} Putting together the results of Sections \ref{app:x-sigma}--\ref{app:second-piece}, we find the desired expression, \eqref{W-simple}.

\bibliographystyle{alpha}
\bibliography{references}

\end{document}